\documentclass{article}
\usepackage[margin=0.8in]{geometry}

\usepackage{graphicx, color, graphpap}
\usepackage{float}
\usepackage[shortlabels]{enumitem}
\usepackage{amssymb}
\usepackage{amsthm}
\usepackage{multirow}
\usepackage[dvipsnames]{xcolor}
\usepackage[T1]{fontenc}
\usepackage{algorithm}
\usepackage{algpseudocode}
\usepackage{bbm}
\usepackage{thmtools,thm-restate}
\usepackage{verbatim}
\usepackage{mathtools}
\usepackage{titlesec}
\usepackage{amsmath}
\usepackage{setspace}
\usepackage{tikz}
\usetikzlibrary{quantikz}

\usepackage{accents}

\usepackage{authblk}

\definecolor{links}{RGB}{32,98,172}
\definecolor{cites}{RGB}{32,148,90}
\usepackage[colorlinks=true,citecolor=cites,linkcolor=cyan!70!black, urlcolor=cyan!80!black]{hyperref}

\usepackage{doi}
\usepackage[nottoc]{tocbibind}

\usepackage[backend=bibtex, style=numeric-comp, sorting=none, backref=true, maxbibnames=8, firstinits=true]{biblatex}
\DefineBibliographyStrings{english}{
  backrefpage={page },
  backrefpages={pages }
}

\usepackage{array}
\newcolumntype{P}[1]{>{\centering\arraybackslash}p{#1}}
\makeatletter
\newcommand{\thickhline}{%
    \noalign {\ifnum 0=`}\fi \hrule height 1.2  pt
    \futurelet \reserved@a \@xhline
} 
\makeatother 
\usepackage{multirow,multicol}
\usepackage[most]{tcolorbox}

\usepackage{diagbox}

\long\def\ca#1\cb{} 

\renewcommand{\braket}[2]{\langle #1 \hspace{1pt} | \hspace{1pt} #2 \rangle}
\newcommand{\kt}[1]{| #1 \rangle}
\newcommand{\br}[1]{\langle #1 |}
\newcommand{\kkt}[1]{\big| #1 \big\rangle}
\newcommand{\bbr}[1]{\big\langle #1 \big|}

\newcommand{\ketbra}[2]{| \hspace{1pt} #1 \rangle \langle #2 \hspace{1pt} |}

\newcommand{\dya}[1]{\ket{#1}\!\!\bra{#1}}
\newcommand{\dyad}[2]{\ket{#1}\!\bra{#2}}        

\newcommand{\poly}{\operatorname{poly}}

\newcommand{\BC}{\mathcal{B}}
\newcommand{\CC}{\mathcal{C}}

\newcommand{\HC}{\mathcal{H}}

\newcommand{\OC}{\mathcal{O}}

\newcommand{\ZC}{\mathcal{Z}}

\newcommand{\Tr}{{\rm Tr}}

\renewcommand{\geq}{\geqslant}
\renewcommand{\leq}{\leqslant}

\newcommand*{\id}{\mathbbm{1}}

\DeclarePairedDelimiter{\ceil}{\lceil}{\rceil}

{}
{}
\newtheorem{theorem}{Theorem}
\newtheorem{lemma}{Lemma}
\newtheorem*{remark}{Remark}
\newtheorem{corollary}{Corollary}
\newtheorem{proposition}{Proposition}

\newtheorem{definition}{Definition}
\newtheorem{supprop}{Supplementary Proposition}

\usepackage[usestackEOL]{stackengine}[2013-10-15]
\stackMath

\usepackage{csquotes}
\MakeOuterQuote{"}

\makeatother

\begin{document}

\title{Qubit-Efficient Randomized Quantum Algorithms for Linear Algebra}




\author[1]{Samson Wang \thanks{samson.wang@outlook.com}}
\author[2,3]{Sam McArdle}
\author[4,5]{Mario Berta}
\affil[1]{Department of Physics, Imperial College London, London, UK}
\affil[2]{AWS Center for Quantum Computing, Pasadena, USA}
\affil[3]{California Institute of Technology, Pasadena, USA}
\affil[4]{Department of Computing, Imperial College London, London, UK}
\affil[5]{Institute for Quantum Information, RWTH Aachen University, Aachen, Germany}

\date{}
\maketitle

\begin{abstract}
We propose a class of randomized quantum algorithms for the task of sampling from matrix functions, without the use of quantum block encodings or any other coherent oracle access to the matrix elements.  As such, our use of qubits is purely algorithmic, and no additional qubits are required for quantum data structures. {Our algorithms start from a classical data structure in which the matrix of interest is specified in the Pauli basis.} For $N\times N$ Hermitian matrices, the space cost is $\log(N)+1$ qubits and depending on the structure of the matrices, the gate complexity can be comparable to state-of-the-art methods that use quantum data structures of up to size $\OC(N^2)$, when considering equivalent end-to-end problems.  Within our framework, we present a quantum linear system solver that allows one to sample properties of the solution vector, as well as algorithms for sampling properties of ground states and Gibbs states of Hamiltonians. As a concrete application, we combine these sub-routines to present a scheme for calculating Green's functions of quantum many-body systems.
\end{abstract} 


\section{Introduction}\label{sec:intro}

\subsection{Overview}

As improvements in hardware increase the number and quality of qubits, we seek quantum algorithms that are able to showcase practical quantum advantage in the earliest possible time-frame. Looking beyond NISQ technologies \cite{preskill2018quantum, bharti2021noisy, cerezo2020variationalreview}, it is reasonable to assume that, given continued progress in quantum hardware, so-called fault-tolerant algorithms will have an important place in the gamut of quantum computing applications. Thus, it is pertinent to ask how soon such algorithms can be useful for real-life applications, and how much can we accelerate this timeline by constructing algorithms with lower and more flexible quantum resource costs. 

Quantum algorithms for manipulating matrices have been proposed for many problems, including factoring, linear systems, ground state energy estimation, simulation, and beyond (e.g.~see Refs.~\cite{montanaro2016review, deWolf2019quantum, bauer2020quantum, lin2022lecture, dalzell2023quantum} and references therein). These algorithms are often phrased in terms of a quantum oracle model from which elements of the matrix of interest can be coherently accessed. Since the seminal proposals, there have been extensive improvements and refinements to the asymptotic runtime for each of these algorithms, which is usually measured in the number of required queries to the oracle. For many problems, the state-of-the-art query complexities are optimal or close to optimal according to known complexity lower bounds. Moreover, many of these recent techniques are unified under the so-called quantum singular value transformation (QSVT) framework \cite{gilyen2019quantum, martyn2021grand}, in which polynomial approximations are applied to the singular values of the desired matrix. Here, the oracle embeds the matrix in a larger unitary, commonly known as a block encoding.

Despite this promise, when considering end-to-end implementations of such quantum algorithms two major hurdles can arise.\footnote{By "end-to-end" we mean that all relevant resources are accounted for in an implementation that solves the required task in full.} First, the implementation of the quantum oracles can require costly additional quantum resources, both in the depth required for each call and the number of qubits consumed. Second, the quantum algorithm can come with certain conditions or caveats that need to be satisfied for efficient applications \cite{dalzell2023quantum}.

As an example, consider the linear systems problem, for which the pioneering HHL algorithm was proposed \cite{harrow2009quantum}. The current state-of-the-art quantum linear systems solver (QLSS) presented in Ref.~\cite{costa2021optimal} uses $\OC(\log(N))$ algorithmic qubits and a block encoding oracle to process $N\times N$ matrices. This algorithm only needs to make a small number of calls to the block encoding if the matrix is well-conditioned, compared to the runtime of classical linear systems algorithms. However, for general matrices an implementation of the block encoding using QRAM \footnote{Quantum Random Access Memory (QRAM): a widely-studied data access structure that allows for coherent access to classical data \cite{giovannetti2008quantum, di2020fault, hann2021resilience}} in depth $\OC(\log(N))$ requires $\OC(N^2)$ qubits \cite{hann2020hardware,giovannetti2008quantum,clader2022quantum}.\footnote{The minimum qubit count implementation in Ref.~\cite{clader2022quantum} requires $\OC(N)$ qubits and $\OC(N)$ depth per call to the block encoding. Alternatively, a QROM (Quantum Read-Only Memory) can be implemented using $\OC(\log N)$ qubits, but $\OC(N^2)$ gates.} Thus, for general matrices the exponential savings in space resources are nullified. In order to circumvent this burden, {one should search for specific classes of matrices with structure for which access is less costly}. {For instance}, for matrices that are $L$-sparse in an efficiently implementable unitary basis, block encodings can be implemented {with substantially less quantum cost} than in the general case \cite{childs2018toward, babbush2018encoding, berry2019qubitization,wan2021exponentially}. 
{However, this still leads to an additional qubit overhead} that we argue could be minimized further when considering early implementations of fault-tolerant algorithms. {For matrices sparse in the computational basis, up to ${\OC}(N)$ qubit overhead is still required \cite{di2020fault}, unless one seeks additional structure such that the matrix entries can be efficiently coherently computed.} {Second}, it is important to consider the exact problem that the quantum algorithm solves: the QLSS returns a quantum state in which the solution vector is encoded with some non-zero additive error, unlike textbook classical solvers that provide full classical vector exactly. Thus, in order to assess the utility of quantum linear systems solvers, full end-to-end applications including possible additional subroutines need to be carefully analyzed \cite{aaronson2015read}. For instance, it may be more fair to compare the quantum algorithm to randomized classical solvers which allow for some error \cite{strohmer2009randomized}, or so-called "dequantized" approaches which operate with a classical data structure analogous to QRAM \cite{gilyen2022improved, shao2022faster}. We discuss various approaches to linear systems in more detail in Section \ref{sec:linear-systems-comparison}.

\begin{figure*}[t]
    \centering
    \includegraphics[width=0.98\columnwidth]{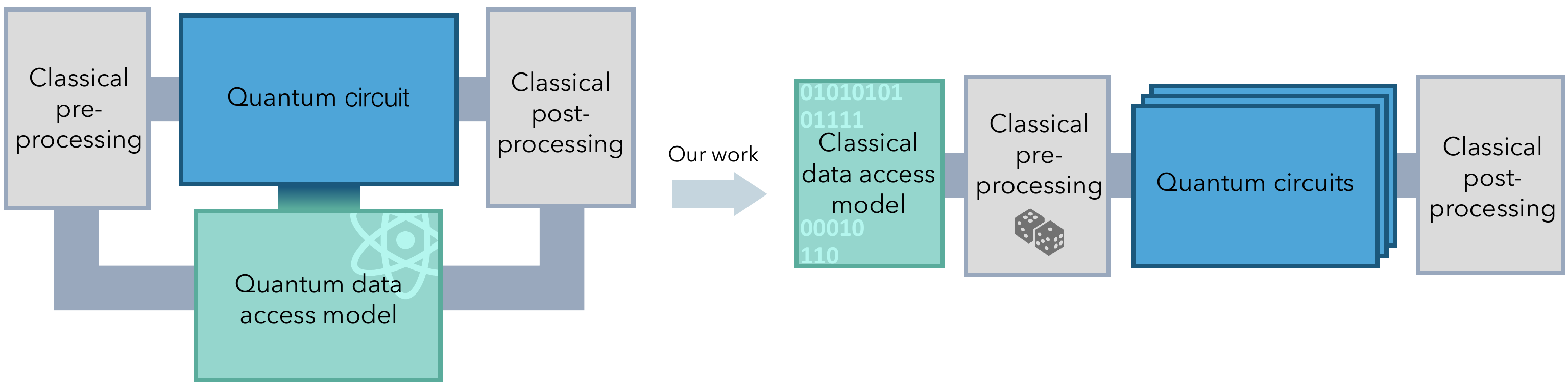}
    \caption{\textbf{Motivation of our work.} We reduce quantum hardware requirements for quantum algorithms on classical data by removing the need for quantum data structures or quantum oracles. This is achieved by replacing coherent access to the data with a classical description of the data in the Pauli basis, and utilizing a randomized algorithm that samples the outputs of many quantum circuits. These circuits are chosen independently and thus in theory can also be parallelized, trading reduced total runtime for additional space cost in the form of many quantum processors. Our approach uses circuits with at most $\log N + 1$ qubits when processing data from $N \times N$ Hermitian matrices. This can be compared to other algorithms that utilize quantum data access models which may have significantly greater qubit overhead overall.  
    }
    \label{fig:summary}
\end{figure*}

{\textit{Our contribution.}} In this work, we present a framework for constructing algorithms that sample properties of matrix functions which do not use quantum oracles to provide coherent access to the matrix in question (see Figure \ref{fig:summary}). Despite having no qubit overhead to implement quantum oracles, the asymptotic complexities of our algorithms can remain comparable with other algorithms in the literature whenever the considered matrices have an amenable structure in the Pauli basis (and when considering equivalent end-to-end problems). Hence, for physically motivated matrices, potential quantum advantages originally requiring QRAM could possibly be similarly obtained in our approach without using a quantum data structure, making them more applicable for the early fault-tolerant regime \cite{campbell2021early, lin2022heisenberg, wan2021randomized, dong2022ground, wang2022quantum, zhang2022computing}. Specifically, given a Fourier series approximation to a function $f$, and an $N\times N$ Hermitian matrix $A$ with known decomposition in the Pauli basis, we give algorithms to sample properties of $f(A)$ using a total of $\log(N)+1$ qubits. {These properties take the form $\Tr[f(A)\rho f(A)^{\dag} O]$ and $\bra{\psi}f(A) \ket{\phi}$ for some quantum states $\rho, \dya{\psi}, \dya{\phi}$, and some measurement observable $O$}. Using this framework we present algorithms for sampling properties of the solution vector in the linear systems problem, as well as from ground states and Gibbs states of a given Hamiltonian. We provide direct comparisons of the complexities of our algorithms with other classical and quantum algorithms for specific end-to-end tasks, and we present an application of our algorithms for computing Green's functions in many-body physics.

As our starting point we take inspiration from algorithms for quantum chemistry \cite{campbell2019random,lin2022heisenberg,wan2021randomized}, where often quantum data structures are not needed. {Instead of} quantum (coherent) access to the matrix $A$ {we ask for} classical access to the coefficients $a_{\ell} \in \mathbb{R}$ in its decomposition in the Pauli basis
\begin{equation}\label{eq:pauli-decomposition}
    A=\sum_{\ell=1}^L a_{\ell} P_{\ell}\,,
\end{equation}
where $P_{\ell}$ are multi-qubit Pauli operators. We refer to this as the Pauli access model and note that it is a natural representation for data coming from physical problems, for instance, when the matrix comes from a Hamiltonian. The model also mathematically matches the physical intuition that running quantum sub-routines for ``quantumly structured data'' has possible potential for quantum speed-ups. Indeed, our algorithms are faster for matrices with small vector $\ell_1$-norm of Pauli coefficients $\lambda := \sum |a_{\ell}|$, which we refer to as the "Pauli weight". Moreover, we remark that there is no explicit dependence on the sparsity, or the number of Pauli terms $L$ (which we will call the ``Pauli sparsity'') in the {quantum} runtimes -- both of which can be substantially larger than the Pauli weight, and appear in the runtime of other algorithms.


\subsection{Related work}

Applications for the early fault-tolerant era of quantum computing have recently begun to be explored, following the motivation to design algorithms that extract practical value out of fault-tolerant quantum algorithms as soon as possible \cite{campbell2019random,lin2022heisenberg,wan2021randomized,faehrmann2022randomizing,wang2022quantum, wang2023faster, dong2022ground, zhang2022computing}. In this spirit, algorithms have been designed to consume fewer quantum resources for Hamiltonian problems including phase estimation \cite{lin2022heisenberg, wan2021randomized, wang2022quantum, wang2023faster, dong2022ground}, ground state preparation \cite{dong2022ground}, and computing ground state properties \cite{zhang2022computing}, by increasing the number of circuit samples required. Until now, these algorithms have predominantly aimed to reduce a proxy for the maximum circuit depth, in the form of the number of calls to a time evolution oracle for a prescribed Hamiltonian in one coherent run of a circuit. It then remains to choose an appropriate time evolution oracle for the intended setting, which can substantially affect the gate overhead or the number of qubits required. {This in contrast to our approach, where the first priority is to reduce qubit overhead}. In Section \ref{sec:gs-comparison} we discuss further the implications for various choices of time evolution oracle, and how the resulting implementations compare to our results for the ground state property estimation problem. A key tool in the aforementioned algorithms is the use of \textit{randomization}. Randomized approaches have also more generally found use in Hamiltonian simulation \cite{childs2019faster, campbell2019random} and in simplifying quantum walk algorithms \cite{apers2021unified}.

One distinctive approach is that of Ref.~\cite{wan2021randomized} for phase estimation which uses the Pauli access model, rather than a time evolution oracle. In this work an algorithm is proposed which randomly compiles the Heaviside function {$H(A)$ via a quantity of the form $\Tr[H(A)\rho]$} by sampling from a Fourier series approximation to the function. Ref.~\cite{faehrmann2022randomizing} also presents an algorithm to perform randomized sampling of a given observable after a time evolution of a given Hamiltonian. We extend these ideas in our work to {a more general class of properties corresponding to any function for which we have a Fourier approximation}. 
{The result of this approach is an overhead of only one additional qubit is required to run the algorithm for Hermitan matrices. Moreover, as with the approach in Ref.~\cite{wan2021randomized}, we sample from the outputs of many quantum circuits, rather than running one long coherent evolution.}

We remark that near-term approaches for the quantum linear systems problem have recently been proposed which use similar data access assumptions to our Pauli access model \cite{bravo2019variational, xu2019variational, huang2019near}. Namely, these works assume the matrix of interest has a known decomposition $A=\sum_t c_t U_t$, where $U_t$ are efficiently implementable unitaries (such as the Pauli basis decomposition as in our work). These approaches use parameterized circuits whose depth can be tuned to whatever a near-term implementation allows. However, despite showing promising numerical performance for small problem sizes, they lack generic runtime guarantees. On the other hand, our algorithms give prescriptive circuits with runtime guarantees.


\subsection{Outline}

The rest of the manuscript is structured as follows. In Section \ref{sec:general} we present our general framework including our main result in Subsection \ref{sec:main-result}; a note on the classical power of our access model in Subsection \ref{sec:classical-power}; and a discussion on sampling from other linear combinations of unitaries in Subsection \ref{sec:other-LCUs}. We then demonstrate applications of our framework: we present our algorithm for sampling properties of the solution vector in the linear systems problem in Section \ref{sec:linear-systems}; we describe our algorithm for sampling properties of the ground state in Section \ref{sec:gs}; we discuss our algorithm for sampling properties of Gibbs states in Section \ref{sec:gibbs}; and we show how these algorithms can be used together to estimate single-particle Green's functions in the context of many-body physics in Section \ref{sec:greens}.  Finally, in Section \ref{sec:discussion} we present our concluding discussions and outlook. In the Appendices we present detailed analytical statements and all proofs thereof. 



\section{General approach}\label{sec:general}

\subsection{Warm-up problem}\label{sec:fourier-sampling}

We start by demonstrating how to extend the ideas from Refs.~\cite{wan2021randomized} and \cite{faehrmann2022randomizing} to sample properties of a Fourier series of a given matrix. This gives intuition for the core routine we will use for the rest of our results. Readers who want a summary of the results may skip ahead to Section \ref{sec:main-result} and the outlines of our specific algorithms in Sections \ref{sec:linear-systems}, \ref{sec:gs}, \ref{sec:gibbs} \textit{\&} \ref{sec:greens}.

Our algorithms will make use of the Hadamard test circuit (see Fig.~\ref{fig:circuits}(a)), pioneered by Lin and Tong in Ref.~\cite{lin2022heisenberg} for use in the ground state energy estimation problem in the early fault-tolerant regime. We will also use the related circuit in Fig.~\ref{fig:circuits}(b), introduced by Childs and Wiebe in Ref.~\cite{childs2012hamiltonian} to implement linear combinations of unitaries for Hamiltonian simulation. Our circuits will ask for elementary controlled unitary operations in the form of controlled Pauli gates and Pauli rotations.

\begin{proposition}[Sampling from Fourier series]\label{prop:fourier-sampling}
Suppose that we have a Fourier series 
\begin{align}\label{eq:fourier-A}
&\quad\quad\quad \text{$s(A) = \sum_{k\in F} \alpha_k \exp{(it_kA)}$}\,,\\ 
&\text{with $\ell_1$-norm of coefficients $\alpha:=\sum_{k\in F} |\alpha_k|$}\,,
\end{align}
for an $N\times N$ Hermitian matrix $A$ with known Pauli decomposition $A=\sum_{\ell} a_{\ell} P_{\ell}$ and Pauli weight $\lambda = \sum_{\ell} |a_{\ell}|$. Then, we find:

\emph{(a)} Given a procedure to prepare the pure states $\kt{\psi}$, $\kt{\phi}$ with respective unitaries $U_{\psi}$, $U_{\phi}$ and respective gate depths $d_{\psi}$, $d_{\phi}$, we have a randomized quantum algorithm that uses $\log(N)+1$ qubits to
    \begin{align}
    \text{approximate $\br{\phi} s(A)\kt{\psi}$ up to additive error $\varepsilon$}\,,
    \end{align}
    with arbitrary constant success probability, using
    \begin{align}\text{${\CC}_{\emph{sample}}^{\phi}= \OC\big( {\alpha^2}/{\varepsilon^2} \big)$ circuit samples}\,,
    \end{align} 
    where each circuit takes the form in Figure \ref{fig:circuits}(a) with
    \begin{align}
    {\CC}_{\emph{gate}}^{\phi} = \OC\big(\lambda^2 t_{max}^2 + d_{\psi} + d_{\phi}\big) \; \text{gate depth}\,,
    \end{align}
    where we denote $t_{max} := \max_{k\in F}t_k$.
    
\emph{(b)} Given a procedure to prepare the quantum state $\rho$ in depth $d_{\rho}$ and perform measurements with measurement operator $O$, we have a randomized quantum algorithm that uses $\log(N)+1$ qubits to 
    \begin{align}
    \text{approximate $\Tr[s(A)\rho s(A) O]$ up to additive error $\varepsilon$}\,,
    \end{align}
    with arbitrary constant success probability, using 
    \begin{align}
    \text{${\CC}_{\emph{sample}}^O =  \OC\big( {\|O\|^2\alpha^4}/{\varepsilon^2} \big)$ circuit samples}\,,
    \end{align}
    where each circuit takes the form in Figure \ref{fig:circuits}(b) with
    \begin{align} {\CC}_{\emph{gate}}^O = \OC\big(\lambda^2 t_{max}^2 + d_{\rho}\big) \; \text{gate depth}.
    \end{align}
\end{proposition}

We see that various properties of the Fourier series determine the complexity of the quantum algorithm. Namely, the weight of the Fourier series $\alpha$ determines the sample complexity, whilst the time parameter $t_{\max}$ determines the gate complexity. The gate complexity also depends on the Pauli weight $\lambda$ of the matrix $A$. In certain cases, the Pauli weight of a matrix can be much smaller than its dimension, despite there being many non-zero Pauli terms. In these cases, we expect the above algorithms to be efficient. Note that, similar to the linear combinations of unitaries (LCU) \cite{childs2012hamiltonian, berry2014exponential, berry2015hamiltonian, berry2015simulating} and QSVT \cite{gilyen2019quantum, martyn2021grand} frameworks for quantum algorithms, our framework enacts a general class of functions, which we will apply to different approximation problems in the rest of this manuscript. Unlike the aforementioned approaches, our framework does not need access to quantum oracles, or any additional coherent resources.

\begin{figure}[t]
  \centering
    \includegraphics[width=0.98\columnwidth]{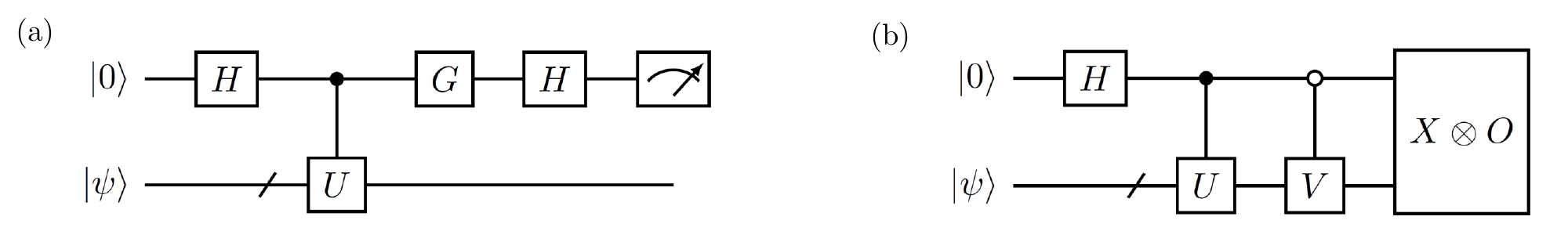}
  \caption{\textbf{Circuits used in our algorithms.} We sample strings of quantum gates, consisting of specified Pauli operators and Pauli rotations, and perform quantum circuit runs with controlled versions of these gates. (a) Hadamard test circuit{.} Measuring the expectation value of $Z$ on the first register returns \textrm{Re}($\br{\psi}U\kt{\psi}$) and \textrm{Im}($\br{\psi}U\kt{\psi}$ for choices of $G = \id$ and $G = S^{\dag}:=\ketbra{0}{0}-i\ketbra{1}{1}$ respectively.
  (b) Generalized Hadamard test circuit. {Applying controlled-$U$ and anticontrolled-$V$, followed by measurement of the observable} $X \otimes O$ yields $\frac{1}{2}\left(\br{\psi}U^{\dag}OV\kt{\psi}+ \br{\psi}V^{\dag}OU\kt{\psi} \right)$. \label{fig:circuits}}
\end{figure}

We provide the proof of Proposition \ref{prop:fourier-sampling} in Appendix \ref{appdx:fourier-sampling-prop}. A key technical tool we use is the \textit{random compiler lemma} of Ref.~\cite[{Lemma 2}]{wan2021randomized} which decomposes fractional time evolution operators into a probabilistic mixture of Pauli matrices and Pauli rotations. This leads to a decomposition of the full time evolution operator simply by taking the product of such fractional operators. Namely, for some time parameter $t\in \mathbb{R}$, for any choice $r\in \mathbb{N}$ and for Hermitian matrix $A$ with known Pauli decomposition as in Eq.~\eqref{eq:pauli-decomposition}, Ref.~\cite{wan2021randomized} shows how one can obtain the decomposition 
\begin{align}\label{eq:string-of-gates}
    e^{iAt/r} &= \gamma\sum_{\ell \in T'}  p_{\ell} \cdot u_{\ell} P_{\ell}\;  \exp\left(i \theta_{\ell} P'_{\ell}\right)\,, 
\end{align}
for some index set $T'$, where $p_{\ell}$ are probabilities, $u_{\ell}\in\{\pm1, \pm i\}$ are phases, and $P_{\ell},P_{\ell}'$ are ($\ceil{\log N}$-qubit) Pauli operators, which all implicitly depend on $t/r$. Moreover, the weight of the mixture can be shown to satisfy $\gamma \leq \exp (\lambda^2 t^2/r^2)$. This implies that one can express a full time evolution operator as a linear combination
\begin{align}\label{eq:kianna-ham-simulation-simp}
    e^{i A t}= \left(e^{iAt/r}\right)^r =\sum_{m \in {T}} \beta_m^{(r)} U_m^{(r)}\,,
\end{align}
where the weight of the coefficients satisfy $\sum_{m \in T} |\beta^{(r)}_m| = \OC(1)$ for  $r = \Theta(\lambda^2 t^2)$, meaning that one can sample from the distribution with bounded variance. Each $U^{(r)}_m$ is a string of gates consisting of $r$ pairs of controlled (multi-qubit) Pauli rotations and {a series of} controlled Pauli gates. Thus, $r$ can be considered to control the quantum runtime of the algorithm. Our algorithms sample from the strings of gates according to the linear combination in Eq.~\eqref{eq:kianna-ham-simulation-simp}. 

{We make two brief remarks on compilation. First, we note that $n$-qubit Pauli rotation gates can be compiled into a single single-qubit Pauli rotation gate and $\OC(n)$ Clifford gates. Thus, in our algorithmic framework compilation of each layer of gates results in a single non-Clifford gate, that is, the number of non-Clifford gates and the total gate depth go hand-in-hand. Second, any series of $n$-qubit Pauli gates can be classically compiled into a single $n$-qubit Pauli gate up to a phase in $\OC(n)$ classical time. From hereon we will generally refer to quantum gate depth whilst stressing that this is equivalent to the number of non-Clifford gates up to a logarithmic factor .}

We provide pseudo-code for the algorithm presented in Proposition \ref{prop:fourier-sampling} to prepare  $\br{\phi} s(A)\kt{\psi}$ in Algorithm \ref{alg:overlap}. The algorithm for $\Tr\left[s(A)\rho s(A)^{\dag} O\right]$ is very similar and we present the pseudo-code for this in full in Algorithm \ref{alg:observable} in Appendix \ref{appdx:fourier-sampling-prop}. These two algorithms will form the core quantum routine for the rest of our results.

We remark that in the pseudocode in Algorithm \ref{alg:overlap} we have left the runtime vector $\vec{r}$ as a freely chosen algorithm parameter. In practice, one would likely choose each element of the runtime vector to be $r_k \propto \lambda^2t_k^2$ for all $k$ which leads to the runtime guarantees specified in Proposition \ref{prop:fourier-sampling}. However, we note that one can still freely choose a proportionality constant, which will trade a constant factor improvement in the sample complexity for increased gate depth, or vice versa.

{Finally, in Appendix \ref{appdx:preprocessing} we detail the classical overheads required to run our core algorithm in Algorithms \ref{alg:overlap} and \ref{alg:observable}. These amount to essentially linear prepossessing overhead in the number of Pauli terms $L$ and the number of Fourier terms $|F|$. Each quantum sample comes with logarithmic classical overhead in problem parameters.}

\begin{algorithm}[t]
\raggedright
\caption{Fourier sampling of $\br{\phi} s(A)\kt{\psi}$}\label{alg:overlap}
\vspace{2pt}
    \textbf{Problem input:}
    \vspace{-2pt}
    \begin{itemize}
        \item $N \times N$ Hermitian matrix ${A}=\sum_{\ell} a_{\ell} P_{\ell}$, with $a_{\ell}\in \mathbb{R}\ \forall \ell$, and known Pauli weight $\lambda = \sum_{\ell} |a_{\ell}|$
        \item Fourier parameters $\{\alpha_k\}_{k \in F}$, $\{t_k\}_{k \in F}$ for series $s(A)$ according to Eq.~\eqref{eq:fourier-A}
        \item Unitary gates $U_{\psi}, U_{\phi}$ which respectively prepare $\kt{\psi}, \kt{\phi}$.
    \end{itemize}
    \textbf{Parameters:} (1) Approximation error $\varepsilon\in(0,1)$;\; (2) Success probability $1-\delta\in(0,1)$;\; (3) Runtime vector $\vec{r}\in\mathbb{N}^{|F|}$. \\ \vspace{5pt}
    \textbf{Output:} Approximation to $\br{\phi} s(A)\kt{\psi}$ with additive error at most ${\varepsilon}$ and success probability at least $1-\delta$. 
    \begin{enumerate}
    \itemsep -0.8\parsep
        \item Compute coefficients $\{\beta_{km}^{(\vec{r})} \}_{(k,m) \in F\times T}$ as in Eq.~\eqref{eq:kianna-ham-simulation-simp} (see Appendix \ref{appdx:fourier-sampling}) for each time parameter $\{t_k \}_{k \in F}$.
        \item $R(\vec{r}) \leftarrow \sum_{(k,m) \in F\times T} \big|\alpha_k \beta_{km}^{(r_k)} \big|$.
        \item $M(\vec{r}) \leftarrow \left\lceil 4\ln (2 / \delta)\left({R(\vec{r})}/{\varepsilon}\right)^{2}\right\rceil$. 
        \item \textbf{for} $j \in \big[M(\vec{r})\big]:$
        \vspace{-0.5em}
        \begin{enumerate}[label=(\roman*)]
            \item Sample indices $(k',m')$ according to probability distribution $\{|\alpha_k \beta_{km}^{(r_k)} |/R(\vec{r})\}_{(k,m) \in F\times T}$ and select corresponding string of gates $U^{(r_{k'})}_{k'm'}$.
            \item Run $n+1$ qubit circuit in Figure \ref{fig:circuits}(a) with input state $\kt{0}^{\otimes \ceil{\log N} +1}$,  $G=\id$,  controlled $U_{\phi}U^{(r_{k'})}_{k'm'}U_{\psi}$ gates, and measurement operator $Z$ on first register. Record measurement outcome $o_j$.
            \item Run same circuit again with $G=S^{\dag}:=\ketbra{0}{0}-i\ketbra{1}{1}$. Record measurement outcome $o'_j$.
            \item $z_j \leftarrow R(\vec{r})(o_j+io'_j) $.
        \end{enumerate}
        \item \textbf{end for}
        \item $\overline{z}^{(M)} \leftarrow \sum_j z_j/M(\vec{r})$. \textbf{return} $\overline{z}$. 
    \end{enumerate}
\end{algorithm}


\subsection{Main result}\label{sec:main-result}


Using the results of the previous section, now we can demonstrate how to sample properties of matrix functions, starting from a sufficiently good Fourier series approximation to the function, and a decomposition of the matrix in the Pauli basis.

Given a real-valued function $f:\mathbb{R}\rightarrow\mathbb{R}$, we consider a scenario where we have a Fourier series approximation $s(\varepsilon,D_A):\mathbb{R}\rightarrow\mathbb{R}$ which is $\varepsilon$-close to $f$ on the domain $D_A$. More precisely, we suppose that
\begin{align}\label{eq:fourier-approx}
    \left| f(x) - s(\varepsilon,D_A,x)\right| \leq \varepsilon\,, \quad  \forall  x \in D_A\,.
\end{align}
We note that the condition in Eq.~\eqref{eq:fourier-approx} is dependent on the maximum deviation of the Fourier approximation over the entire domain of interest -- this is the property that will determine the rigorous worst case complexity of our algorithms -- thus, it is important to find Fourier series approximations that are accurate even on the extremities of a domain, not just on average. The Fourier series can always be expressed as
\begin{align}
s(\varepsilon,D_A,x) = \sum_{k\in F_{\varepsilon,D_A}} \alpha_k(\varepsilon,D_A) \exp{\big(it_k(\varepsilon,D_A)x\big)}\,,
\end{align}
over some index set $F_{\varepsilon,D_A}$, where $\{\alpha_k(\varepsilon,D_A)\}_k$ and $\{t_k(\varepsilon,D_A) \}_k$ are Fourier parameters which in general have dependence on the approximation error $\varepsilon$ and the domain of approximation $D_A$.  This setting can be simply translated to matrix functions, if one considers a Hermitian matrix $A$ whose spectrum lies in $D_A$. In this case $\varepsilon$ then corresponds to closeness of matrix functions in operator norm.

\begin{theorem}[Generalized sampling from Fourier approximations]\label{thm:general-sampling}
Suppose we have a matrix function $f(A)$ of an $N\times N$ Hermitian matrix $A$ that is approximated by a Fourier series $s(\tilde{\varepsilon},A)$ with tunable error parameter $\tilde{\varepsilon}$, and has a known Pauli decomposition of $A$ with Pauli weight $\lambda$. Suppose further that we have unitary oracles $U_{\psi}$, $U_{\phi}$, $U_{\rho}$ to prepare $\kt{\psi}$, $\kt{\phi}$ and $\rho$ respectively. Then, we give explicit randomized algorithms to
\begin{align}
\text{approximate\; $\frac{\br{\phi}f(A)\kt{\psi}}{q}$ \,\text{\&}\; $\frac{\Tr\left[f(A)\rho f(A)^{\dag} O\right]}{q^2}$}\,, \label{eq:thm1-quantities}
\end{align}
to a given small enough additive error $\varepsilon$, using Algorithms \ref{alg:overlap} and \ref{alg:observable} respectively on $\log(N)+1$ qubits, where $q$ is some arbitrary normalization factor.
\end{theorem}

We provide a more precise statement with exact complexities, and accompanying proofs accounting for constant factor terms, in Appendix \ref{appdx:thm1}. Theorem~\ref{thm:general-sampling} provides a general recipe for sampling from functions of $N \times N$ Hermitian matrices by using $\log(N)+1$ qubits, given a Fourier series approximation and a Pauli decomposition of the matrix. In particular, it specifies how to tune the Fourier approximation parameter $\tilde{\varepsilon}$ such that we can directly use Algorithms \ref{alg:overlap} and \ref{alg:observable}. We stress that as with the warm-up problem, we do not use any hidden quantum oracles and specify circuits explicitly in terms of controlled Pauli gates and controlled Pauli rotations. 

{We note that a property of our algorithms in Theorem \ref{thm:general-sampling} is that the output is a number, rather than a quantum state. We envision in the majority of potential applications for quantum algorithms the goal is to extract classical information out of a quantum state \cite{dalzell2023quantum}. In many cases this would be captured by the quantities in Eq.~\eqref{eq:thm1-quantities}. If the application of a matrix function is to be used as a subroutine as part of a larger quantum computation, our framework allows further quantum processing by appending fixed controlled unitaries in the case of Algorithm \ref{alg:overlap} or simply appending fixed unitaries on the second register in Algorithm \ref{alg:observable}. Finally, some algorithmic frameworks ask instead to sample from a quantum state in the computational basis \cite{kerenidis2017quantum}. We highlight that in our with the following remark that we can similarly statistically recover the output vector in the computational basis. Full exposition is provided in Appendix \ref{appdx:vector-sampling}}

\begin{remark}[{Sampling from output vector}]
{  By modifying the measurement in Algorithm \ref{alg:observable} to computational basis measurements we give an unbiased estimator for the vector whose entries are (approximately) $\bra{\vec{z}_n}f(A)\ket{\psi}$ for each $\vec{z}_n \in \{0,1\}^n$.
}
\end{remark}

Our randomized scheme allows for generic normalization of the answer by some factor $q$, with all complexities accounted for. So far, we have assumed that $q$ is exactly given. However, the desired $q$ may in general not be exactly known. One salient example is if we wish for $f(A)\kt{\psi}$ to be a normalized quantum state for some input state $\kt{\psi}$. In this case the relevant normalization quantity is $\|f(A)\kt{\psi}\|$. In Appendix \ref{appdx:norm-sampling} we present a randomized subroutine to estimate quantities of this form, which can be directly integrated with our core result in Theorem~\ref{thm:general-sampling}. We summarize this with the following proposition.

\begin{proposition}[Sampling normalization constant]\label{prop:norm-sampling}
Under the conditions of Theorem \ref{thm:general-sampling}, the quantities
\begin{align}\label{eq:norm-sampling}
\text{$\frac{\br{\phi} f(A)\kt{\psi}}{\|f(A)\kt{\psi}\|}$ \,\text{\&}\; $\frac{\Tr\left[f(A)\ketbra{\psi}{\psi} f(A)^{\dag} O\right]}{\|f(A)\kt{\psi}\|^2}$}\,,
\end{align}
can be approximated to a given desired additive error $\varepsilon$, with the addition of a subroutine on $\log(N)+1$ qubits to approximate $\|f(A)\kt{\psi}\|$.
\end{proposition}

Proposition \ref{prop:norm-sampling} specifies how to approximate quantities in Eq.~\eqref{eq:norm-sampling} by first running a subroutine to approximate $\|f(A)\kt{\psi}\|$ to a specified error, followed by the core algorithm using Theorem \ref{thm:general-sampling}.


\subsection{Classically easy functions}\label{sec:classical-power}

We briefly remark on the implications of our techniques for the power of randomized classical approaches using the Pauli access model. Namely, we show in the following proposition that sufficiently low-degree polynomials of matrices (such as for matrix multiplication) can be sampled from, if these matrices have low Pauli weight.

\begin{proposition}[Classical polynomial sampling]\label{prop:classical-simp}
For stabilizer states $\kt{s}$, $\kt{t}$, Hermitian matrix $A$ with known Pauli decomposition and Pauli weight $\lambda$, and $p_d$ a $d$-degree polynomial {with coefficients of magnitude $\OC(1)$}, we give a classical randomized scheme to
\begin{align}\label{eq:classical-alg-goal}
    \text{approximate\, $\br{t}p_d(A)\kt{s}$\, up to additive error $\varepsilon$}\,,
\end{align}
using
\begin{align}\label{eq:classical-complexity}
    \CC_{\emph{sample}}=\widetilde{\OC}\left( {\lambda^{2d}}/{\varepsilon^2}\right)\,
\end{align}
independent classical subroutines each with time and space complexity
\begin{align}
    \CC_{\emph{time}}=\widetilde{\OC}\left( d\log^2(N) \right) \quad \text{\&} \,\quad
    \CC_{\emph{bits}}=\widetilde{\OC}\left( \log^2(N) \right) \,,
\end{align}
respectively, where $\lambda$ is the Pauli weight of $A$.
\end{proposition}

A more detailed statement, along with the proof thereof, can be found in Appendix \ref{appdx:classical}. Proposition \ref{prop:classical-simp} works by observing that the task in Eq.~\eqref{eq:classical-alg-goal} can be obtained statistically via measurement outcomes of depth $\OC(d)$ Clifford circuits, which are efficiently simulable. We note that, as with our previous specified algorithms, each sample can be parallelized, thus moving $\CC_{{sample}}$ into the space complexity. Proposition \ref{prop:classical-simp} also trivially extends to polynomials of multiple matrices. Thus, our result implies that classically sampling from certain primitives such as matrix multiplication can be efficient, if the matrices have low Pauli weight. 

Proposition \ref{prop:classical-simp} implies similar efficiency for quantum or classical algorithms for functions with low-degree polynomial approximations. However, in general we do not expect such low-degree approximations to always exist. For instance, we can investigate some implications for the linear systems problem and ground state property estimation problem. Ref.~\cite{childs2017quantum} gives a polynomial approximation to the inverse function with degree linear in the condition number. This implies a classical algorithm to sample an element of the solution vector of the linear systems problem with exponential sample complexity in the condition number. Likewise, one could use the power law method to approximately project to the ground state for Hamiltonians with negative spectra. However, this results in a classical algorithm with sampling complexity exponential in the inverse spectral gap (see Appendix \ref{appdx:power-method} for more details). 
It thus remains to see if there are problems of interest with low-degree polynomial approximations.

{Similar results have also been shown for matrix powers for the sparse access model}~\cite{cade2018quantum, apers2022simple}{, where the base of the exponential in Eq.~\eqref{eq:classical-complexity} is different. Interestingly, the problem of evaluating powers to additive error $\OC(\|A\|^d\varepsilon)$ (rather than $\OC(\lambda^d\varepsilon)$) has been shown to be BQP-complete} \cite{janzing2007simple}{, and classically hard} \cite{montanaro2023quantum}{, again in the sparse access model} {. The QSVT framework} \cite{gilyen2019quantum, martyn2021grand} {can apply a more general class of polynomial transformations with cost scaling only linearly in the degree. We leave it as an open question as to whether there are more efficient randomized \textit{early fault-tolerant}} quantum algorithms for high-degree polynomials.

\subsection{Sampling from other linear combinations of unitaries}\label{sec:other-LCUs}

In our main result we present Monte Carlo sampling algorithms for a specific decomposition of functions into linear combinations of implementable unitaries -- first by decomposing the function into time evolution operators via a Fourier decomposition, second by decomposing those time evolution operators into Pauli gates and Pauli rotations. The resulting runtime complexities depend on the properties of the Fourier approximation chosen for the matrix function of interest $f(A)$. One pertinent question is then: when searching for appropriate Fourier approximations, what sample complexity would we expect at best with these techniques?

Our algorithms sample from a weighted probability distribution of unitaries with the weight directly factoring into the complexity. Assuming that we require an $\varepsilon$-close approximation on all states, this weight is lower bounded by $\|f(A)\|$, for any chosen Fourier decomposition. Thus, under our presented sampling framework, Hoeffding's inequality gives sufficient conditions for to approximate $\br{\phi}f(A)\kt{\psi}$ and $\Tr[f(A)\rho f(A)^{\dag}O]$ using $\OC\left({\|f(A)\|^2}/{\varepsilon^2} \right)$ and $\OC\left({\|f(A)\|^4}/{\varepsilon^2} \right)$ respective samples at best. We remark that this argument holds when  sampling from \textit{any} decomposition of $f(A)$ into a linear combination of unitaries, not just the one we consider specifically in this work. As we shall see in the following applications sections, we achieve this for the linear systems problem (Section \ref{sec:linear-systems}) and the ground state property estimation problem (Section \ref{sec:gs}) up to logarithmic factors. 

We stress that this discussion does not constitute generic sample complexity lower bounds (for instance, one could in practice obtain the desired result using fewer samples than Hoeffding's inequality specifies). Additionally, we do not rule out that better sample complexity guarantees can be obtained for other {randomized} schemes beyond Monte Carlo here. Finally, one can achieve problem-dependent improvements to the sample complexity by approximating $f(A)$ on the relevant subspace for the problem at hand. This is how a better sample complexity guarantee is achieved for our Gibbs state algorithm in Section \ref{sec:gibbs}.


\section{Linear systems}\label{sec:linear-systems}

\subsection{Randomized quantum linear system solver}

In this section we show how to apply Theorem \ref{thm:general-sampling} to sample from the inverse of a matrix, with applications for linear systems problems. We use the Fourier series approximation for the inverse function found in Ref.~\cite{childs2017quantum}. With this, we establish the following result \footnote{In the statement of results in the main text, for simplicity we write \unexpanded{$\CC_{sample}= \widetilde{\OC}(g)$} if \unexpanded{$\CC_{sample}= \OC(g\, \mathrm{polylog}\, g\, \mathrm{polylog}\, h)$} and \unexpanded{$\CC_{gate}= \Omega(h)$}, and vice versa for \unexpanded{$\CC_{gate}$}. That is, with \unexpanded{$\widetilde{\OC}(\cdot)$} notation we drop any terms that contribute polylogarithmically to the total (unparallelized) runtime.}.

\begin{corollary}[Linear systems]\label{cor:linear-systems}
Consider a Hermitian matrix $A$ with known Pauli decomposition as in Eq.~\eqref{eq:pauli-decomposition} with Pauli weight $\lambda$. Denote $q$ as a freely chosen normalization parameter. Finally, suppose we have ability to prepare state $\kt{\vec{b}}$ in $\OC(d_{\vec{b}})$ depth. Then, we find:

\emph{(a)} Given ability to implement $\kt{\psi}$ via unitary $U_{\psi}$ in gate depth $d_{\psi}$, we have a randomized quantum algorithm to
    \begin{align}
        \text{approximate\, $\frac{\br{\psi}A^{-1}\kt{\vec{b}}}{q}$ up to additive error $\varepsilon$\,,}
    \end{align}
    with arbitrary constant success probability, utilizing $\CC^{\psi}_{\emph{sample}}$ quantum circuits of the form in Figure \ref{fig:circuits}(a) each consisting of  $\CC^{\psi}_{\emph{gate}}$ layers of gates, where 
    \begin{equation}
        \CC^{\psi}_{\emph{sample}}=\widetilde{\OC}\left( \frac{\|A^{-1}\|^2}{\varepsilon^2q^2}\right),\; 
        \CC^{\psi}_{\emph{gate}}= \widetilde{\OC}\left(\|A^{-1}\|^2\lambda ^2  + d_{\psi} + d_{\vec{b}} \right). 
    \end{equation}
    
\emph{(b)} Given ability to measure observable $O\,;\; \|O\|\leq 1$, we have a randomized quantum algorithm to
    \begin{align}
        \text{approximate\, $\frac{\br{\vec{b}}A^{-1}OA^{-1}\kt{\vec{b}}}{q^2}$ up to additive error $\varepsilon$\,,}
    \end{align}
    with arbitrary constant probability, utilizing $\CC^{O}_{\emph{sample}}$ quantum circuits of the form in Figure \ref{fig:circuits}(b) each consisting of $\CC^{O}_{\emph{gate}}$ layers of gates and one measurement of $O$, where
    \begin{equation}\label{eq:linear-systems-observable}
        \CC^{O}_{\emph{sample}}=\widetilde{\OC}\left( \frac{\|A^{-1}\|^4}{\varepsilon^2q^4}\right)\,,\;\; 
        \CC^{O}_{\emph{gate}}= \widetilde{\OC}\left(\|A^{-1}\|^2\lambda ^2 + d_{\vec{b}} \right)\,.
    \end{equation}
\end{corollary}

We detail a proof of this corollary in Appendix \ref{appdx:linear-systems} where we specify polylogarithmic contributions and account for constant factors. Corollary \ref{cor:linear-systems} provides an oracle-free quantum linear systems algorithm that only uses one ancillary qubit for Hermitian matrices. We note that we can consider non-Hermitian matrices for linear systems by embedding them in a larger Hermitian matrix via an additional qubit (details are provided in Appendix \ref{appdx:linear-systems}).

\begin{remark}[Non-Hermitian matrices]
We can extend Corollary \ref{cor:linear-systems} to include non-Hermitian matrices by embedding the matrix in a larger Hermitian matrix with the aid of a single additional qubit. The algorithm then uses $\log(N)+2$ qubits with a factor $2$ increase in sample complexity.
\end{remark}

Further, we note that Proposition \ref{prop:norm-sampling} allows the evaluation of the norm of the output vector $\|A^{-1}\kt{\vec{b}}\|$ via a subroutine with complexities equivalent to Eq.~\eqref{eq:linear-systems-observable} up to logarithmic terms.


\subsection{Statistical encoding of input vector}

So far we have supposed that $\kt{\vec{b}}$ is provided via a quantum oracle. With the below proposition we demonstrate an approach to statistically encode the vector at a cost of classical preprocessing overhead and minimal gate overhead.

\begin{proposition}[Statistical encoding of input vector]\label{prop:statistical-encoding}
Given a classical vector $\vec{b}$ with sparsity $s$, we give a scheme to replace the quantum oracle for the input state $\kt{\vec{b}}$ in Theorem \ref{thm:general-sampling} and Corollary \ref{cor:linear-systems}. This has overhead:
\begin{itemize}
    \item $\OC(s)$ time for classical preprocessing.
    \item Sample complexity factor increase for problem (a) of $\OC(\|\vec{b}\|_1^2/\|\vec{b}\|_2^2)$ and for problem (b) of $\OC(\|\vec{b}\|_1^4/\|\vec{b}\|_2^4)$.
    \item 1 or 2 single-qubit controlled $\ceil{\log N}$-qubit Pauli gates each circuit sample.
\end{itemize}
\end{proposition}

More details are provided in Appendix \ref{appdx:randomized-encoding}. A key distinction is that this approach can return an encoding of $\vec{b}$ that represents the true magnitude of the vector rather than a normalized quantum state which a standard quantum oracle would provide, simply by increasing the sample complexity. Moreover, this is a general scheme for our randomized approach that is applicable beyond the linear systems setting. We remark that our representation for the quantum state $\kt{\vec{b}}$ is similar to the so-called CQS states considered in the near-term quantum linear systems approach of Ref.~\cite{huang2019near}. {Finally, we stress that in this scheme we never explicitly prepare the quantum state. A comparison with explicit state preparation schemes \cite{zhang2022quantum} shows that this avoids $\OC(Ns)$ cost in qubit overhead, and trades it for classical preprocessing overhead and sample overhead.}


\subsection{Comparison}\label{sec:linear-systems-comparison}

\subsubsection{Quantum linear systems solvers}

The quantum linear systems solver (QLSS) is an algorithm that takes an encoding of a matrix $A$ and vector $\vec{b}$, and returns the quantum state proportional to $A^{-1}\kt{\vec{b}}$, where $\kt{\vec{b}}:=\sum_i b_i \kt{i}/ \ {\|\vec{b}\|}$. Since the proposal of the original QLSS, known as the Harrow, Hassidim, and Lloyd (HHL) algorithm \cite{harrow2009quantum}, many improved schemes have been proposed \cite{ambainis2012variable,childs2017quantum, gilyen2019quantum, subasi2019quantum,lin2020optimal, martyn2021grand, an2022quantum,costa2021optimal} utilizing $\OC(\log(N))$ algorithmic qubits. Early algorithms made use of a time evolution oracle in the matrix of interest \cite{childs2017quantum, ambainis2012variable, subacsi2019quantum, an2022quantum}. 
More recently, approaches which use a block encoding have been proposed \cite{gilyen2019quantum, martyn2021grand, costa2021optimal}. 

We remark that, in contrast to other fault-tolerant approaches to the QLSS, our algorithms do not require coherent access to the matrix of interest $A$ and thus there are no additional hidden dimension dependencies. Thus, for a general $N\times N$ Hermitian matrix, only $\log(N) + 1$ logical qubits are required to carry out the full algorithm, with a runtime that does not explicity depend on the dimension. We further note that our algorithm has the following additional distinctive features:

\vspace{-1mm}
\begin{itemize}
\setlength\itemsep{0.05em}
    \item In contrast to the QLSS, the error parameter $\varepsilon$ in Corollary \ref{cor:linear-systems} is specified in terms of measurement errors in extracting information out of states of the form $A^{-1}\kt{\vec{b}}$. Thus, the QLSS is distinct from the problem our algorithm solves, as our algorithm returns a number rather than a state. If one wishes to extract out classical information from the QLSS one requires $\OC(\varepsilon^{-2})$ circuit samples with incoherent measurements, or an additional factor of  $\OC(\varepsilon^{-1})$ runtime with coherent approaches \cite{knill2007optimal, wang2019accelerated}. 
    
    \item Our complexities are determined by the operator norm $\|A^{-1}\|$, instead of the condition numbers $\kappa:= \|A^{-1}\|\|A\|$ or $\kappa_F:=\|A^{-1}\|\|A\|_F$.
    Moreover, there is no sparsity or explicit dimension dependence in the runtime. Instead, the runtime depends on the Pauli weight $\lambda$ which for certain problems can be much smaller than the sparsity or dimension. 

    \item We allow for arbitrary normalization by any normalization factor $q$, which rescales the asymptotic complexities by $\|A^{-1}\| \rightarrow \frac{1}{q}\|A^{-1}\|$. We recall that in the QLSS, the solution is given as a normalized quantum state, which may be undesired. {We discuss the role of normalization further in the following subsection.} 
\end{itemize}

\vspace{-1mm}

\subsubsection{Complexity comparison}

In Table \ref{table:linear-systems} we compare the complexities of our linear systems algorithm with other classical and quantum algorithms in the literature for a specific task.\footnote{In the tables (different from the main text) for each cell we detail contributions from all parameters even if the contribution of that parameter to the total unparallelized runtime is exponentially worse. This is to highlight the dependence of each parameter in each cell in isolation so that each dependency can be compared across different schemes.} The task we consider in Table \ref{table:linear-systems} is to
\begin{align}\label{eq:task-linear-systems}
    &\text{approximate the element $\frac{1}{c}\cdot\left(A^{-1}\vec{b}\right)_i$ to additive error $\varepsilon$}, \; \text{for some choice of $i \in [N]$ and normalization $c\in\mathbb{R}$} \,.
\end{align}
We assume that the normalization constant $c$ is whatever is natural to the algorithm at hand, and that we are querying the vector $A^{-1}\vec{b}$ in the computational basis. Thus, $c = 1$ for {in the usual} classical setting, and {$c=\|A^{-1}\vec{b}\|$ for quantum solvers and quantum-inspired classical solvers}. This is to allow a "middle ground" comparison with the other quantum and classical solvers, as relaxing these two assumptions incurs additional overhead for other quantum algorithms and classical algorithms, respectively.

We compare the space resources required to create the data access to $A$, and to perform the algorithm. We also compare the maximum coherent quantum runtime required (if applicable), and the total runtime. For our randomized scheme, we assume no parallelization, thus the total runtime is simply the product of the gate depth and the sample complexity. For all algorithms, we assume that the data starts in the most amenable \textit{classical} format for that algorithm, such as in a sparse row representation, or in our case the classical Pauli access model. From there, we keep track of any additional \textit{classical}  or \textit{quantum} overheads (both space and time) needed to provide any required data structures, for instance, the resources required to provide a block encoding.

\begin{center}
\begin{table*}[t]
\centering
\setlength\extrarowheight{-1pt}
{\renewcommand{\arraystretch}{2}%
\resizebox{\textwidth}{!}{%
    \begin{tabular}{ | P{6em} | P{6em} !{\vrule width 1.2pt}  P{6em} | P{8.1em} | P{2.5em} | P{10.5em}  | P{8.5em} | P{3.5em} | }
    \hline\hline Algorithm & Access Model
    & Data Access Space & Algorithmic Space & {AE ?} &  Quantum Gate Depth & Total Runtime & {Norm} \\ \thickhline
    Textbook Classical & Comp.~basis elements & $\OC(N^2)$ bits & $\OC(N^2)$ bits & - & - & $\OC(N^{\omega})$ & 1 \\ \hline
    Randomized Classical \cite{strohmer2009randomized} &  Sparse row  & $\OC(Ns)$ bits & ${\OC}\left(s\kappa_F^2 \log(1/\varepsilon)\right) $ bits 
    & - & - & ${\OC}\left(s{\kappa_F}^2 \log(1/\varepsilon)\right) $ & 1 \\ \hline
    Quantum- Inspired \cite{shao2022faster} & Sampling \& query & $\OC(N^2)$ bits & ${\OC}\left(\kappa_F^2\log(1/\varepsilon) \right)$ bits 
    & - & - & $\widetilde{\OC}\left(\kappa^2 \kappa^6_F /\varepsilon^{2}\right)$ & $\|A^{-1}\vec{b}\|$ \\ \hline\hline
    \multirow{2}{6em}{\centering  Textbook Quantum \cite{harrow2009quantum}} & \multirow{2}{6.5em}{\centering Block encoding time evol.~\cite{low2019hamiltonian, clader2022quantum}}& \multirow{2}{6.5em}{\centering $\OC(N^2)$ qubits} & \multirow{2}{8.6em}{\centering $\OC(\log N)$ qubits} & yes & $\widetilde{\OC}\big({\kappa\kappa_F} \log(N)/\varepsilon^2 \big)$ & $\widetilde{\OC}\big({\kappa\kappa_F}\log(N)/\varepsilon^2 \big)$ & $\|A^{-1}\vec{b}\|$  \\ \cline{5-8} & & & & no &  $\widetilde{\OC}\big({\kappa\kappa_F} \log(N)/\varepsilon \big)$  & $\widetilde{\OC}\big({\kappa\kappa_F} \log(N)/\varepsilon^3 \big)$ & $\|A^{-1}\vec{b}\|$ \\ \hline
    \multirow{2}{6em}{\centering  Quantum Query-Optimal \cite{costa2021optimal}} & \multirow{2}{6.5em}{\centering Frobenius norm block encoding \cite{clader2022quantum}} & \multirow{2}{6.5em}{\centering $\OC(N^2)$ qubits} & \multirow{2}{8.6em}{\centering $\log N + 6$ qubits} & yes & $\widetilde{\OC}\big({\kappa_F}\log(N)/\varepsilon \big)$ & $\widetilde{\OC}\left({\kappa_F} \log(N) / \varepsilon \right)$ & $\|A^{-1}\vec{b}\|$  \\ \cline{5-8} & & & & no &  $\widetilde{\OC}\big({\kappa_F}\log(N) \log(1/\varepsilon) \big)$ & $\widetilde{\OC}\left({\kappa_F} \log(N) / \varepsilon^2 \right)$ & $\|A^{-1}\vec{b}\|$ \\ \hline
    Randomized Quantum (This Work) & Pauli coefficients & $\OC(N^2)$ bits & $\log N+2$ qubits  & no &   $\widetilde{\OC}\left(\|A^{-1}\|^2\lambda^2 \log^2(1/\varepsilon)\right)$  & $\widetilde{\OC}\left(\|A^{-1}\|^4\lambda^2/\varepsilon^{2} \right) $ & $\|\vec{b}\|$ \\ \hline\hline
    \end{tabular}
}
    \caption{End-to-end complexities for the task in Eq.~\eqref{eq:task-linear-systems} of querying one element of the solution vector in the linear systems problem to additive error $\varepsilon$. We consider a general $N \times N$ complex coefficient matrix $A$, with row sparsity $s$, which has Pauli decomposition (as in Eq.~\eqref{eq:pauli-decomposition}) into up to $N^2$ Pauli terms with weight $\lambda $.  We denote $\kappa_F:=\|A^{-1}\|\|A\|_F$ where $\|\cdot \|$ denotes the operator norm and $\|\cdot\|_F$ denotes the Frobenius matrix norm, and note that $\kappa_F \geq \sqrt{\textrm{rk}(A)}$. We also denote $\kappa := \|A^{-1}\|\|A\|$. Under {``Data Access Space''} we have listed the dominant classical or quantum resources needed in order to represent $A$ in the required format for the algorithm at hand. {We have marked whether a given scheme requires approaches based on amplitude estimation (AE) to estimate the target quantity. Finally, each algorithm returns a different normalization $c$ according to Eq.~\eqref{eq:task-linear-systems}, which can factor into the complexity of each algorithm; we indicate this normalization in the final column.} One can also consider different schemes that may be advantageous for Pauli-sparse matrices or matrices which are sparse in the computational basis; we discuss this later in this section and in Table \ref{table:linear-systems2}. \label{table:linear-systems}}
}
\end{table*}
\end{center}

Taking into account the aforementioned features, in Table \ref{table:linear-systems} we consider the following algorithms:

\vspace{-1mm}

\begin{itemize}
    \item The Gaussian elimination "textbook" classical method, which returns the full exact solution vector. This has runtime dependent on the matrix multiplication exponent $\omega<2.372$.
    \item A randomized classical Kaczmarz method \cite{strohmer2009randomized}, which returns an $\varepsilon$-approximation to the full solution vector.
    \item A "quantum-inspired" approach \cite{shao2022faster}. This starts from a classical data structure intended to mimick QRAM, which allows sampling from probability distributions with probabilities proportional to the magnitude of elements in a given row of $A$. The algorithm returns a classical data structure which allows one to sample from individual elements of an $\varepsilon$-approximaton to the solution vector. In order to estimate the norm one needs to pay an extra $1/\varepsilon^2$ factor \cite{chia2022sampling, shao2022faster}. 
    \item The HHL "textbook" quantum  algorithm \cite{harrow2009quantum}, which requires a time evolution oracle. In Table \ref{table:linear-systems} we assume that this oracle is provided in a runtime-efficient manner via a block encoding \cite{low2019hamiltonian}. We assume a Frobenius norm block encoding, with explicit construction via a QRAM in minimal depth as detailed in Ref.~\cite{clader2022quantum}. We note there are also other approaches which are more space efficient in quantum resources, but more costly in quantum runtime. Further, there is a choice of how to extract the classical information from the output of the HHL algorithm. We include two rows in Table \ref{table:linear-systems}, which quantifies resources required for performing  {coherent approaches for extracting expectation values }\cite{knill2007optimal, wang2019accelerated} {(first row), and standard incoherent measurements assuming no parallelization which have worse overall runtime, but use substantially less gate depth (second row)}.
    \item A state-of-the-art quantum linear systems solver which achieves optimal query complexity \cite{costa2021optimal}, using a block encoding access model. As with HHL we assumed a Frobenius norm block encoding implemented in low depth via a QRAM, and detail complexities both for incoherent measurements (first row) and for coherent approaches \cite{knill2007optimal, wang2019accelerated} (second row).
    \item Our algorithm using Corollary \ref{cor:linear-systems}, assuming no parallelization. We remark that a particular artefact of the simplified task in Eq.~\eqref{eq:task-linear-systems} is that the runtime of our algorithm can be reduced by an arbitrary factor of $q^2$ as in Theorem \ref{thm:general-sampling}. However, we do not include this additional freedom in Table \ref{table:linear-systems} as we do not expect it to arise in more practical tasks. For instance, if one wished to  compare two elements of $A^{-1}\vec{b}$ by estimating their ratio, $q$ would not appear in the complexity.
\end{itemize}

The goal of our work is to reduce the quantum resources required for quantum approaches. This can be quantified by the three columns highlighted in Table \ref{table:linear-systems}: data access space, algorithmic space, and quantum gate depth. Our algorithm is clearly more efficient in space requirements compared to other quantum algorithms, and in particular it uses no qubits for data access. We further note that if the Pauli weight $\lambda $ is small, there is scope that our algorithm has competitive gate depth with the best alternate option, which is to use the adiabatic quantum algorithm of Ref.~\cite{costa2021optimal} with incoherent measurements (we remark there exist matrices for which $\lambda $ is much smaller than $L$, which in general can grow with the dimension of the system). If a coherent approach is used to extract the vector element, then our algorithm has exponentially better gate depth in the error parameter, but worse overall runtime. 

{\textit{The role of normalization.}
We note that in some practical tasks one may care about the true magnitude of vector elements. In this case, the dependence of other quantum algorithms on $\varepsilon^{-1}$ should be scaled up by a factor $\|A^{-1}\kt{b}\|$ in order to match our randomized algorithm, or our dependence on $\varepsilon^{-1}$ should be scaled down by a factor of $\|A^{-1}\kt{b}\|$ in order to match other quantum approaches. In addition, if the normalization is not given, computing it may incur significant overhead which should be accounted for. When starting in the block encoding model, Chakraborty et al.~demonstrate how to obtain the state normalization to multiplicative error in $\widetilde{\OC}(\kappa \mu \poly\!\log(N)/\varepsilon)$ queries to the block encoding, where $\mu$ is the block encoding subnormalization \cite{chakraborty2019power}. When searching for potential applications for quantum algorithms, the desired normalization for the problem of interest is an important consideration when benchmarking different algorithms.
}

{\textit{Criteria for quantum advantage.}}
For our algorithm to potentially be useful, it must also compete with classical algorithms. We remark that all the classical runtimes stated have at least linear dimension dependence for generic matrices, as $\kappa_F \geq \sqrt{\textrm{rk}(A)}$, where we denote $\textrm{rk}(A)$ as the rank of $A$. {In theory a block encoding with a subnormalization of $\|A\|$ rather than $\|A\|_F$ may be possible, but at present an explicit construction for generic matrices is not known to the best of the authors' knowledge.} For any of the quantum algorithms to display superpolynomial advantage for generic matrices, one requires at minimum $\kappa = \|A^{-1}\|\|A\| = \OC(\textrm{polylog}(N))$. {Block encodings with stronger subnormalizations also need to be found (for generic matrices), as the scope for superpolynomial advantage is constrained to high-rank matrices.} Further, one must consider a setting where the  significant difference in error dependence is also accounted for. Whilst our algorithm has significantly worse overall runtime in $\|A^{-1}\|$, we note the requirement for advantage on our algorithm is similar ($\|A^{-1}\| = \OC(\textrm{polylog}(N))$), although for practical advantages the degree of this polynomial may need to be kept small. {Moreover, there is no rank condition yet to be overcome here.} Thus, we expect our algorithm to potentially show quantum advantage if there is a setting where other quantum algorithms also show advantage, given that the Pauli weight $\lambda$ is small. We remark that it is possible for matrices to have a small $\lambda $, but large computational basis sparsity $s$, Pauli sparsity $L$, or rank $\textrm{rk}(A)$.

{Two known domains where block encodings with subnormalizations that don't lead to rank dependence are sparse matrices in the computational basis and the Pauli basis. As discussed above this is a requirement for any hope of superpolynomial quantum advantage. Our algorithmic construction can be thought of as extending the scope of the second category to also include non-sparse matrices in the Pauli basis that have a low Pauli weight $\lambda$. We now compare our algorithm to other approaches in these two settings in further detail.}

\textit{Pauli-sparse matrices.} So far we have compared our algorithm against other approaches presuming that the Pauli sparsity $L$ is large (note that $L= \OC(N^2)$ in general). If $L$ is small and the Pauli decomposition is known, then the \texttt{SELECT} + \texttt{PREPARE} oracles (previously studied in algorithms for quantum chemistry \cite{childs2018toward, babbush2018encoding, berry2019qubitization,wan2021exponentially}) can implement a more efficient block encoding than the general construction with a QRAM considered in Table \ref{table:linear-systems}. We present the complexities for this block encoding in Table \ref{table:linear-systems2}. The qualitative conclusions are similar; the block-encoded quantum algorithm requires more qubits, and our algorithm can have comparable gate depth if the Pauli weight $\lambda $ is small compared to the Pauli sparsity $L$.

\begin{center}
\begin{table*}[t]
\centering
{\renewcommand{\arraystretch}{2}%
\resizebox{\textwidth}{!}{%
    \begin{tabular}{ | P{6em} | P{5.3em} !{\vrule width 1.2pt}  P{5.8em} | P{7.2em} | P{2.5em} | P{11.2em}  | P{8.4em} | P{3.5em} | }
    \hline\hline Algorithm & Access Model
    & Data Access Space & Algorithmic Space & {AE?} &  Quantum Gate Depth & Runtime & {Norm}  \\ \thickhline
    \multirow{2}{6em}{\centering  Quantum Query-Optimal \cite{costa2021optimal}} & \multirow{2}{5.8em}{\centering \texttt{SELECT} + \texttt{PREPARE} \cite{childs2018toward, babbush2018encoding, berry2019qubitization,wan2021exponentially}} & \multirow{2}{6em}{\centering $\OC( L)$ bits , $\OC(\log L)$ qubits} & \multirow{2}{7.5em}{\centering $\log N + 6$ qubits} & yes &$\widetilde{\OC}\big({\|A^{-1}\|\lambda} L/\varepsilon \big)$ & $\widetilde{\OC}\left({\|A^{-1}\|\lambda} L / \varepsilon \right)$ & $\|A^{-1}\vec{b}\|$ \\ \cline{5-8} & & & & no & $\widetilde{\OC}\big({\|A^{-1}\|\lambda} L \log(1/\varepsilon) \big)$ & $\widetilde{\OC}\left({\|A^{-1}\|\lambda} L / \varepsilon^2 \right)$ & $\|A^{-1}\vec{b}\|$ \\ \hline
    Randomized Quantum (This Work) & Pauli coefficients & $\OC(L)$ bits & $\log N+2$ qubits  & no & $\widetilde{\OC}\left(\|A^{-1}\|^2\lambda^2 \log^2(1/\varepsilon)\right)$ & $\widetilde{\OC}\left(\|A^{-1}\|^4\lambda^2/\varepsilon^{2} \right)$ & $\|\vec{b}\|$ \\ \hline\hline
    \end{tabular}
}
    \caption{End-to-end complexities for the task in Eq.~\eqref{eq:task-linear-systems} of querying one element of the solution vector in the linear systems problem to additive error $\varepsilon$. Here {for all algorithms} we have assumed that we start from a Pauli description of the matrix of interest $A$, as in Eq.~\eqref{eq:pauli-decomposition}, where the number of terms $L$ is small (for instance, a low degree polynomial in $\log N$). In this case, it can be advantageous to use a \texttt{SELECT} + \texttt{PREPARE} block encoding for the query-optimal quantum algorithm. {We compare complexities where the quantity in Eq.~\eqref{eq:task-linear-systems} is obtained using amplitude estimation (A.E.)-based approaches as well as incoherent sampling. In the final column we indicate the implicit normalization $c$ of each approach (see Eq.~\eqref{eq:task-linear-systems}).}}\label{table:linear-systems2}}
\end{table*}
\end{center}

\textit{Matrices sparse in the computational basis.} For the quantum query-optimal algorithm one can also consider the quantum sparse access model,\footnote{This consists of oracles to coherently access the non-zero row and column entries, as well as an additional oracle to coherently access all matrix entries.}  which can be efficiently converted into a block encoding model using $\OC(\textrm{polylog} (N/\varepsilon))$ qubits and $\OC(\textrm{polylog} (N/\varepsilon))$ elementary gates \cite{gilyen2019quantum}. Thus, if this access model is naturally available, a significant space and time saving can be made. However, this access will typically arise because of inherent structure in the matrix, which enables the computation of entries, given their indices. If this structure is not present, and we just have a generic sparse matrix, we still require QRAM for the most efficient block-encoding of the matrix. In this case, implementation of the quantum sparse access structure in minimal depth implementation uses $\widetilde{\OC}(Ns)$ qubits and ${\OC}(\log N)$ overhead \cite{di2020fault}, where $s$ is the minimization over the row sparsity and column sparsity. {Thus, the quantum space complexities could be up to quadratically better for sparse matrices than what is quoted in Table~\ref{table:linear-systems}, but it still is linear in $N$. The corresponding runtime is $\widetilde{\OC}(s\kappa \log(N)/\varepsilon)$ using coherent approaches for expectation value estimation} \cite{knill2007optimal, wang2019accelerated}. We remark that for specific classes of matrices there may also be more efficient ways to directly enact the sparse data access model \cite{camps2022explicit}, but we leave the comparison for such special cases to be beyond the scope of this work.

Finally, we remark that in Tables \ref{table:linear-systems} and \ref{table:linear-systems2} we presume that  the relevant data structure for the input vector $\vec{b}$ is efficiently provided. In general, providing the input state $\kt{\vec{b}}$ also requires a QRAM. We recall that with our randomized scheme we can circumvent this quantum resources cost by providing additional classical pre-processing overhead (see Proposition \ref{prop:statistical-encoding}). 

In this section, we have compared resource costs only for one specific task. For other applications, various classical and quantum overheads need to be carefully considered which can vastly change the complexities. One end-to-end problem in which one does not expect to achieve superpolynomial quantum advantage is where one needs to read off the entire solution vector. In this case, the additional ${\Omega}(N)$ tomographic overhead can be compared to the polynomial scaling of the classical runtimes in $\widetilde{\kappa}(A) \geq \sqrt{\textrm{rk}(A)}$. We remark that in this setting, one would only hope to achieve approximately at most a quadratic speedup in total runtime to the best randomized classical solver in dimension dependence, ignoring dependencies on the approximation error $\varepsilon$. 


\section{Ground state sampling}\label{sec:gs}

\subsection{Randomized quantum algorithm}

In this section we consider the problem of sampling properties of the ground state of a given Hamiltonian. In order to approximately project to the ground state we use the Gaussian function $e^{-\frac{1}{2}\tau^2x^2}$ for Hamiltonians with positive spectra, which has been proposed in previous works \cite{zeng2021universal, keen2021quantum}. A Fourier series representation was found for this function in Ref.~\cite{keen2021quantum}, which we use to establish the following corollary of our main result in Theorem \ref{thm:general-sampling}.

\begin{corollary}[Ground state property estimation]\label{cor:ground-state} Consider a Hamiltonian $H=\sum_l E_l \ketbra{E_l}{E_l}$ with all eigenvalues $E_l \geq 0$, and known Pauli decomposition $H=\sum_{\ell=1}^{L}a_{\ell}P_{\ell}\,;\; \lambda=\sum_{\ell}|a_{\ell}|$. Assume the spectral gap of $H$ is lower bounded by $\Delta \leq E_1 - E_0$. Additionally, we suppose that we have an initial trial state $\ketbra{\psi_0}{\psi_0}$ with overlap with the ground state $\gamma:= |\braket{\psi_0}{E_0}|$. Finally, we assume that $E_0\leq {\Delta}/\sqrt{2 \log\frac{\|O\|}{\varepsilon\gamma}}$ for some error parameter $\varepsilon$. Then, given ability to measure observable $O$, we give a randomized algorithm to 
\begin{align}
    \text{approximate $\br{E_0}O\kt{E_0}$ up to additive error $\varepsilon$}\,,
\end{align}
with arbitary constant success probability, using 
\begin{align}
    \text{$\CC^{O}_{\emph{sample}} = \widetilde{\OC}\left( \frac{\|O\|^2}{\varepsilon^2\gamma^4}\right)$ circuit runs}\,,
\end{align}
each of the form in Figure \ref{fig:circuits}(b), each using one instance of $\kt{\psi_0}$ with
\begin{align}
    \text{$\CC^{O}_{\emph{gate}} = \widetilde{\OC}\left(\frac{\lambda^2}{\Delta^2} \right)$ gate depth.}
\end{align}
Moreover, the assumption of the positivity of the spectrum can be dropped in place of a lower bound on the ground state energy, with no change to the computational complexity of the algorithm. 

\end{corollary}

We provide a proof of this corollary, including the extension to non-positive spectra, and accounting of polylogarithmic or constant factor contributions in Appendix \ref{sec:appdx-gs}. Similar to our algorithms for linear systems, we emphasize that we do not ask for any quantum oracle or quantum subroutines encoding information about $H$. Additionally, if the trial state $\kt{\psi_0}$ is accessible as a classical description of vector amplitudes, or from a linear combination of gates, then the trial state can be prepared statistically as part of the randomized algorithm, with minimal depth requirements (see Proposition \ref{prop:statistical-encoding} and Appendix \ref{appdx:randomized-encoding}). Our algorithm is particularly efficient when the Pauli weight $\lambda$ is small, which can be the case for certain physically motivated problems. We discuss more detailed analysis of the complexities in Section \ref{sec:gs-comparison}. Finally, we remark that if a trial state is instead provided with bounded positive overlap $\bar{\gamma}:= \braket{\psi_0}{E_0}$, then quantities of the form $\braket{\phi}{E_0}$ (given preparation of $\ketbra{\phi}{\phi}$) can be approximated with a reduced $\bar{\gamma}^{-2}$ dependence in the sample complexity.


\subsection{Comparison}\label{sec:gs-comparison}

\begin{center}
\begin{table*}[t]
\centering
\resizebox{\textwidth}{!}{%
{\renewcommand{\arraystretch}{1.6}%
    \begin{tabular}{ | P{5.8em} !{\vrule width 1.2pt} P{6.2em} | P{5.7em} | P{7em} | P{9em} | P{7.5em}  | P{5.5em} |}
    \hline\hline   
    & Quantum Oracle & GSE/Gap Assumptions & Data Access Space & Ancillary Algorithmic Space &  Quantum Gate Depth & Runtime  \\ \thickhline
    \multirow{2}{5.8em}{\centering  LCU \cite{keen2021quantum}} & $1^{o}$ Trotter time evol. & \multirow{1}{5em}{\centering  (B1)\vspace{-3.2em}} & $\OC(L)$ bits & \multirow{2}{9em}{\centering $\OC\left(\log\big(\frac{1}{\Delta}\log\frac{1}{\gamma \varepsilon^2}\big) \right)$ qubits \vspace{-2em} } & $\widetilde{\OC}\big( \frac{L^3\Lambda^2}{\gamma\Delta^2}\frac{1}{\varepsilon} \big)$ & $\widetilde{\OC}\big( \frac{L^3\Lambda^2}{\gamma\Delta^2}\frac{1}{\varepsilon^3} \big)$  \\ \cline{2-2}\cline{4-4}\cline{6-7} & Time evol.~via qubitization \cite{low2019hamiltonian} & & $\OC(L)$ bits , $\OC(\log L)$ qubits & & $\widetilde{\OC}\big( \frac{L\lambda}{\gamma\Delta^2}\log^{3/2}\frac{1}{\varepsilon} \big)$ & $\widetilde{\OC}\big( \frac{L\lambda}{\gamma\Delta^2}\frac{1}{\varepsilon^2} \big)$ \\ \hline
    \multirow{2}{5.8em}{\centering  Near Query Optimal \cite{lin2020near}} & \multirow{2}{6.2em}{\centering Block encoding $H$} & (B2) & \multirow{2}{7em}{\centering $\OC(L)$ bits , $\OC(\log L)$ qubits} & $1$ qubit & \multirow{2}{7.5em}{\centering $\widetilde{\OC}\big( \frac{L\lambda}{\gamma\Delta}\log\frac{1}{\varepsilon} \big)$ } & \multirow{2}{6em}{\centering $\widetilde{\OC}\big( \frac{L\lambda}{\gamma\Delta}\frac{1}{\varepsilon^2}\big)$ } \\ \cline{3-3}\cline{5-5} & & none & & $\OC(\log(\frac{1}{\gamma}))$ qubits &  &  \\ \hline
    \multirow{2}{5.8em}{\centering  Early F.T. \cite{zhang2022computing}} & $1^{o}$ Trotter time evol. & \multirow{1}{5em}{\centering  none\vspace{-3.2em}} & $\OC(L)$ bits & \multirow{2}{9em}{\centering 1 qubit \vspace{-2em} } & $\widetilde{\OC}\big( \frac{L^3\Lambda^2}{\Delta^2}\frac{1}{\varepsilon} \log\frac{1}{\gamma} \big)$ & $\widetilde{\OC}\big( \frac{L^3\Lambda^2}{\gamma^4\Delta^2}\frac{1}{\varepsilon^3} \big)$  \\ \cline{2-2}\cline{4-4}\cline{6-7} & Time evol.~via qubitization \cite{low2019hamiltonian} & & $\OC(L)$ bits , $\OC(\log L)$ qubits & & $\widetilde{\OC}\big( \frac{L\lambda}{\Delta}\log\frac{1}{\gamma\varepsilon} \big)$ & $\widetilde{\OC}\big( \frac{L\lambda}{\gamma^4\Delta}\frac{1}{\varepsilon^2} \big)$ \\ \hline
    This Work (Randomized Quantum) & - & (B1) & $\OC(L)$ bits & $1$ qubit & ${\OC}\big( \frac{\lambda^2}{\Delta^2}\log^2\frac{1}{\gamma\varepsilon} \big)$ & $\widetilde{\OC}\big( \frac{\lambda^2}{\gamma^4 \Delta^2}\frac{1}{\varepsilon^2} \big)$\\ \hline\hline
    \end{tabular}
}
}

    \caption{End-to-end comparison of complexities to approximate the expectation value of an observable with respect to the ground state, up to $\varepsilon$-additive error. These algorithms presume access to a trial state approximating the ground state with overlap $\gamma$; a spectral gap lower bound $\Delta$, and possible further assumptions on the ground state energy (GSE) and spectral gap (B1): given $\mu_1\leq E_0\; s.t.$ $E_0 - \mu_1 \leq {\Delta}/{\sqrt{\log\frac{1}{\gamma \varepsilon}}}$, and (B2): given $\mu_2 \geq E_0 \; s.t.$ $\mu_2 - \Delta/2 \geq E_0$, $\mu_2 + \Delta/2 \leq E_1$. In all cases we assume that the Hamiltonian is given in terms of its Pauli decomposition (see Eq.~\eqref{eq:pauli-decomposition}) with $L$ Pauli terms, largest coefficient magnitude $\Lambda$ and coefficient weight $\lambda$. \label{table:gs}}
\end{table*}
\end{center}

Different approaches to quantum algorithms for ground state preparation have been previously studied in Refs.~\cite{yimin2019faster,lin2020near,keen2021quantum, dong2022ground, kyriienko2020quantum, bespalova2021hamiltonian, zhang2022computing, zeng2021universal}. Recently, Ref.~\cite{zhang2022computing} also established an explicit algorithm for ground state property estimation. In Table \ref{table:gs} we compare the complexities of various algorithms for the task of estimating a given observable with respect to the ground state. Namely, we compare the following approaches:
\begin{itemize}
    \item A linear combination of unitaries (LCU) ground state preparation algorithm \cite{keen2021quantum}, which requires a time evolution oracle. There are many approaches to time evolution in the literature \cite{low2019hamiltonian, childs2021theory, berry2015simulating, berry2015hamiltonian, childs2018toward, childs2019faster, campbell2019random, zeng2022simple}. In Table \ref{table:gs} we detail the complexities for two schemes that are in some sense the two ends of the spectrum: for first-order Trotter time evolution which has no additional quantum space requirements, as well as the quantum-runtime-efficient block encoding approach \cite{low2019hamiltonian}.
    \item The block encoding approaches of Ref.~\cite{lin2020near} to prepare the ground state, which is query optimal in $\gamma, \Delta, \varepsilon$ up to logarithmic factors. We note that Ref.~\cite{lin2020near} gives two approaches; one where there is an a priori bound for the ground state energy and spectral gap, and one where there is no such assumption which uses more ancillary qubits and has equivalent runtime up to dominant order.
    \item The early fault-tolerant approach of Ref.~\cite{zhang2022computing} to sample properties of the ground state, which again requires a time evolution oracle. As a caveat, this approach can only estimate expectation values of unitary observables with respect to the desired state. Non-unitary observables can be taken into account with a block encoding, at the expense of more qubits and complexities augmented by the block encoding factor. For the sake of easy comparison, in Table \ref{table:gs} we assume that the observable is unitary.
    \item Our approach using Corollary \ref{cor:ground-state}, assuming no parallelization. 
\end{itemize}

\vspace{-1mm}

For the ground state preparation algorithms considered in Table \ref{table:gs}, we assume that the observable is measured incoherently, in order to establish the most competitive gate depth with regards to our approach. {We note that as in our discussion in Section \ref{sec:linear-systems-comparison}, coherent estimation of observables yields a faster overall runtime (linear scaling in $\OC(\varepsilon^{-1})$), at the expense of $\OC(\varepsilon^{-1})$ scaling in the gate depth, which can be exponentially worse} \cite{knill2007optimal, wang2019accelerated}. For all algorithms we assume that the Hamiltonian is provided classically in the form of its Pauli coefficients (as in Eq.~\eqref{eq:pauli-decomposition}), and any subsequent resources required for data access are noted. If the total number of Pauli terms $L$ is small, this allows for an efficient block encoding via the \texttt{SELECT} and \texttt{PREPARE} oracles, which is what we consider in Table \ref{table:gs}. We also consider Trotterized time evolution, which gives additional dependencies on commutators between terms \cite{childs2021theory}, which here we bound with the magnitude of the largest Pauli coefficient $\Lambda:=\max_{\ell}a_{\ell}$ for simplicity. {One can also opt for higher order Trotterized evolution, which gives improved scaling in $\varepsilon^{-1}$ but worst scaling in $L$ and $\lambda t$.} 

As with our linear systems algorithm, we focus on the quantum hardware requirements, which consists of the total number of logical qubits required for data access and to run the algorithm, as well as the gate depth. Similar to the linear systems task considered in Section \ref{sec:linear-systems-comparison}, our approach has the smallest space requirement, and there is scope for our algorithm to have competitive gate depth if the Pauli weight $\lambda$ is smaller than number of Pauli terms $L$ {(e.g. see values given in Table IX and X of Ref.~\cite{lee2021even}).} 

Finally, we note that in Table \ref{table:gs} we have not accounted for the runtime to prepare the trial state. Other coherent approaches require coherent calls to the unitary that prepares the trial state. Thus, if the runtime of this subroutine is significant, this would also cause discrepancies in the total runtimes listed.


\section{Gibbs state sampling}\label{sec:gibbs}

\subsection{Randomized quantum algorithm}

We can also consider randomizing prior approaches to Gibbs state preparation. Here the task is usually to recover an approximation to the Gibbs state $e^{-\beta H}/\ZC$ (also known as the thermal state)  given some information about the Hamiltonian $H$, where $\ZC:=\Tr[e^{-\beta H}]$ is the partition function and $\beta$ is a parameter physically corresponding to the inverse temperature. In our setting, we will recover observables with respect to the Gibbs state.

Under standard complexity-theoretic assumptions, this is expected to be hard in general \cite{bravyi2021complexity}. However, one prominent line of work has aimed to emulate the thermalization process, following the intuition that there should exist efficient algorithms for certain physical systems \cite{wocjan2021szegedy, yung2012quantum, temme2011quantum, chen2023quantum} (see Ref.~\cite{chen2021fast} for the current state-of-the-art approach). This leads to complexity dependent on the mixing time, which whilst small for physical systems, may be difficult to estimate or may require additional assumptions to bound. Similar to classical Monte Carlo approaches for classical Hamiltonians, this nevertheless could be a promising approach to achieve efficient practical performance for interesting problems.

In this section we follow a different line of work that aims to approximate the exponential operator directly for generic Hamiltonians \cite{chowdhury2016quantum, gilyen2019quantum, van2020quantum, van2019improvements, tong2021fast, holmes2022quantum}. When applied to one half of the maximally entangled state (with the correct normalization), this corresponds to the purification of the Gibbs state. Reflecting the expected hardness of the general problem, prior approaches have exponential runtime in the inverse temperature $\beta$ and operator norm $\|H\|$. A perturbative approach was proposed in Ref.~\cite{holmes2022quantum} where one assumes access to the purification of the Gibbs state of an intermediate Hamiltonian $H_0$. In this case, the dominant contribution to the complexity instead depends on the operator norm of the perturbation $V:=H-H_0$, which could in principle be much smaller. Additionally, this dependence can be further reduced with the action of a so-called non-equilibrium unitary, though we omit this discussion for simplicity here. We note that, in the trivial case $H_0 = \boldsymbol{0}$, this starting assumption reduces to starting with the maximally entangled state, which is the same as the other aforementioned approaches. Here the relevant function of interest is the exponential function of the work operator, defined as $W:= H \otimes \id-\id\otimes H_0^*$. In this section we import the insights of Ref.~\cite{holmes2022quantum} into our randomized framework which allows us to state the following corollary:

\begin{corollary}[Gibbs state property estimation]\label{cor:gibbs-state}
Suppose access to the Gibbs state of an intermediate Hamiltonian $H_0$ via its purification $\kt{\Psi_0}$ and the ability to measure the observable $O$. Further, assume the Pauli decompositions of $H_0$ and $H$ are known as in Eq.~\ref{eq:pauli-decomposition}, and $[H_0,H]=0$. Then, we give a randomized quantum algorithm to
\begin{equation}\label{eq:gibbs-task}
    \text{approximate $\frac{\Tr[e^{-\beta H} O]}{\ZC}$ to additive error $\varepsilon$}\,,
\end{equation}
with arbitrary constant success probability, utilizing
\begin{equation}
    \text{$\CC^{O}_{\emph{sample}} =\widetilde{\OC}\left(   e^{2\beta\|V\|} \frac{\ZC_0^2}{\ZC^2} \frac{e^{\sqrt{\ln (\|O\|/\varepsilon)} }}{\varepsilon^2} \right)$ circuit runs\,,}
\end{equation}
where $\ZC_0:=\Tr[e^{-\beta H_0}]$ and $\ZC:=\Tr[e^{-\beta H}]$ are the partition functions of $H_0$ and $H$ respectively, and where we denote $V=H-H_0$. Each circuit is of the form in Figure \ref{fig:circuits}(b), and uses one instance of $\kt{\Psi_0}$ with at most
\begin{align}
    \text{$\CC^{O}_{\emph{gate}} = \widetilde{\OC}\left(\beta^3 \lambda_W^2\|W\| \right)$ gate depth\,,}
\end{align}
where $\lambda_W$ is the Pauli weight of the work operator $W:= H \otimes \id-\id\otimes H_0^*$.
\end{corollary}

Similar to before, these complexities can be established via Proposition \ref{prop:fourier-sampling} and Proposition \ref{prop:norm-sampling}, along with error bounds established in Ref.~\cite{holmes2022quantum}. We refer the reader to Appendix \ref{appdx:gibbs} for details. 

\subsection{Comparison}

In Table \ref{table:gibbs} we compare the complexities of our randomized approach against other algorithms for the task in Eq.~\eqref{eq:gibbs-task} with an operator $O$ with $\|O\|\leq 1$. The other approaches we consider are as follows:

\begin{itemize}
    \item An LCU approach based on the Hubbard-Stratonovich transform where one requires access to time evolution of an operator $\widetilde{H}$ which, when squared, recovers the action of ${H}$ conditioned on an ancillary register. We detail resources required to obtain a block encoding of $\widetilde{H}$ and using a qubitization approach for the time evolution \cite{low2019hamiltonian}. Specifically, this will incur runtime overhead in the number of Pauli terms $L_H$ and $\lambda_{\widetilde{H}}:=\sum_{\ell=1}^{L_H}\sqrt{a_{\ell}}$, where $a_{\ell}$ are the Pauli coefficients of $H$
    \item A QSVT approach to implement polynomial approximations of the  Gaussian and exponential function. We consider block encodings of $\widetilde{H}$ and ${H}$ for these approaches respectively, which again incurs overhead in $L_H$ and $\lambda_{\widetilde{H}}$, or  $\lambda_H$ which denote as the Pauli weight of $H$.
    \item The perturbutive approach of Ref.~\cite{holmes2022quantum} starting from the purified Gibbs state of an intermediate Hamiltonian $H_0$. An LCU approximation of $e^{-\beta W/2}$ is implemented. As with the above approaches, the exponentially costly step comes from the number of rounds of amplitude amplification required. As we compare complexities with our randomized version of this algorithm, we add a second row in Table \ref{table:gibbs} to demonstrate complexities using incoherent sampling rather than amplitude amplification, in which case the required gate depth is greatly reduced. Here, the runtimes are dependent on the number of Pauli terms constituting $W$ and its Pauli weight, which we denote as $L_W$ and $\lambda_W$ respectively.
    \item Our approach, given in Corollary \ref{cor:gibbs-state}, based on the insights of Ref.~\cite{holmes2022quantum}.
\end{itemize}

\begin{table*}[t]
\centering
{\renewcommand{\arraystretch}{2}%
\resizebox{\textwidth}{!}{%
    \begin{tabular}{ | P{5.8em} | P{5.6em} !{\vrule width 1.2pt} P{6em} | P{6em} | P{17.6em}  | P{15.6em} |}
    \hline\hline Algorithm & Access Model
    & Data Access Space & Algorithmic Space &  Quantum Gate Depth & Runtime  \\ \thickhline
    LCU \cite{chowdhury2016quantum} & Time evol.~$\widetilde{H}$ via qubitization \cite{low2019hamiltonian} & $\OC(L_H)$ bits , $\OC(\log L_H)$ qubits
    & $2n+$ $\OC(\log \beta\|H\|+\log\log\frac{1}{\varepsilon} ))$ qubits
    &  $\widetilde{\OC} \left( L_H \lambda_{\widetilde{H}}    \sqrt{\beta}\sqrt{\frac{2^n}{\ZC}} \log \frac{1}{\varepsilon} \right)$
    & $\widetilde{\OC} \left( L_H \lambda_{\widetilde{H}}    \sqrt{\beta}\sqrt{\frac{N}{\ZC}} \frac{1}{\varepsilon^2} \right)$
    \\ \hline
    \multirow{2}{6em}{\centering QSVT \cite{gilyen2019quantum}} & Block encoding $\widetilde{H}$  & $\OC(L_H)$ bits , $\OC(\log L_H)$ qubits
    & $2n+ 2$ qubits
    &  $\widetilde{\OC} \left( L_H \lambda_{\widetilde{H}}    \sqrt{\beta}\sqrt{\frac{2^n}{\ZC}} \log \frac{1}{\varepsilon} \right)$
    & $\widetilde{\OC} \left( L_H \lambda_{\widetilde{H}}    \sqrt{\beta}\sqrt{\frac{2^n}{\ZC}} \frac{1}{\varepsilon^2} \right)$
    \\ \cline{2-6} & Block encoding ${H}$  & $\OC(L_H)$ bits , $\OC(\log L_H)$ qubits
    & $2n+ 2$ qubits
    &  $\widetilde{\OC} \left( L_H \lambda_{{H}}    e^{\beta/2}\sqrt{\frac{2^n}{\ZC}} \log \frac{1}{\varepsilon} \right)$
    & $\widetilde{\OC} \left( L_H \lambda_{{H}}    e^{\beta/2}\sqrt{\frac{2^n}{\ZC}} \frac{1}{\varepsilon^2} \right)$
    \\ \hline
    \multirow{2}{6em}{\centering  Work operator LCU \cite{costa2021optimal}} & \multirow{2}{5.8em}{\centering Block encoding $W$} & \multirow{2}{6em}{\centering $\OC(L_W)$ bits , $\OC(\log L_W)$ qubits} & \multirow{2}{6em}{\centering $2n+3$ qubits} & $\widetilde{\OC} \left( L_W \lambda_W \sqrt{\|W\|}  \sqrt{\frac{\ZC_0}{\ZC}} e^{\beta\|V\|/2} e^{\sqrt{\log 1/\varepsilon}} \right)$ & $\widetilde{\OC} \left( L_W \lambda_W \sqrt{\|W\|} \sqrt{\frac{\ZC_0}{\ZC}} e^{\beta\|V\|/2} \frac{1}{\varepsilon^2} \right)$ \\
    \cline{5-6} & & & & $\widetilde{\OC} \left( L_W  \sqrt{\|W\|}\beta^{3/2} \left(\lambda_W + \|W\|\log\frac{1}{\varepsilon} \right)  \right)$ & $\widetilde{\OC} \left(L_W \lambda_W \sqrt{\|W\|} \frac{\ZC_0}{\ZC} e^{\beta\|V\|} \frac{1}{\varepsilon^2}\right)$ \\
    \hline
    Randomized Quantum (This Work) & Pauli coefficients of $W$ & $\OC(L_W)$ bits & $2n+1$ qubits  &  $\widetilde{\OC}\left( \lambda_W^2\|W\|\beta^3 + \log\frac{1}{\varepsilon} \right)$  & $\widetilde{\OC}\left( \lambda_W^2 \|W\| \frac{\ZC_0^2}{\ZC^2}     e^{2\beta\|V\|} \frac{1 }{\varepsilon^2} \right)$  \\ \hline\hline
    \end{tabular}
}
    \caption{End-to-end complexities of estimating the expectation value of an operator $O\,;\|O\|=\OC(1)$ to additive error $\varepsilon$ with respect to the Gibbs state of Hamiltonian $H$ on a system of $n$ qubits. {In the cells we have omitted subpolynomial factors to simplify expressions.} Certain approaches require access to a matrix $\widetilde{H}$ which recovers the action of $\sqrt{H}$ and has weight $\lambda_{\widetilde{H}}$ defined in the main text which satisfies $\sqrt{\lambda_H} \leq \lambda_{\widetilde{H}} \leq \sqrt{L_H\lambda_H}$. The pertubtive approach assumes access to a purification of the Gibbs state of an intermediate Hamiltonian $H_0$ with complexity dependent on the perturbation $V:=H-H_0$ and the work operator $W$ defined in the main text which satisfies $\|W\| \leq \|H_0\| + \|H\|$, $L_W = L_{H_0} + L_H$, $\lambda_W = \lambda_{H_0} + \lambda_H$. {For the algorithm of Ref.~\cite{holmes2022quantum} we have included two rows: first, the more runtime-efficient approach which uses amplitude amplification, and the measure-until-success approach which uses exponentially less gate depth in $\beta$.}} \label{table:gibbs}}
\end{table*}

Similar to our ground state property estimation comparison, in all the above settings we have considered incoherent sampling of the observable $O$ in order to give the most competitive gate depth. In all cases, we assume classical access to the Pauli decomposition of $H$ (and $H_0$) as a starting point. From then on, the classical and quantum resources required to process this data are recorded in the table. Specifically, all block encodings are presumed to be constructed via \texttt{SELECT} and \texttt{PREPARE} oracles. The weight of $\widetilde{H}$ can be found to satisfy $\sqrt{\lambda_H} \leq \lambda_{\widetilde{H}} \leq \sqrt{L_H\lambda_H}$. We also remark that $\|W\|\leq \|H\| + \|H_0\|$, and in the trivial case $H_0 = \boldsymbol{0}$ we have $\|W\| = \|V\| = \|H\|$. Additionally, the number of Pauli terms and Pauli weight simply follows as $L_W = L_{H_0} + L_H$, $\lambda_W = \lambda_{H_0} + \lambda_H$.  

In Table \ref{table:gibbs} we see that our randomized approach offloads the exponential complexity in $\beta\|V\|$ from the gate depth onto the sample overhead, compared to the fully coherent approach. However, compared to replacing amplitude amplification with incoherent sampling, our approach does not incur any explicit dependence on $L_W$ in the gate complexity or overall runtime. In addition, it uses fewer qubits than any other approach. Finally, we recall that in our framework our stated sample complexities serve as sufficient conditions, and in practice an $\varepsilon$-approximation of the observable could be achieved with a smaller circuit sample count. This could be especially relevant here, where our required gate depths are only polynomially large, in contrast to other approaches with exponentially large gate depth. Thus, given an efficient verifier, in our setting one could potentially collect samples from an efficient circuit and achieve convergence much faster than the stated runtime bounds.


\section{Application: Estimation of Green's functions}\label{sec:greens}

So far we have seen that our algorithms are naturally suited to settings where the matrix of interest is already given in the Pauli basis. Thus, our algorithms are suited to physically motivated problems. One natural application is the evaluation of single particle Green's functions in the context of many-body physics. Green's functions can be used to calculate single-particle expectation values such as the kinetic energy, as well as to determine the many-body density of states. For more detailed background we refer the reader to Refs.~\cite{giuliani2005quantum, fetter2012quantum}. Previous works have proposed quantum approaches for preparing Green's functions in both the frequency and time domains \cite{tong2021fast,keen2021quantum, bauer2016hybrid, wecker2015solving, kreula2016non, kreula2016few, rungger2019dynamical,endo2020calculation, kosugi2020construction}. One particular idea has been to use a quantum algorithm to evaluate Green's functions for the computationally expensive  subroutine of the quantum impurity problem in Dynamical Mean Field Theory (DMFT) calculations, in order to potentially extend their scope \cite{bauer2016hybrid, wecker2015solving, kreula2016non, kreula2016few}.

We define the advanced and retarded Green's function in the frequency domain (denoted as $G^{(+)}(\omega)$ and $G^{(-)}(\omega)$ respectively) as the matrix-valued functions with elements
\begin{align}
    G_{ij}^{(+)}(\omega) &:= \bbr{E_0}\hat{a}_i \left(\hbar\omega - (H-E_0) + i\eta \right)^{-1} \hat{a}^{\dag}_j \kkt{E_0}\,,\label{eq:greens1} \\
    G_{ij}^{(-)}(\omega) &:= \bbr{E_0}\hat{a}^{\dag}_i \left(\hbar\omega + (H-E_0) - i\eta \right)^{-1} \hat{a}_j \kkt{E_0}\,, \label{eq:greens2}
\end{align}
where $E_0$ is the ground state energy of $H$, $\eta$ is a broadening parameter that determines the resolution of the Green's function, and $\hat{a}^{\dag}_i,\hat{a}_i$ are Fermionic single-particle creation and annihilation operators. We note that these quantities are expectation values of an operator which contains the inverse of some matrix, where the expectation value is taken with respect to the ground state. Thus, our earlier results in Sections \ref{sec:linear-systems} and \ref{sec:gs} can be readily applied. 
The creation and annihilation operators can be expressed using Pauli operators via the Jordan-Wigner transformation
\begin{align}
    &\hat{a}_i=Z^{\otimes(i-1)} \otimes \frac{1}{2}(X+\mathrm{i} Y) \otimes I^{\otimes(N-i)}\,, \\ &\hat{a}_i^{\dagger}=Z^{\otimes(i-1)} \otimes \frac{1}{2}(X-\mathrm{i} Y) \otimes I^{\otimes(N-i)}\,.
\end{align}

We now present our result for Green's function estimation. From hereon we denote $\Gamma^{(+)}=\hbar\omega - (H-E_0) + i\eta $ and $\Gamma^{(-)}=\hbar\omega + (H-E_0) - i\eta $. We remark that whilst it is possible to use the algorithms of Section \ref{sec:linear-systems} and \ref{sec:gs} as separate subroutines to evaluate the Green's functions, it is beneficial to compile the ground state projection and matrix inversion all at once. This is the scheme we present in the following proposition. 

\begin{proposition}[Green's function estimation]\label{cor:greens-functions}
Consider a Hamiltonian $H=\sum_l E_l \ketbra{E_l}{E_l}$ with all eigenvalues $E_l \geq 0$, and known Pauli decomposition $H=\sum_{\ell}a_{\ell}P_{\ell}\,;\; \lambda_H:=\sum_{\ell}|a_{\ell}|$. Assume the spectral gap is lower bounded as $\Delta \leq E_1 - E_0$. Additionally, we suppose that we can freely prepare an initial trial state $\kt{\psi_0}$ with overlap with the ground state $\gamma:= |\braket{\psi_0}{E_0}|$. Given parameters $\omega, \eta$ and the ground state energy $E_0$, we give a random compiler to
\begin{align}
    \text{approximate $G_{ij}^{(+)}(\omega)$ \& $G_{ij}^{(-)}(\omega)$ up to additive error $\varepsilon$\,, }
\end{align}
and arbitrary constant success probability, each utilizing 
\begin{align}
    \text{$\widetilde{\OC}\left( \frac{\|(\Gamma^{(\pm)})^{-1}\|^2}{\gamma^4\varepsilon^2} \right)$ circuit runs\,, }
\end{align}
respectively, each consisting of at most
\begin{align}
   \text{$\widetilde{\OC}\left(\frac{\lambda_H^2}{\Delta^2} +  (|\hbar\omega \pm E_0|+ \lambda_{H}+\eta)^2\|\Gamma^{(\pm)-1}\|^2 \right)$ gate depth\,.} 
\end{align}
\end{proposition}

We provide a proof of this result in Appendix \ref{appdx:greens}. As with our previous algorithms, the scheme we use has an advantage over other algorithms in that it does not use any additional ancillary qubits and it does not have any explicit dependence on the number of Pauli terms $L$. We remark that in the algorithm of Proposition \ref{cor:appdx-greens-functions}, the ground state energy is given exactly as an input. If this is not available, the ground state energy may be approximated via the techniques of other early fault-tolerant schemes \cite{wan2021randomized, wang2022quantum, wang2023faster, lin2022heisenberg, dong2022ground}. Moreover, if the ground state energy is approximated to sufficiently small precision, error contribution to the Green's functions can be constrained (see Appendix \ref{appdx:greens} for more details). Finally, we remark that Green's functions in the (real) time domain consist of expectation values of time-evolved creation and annihilation operators. Thus, again using the Jordan-Wigner transformation this can be directly evaluated via Theorem \ref{thm:general-sampling} and the tools of our ground state property estimation algorithm, with the same asymptotic complexities as stated in Corollary \ref{cor:ground-state}. 


\section{Outlook}\label{sec:discussion}

We presented a framework for sampling properties of general matrix Fourier series and applied this to give explicit algorithms for the linear systems problem and the ground state property estimation problem. By starting with a (classical) description of the matrix in the Pauli basis, we circumvent the need for coherent data structures, which adds to the hardware burden of other quantum algorithms. Another distinct feature of our approach is that there is no explicit dimension or sparsity dependence in our complexities; instead, the runtime depends on the norm of the Pauli coefficients for the matrix, which in principle can be much smaller than number of Pauli terms or the dimension of the system. As such, our framework is particularly suited to physically-motivated matrices, where the Pauli description is readily available and of low weight.

There are immediate open questions that have yet to be explored:
\begin{itemize}
    \item Is it possible for certain special classes of matrices to efficiently obtain the Pauli decomposition starting from a description in the computational basis? If the number of Pauli terms is known and small, up to a quadratic saving in the dimension can be made compared to the na\"ive approach, though this would still present a barrier to any possible superpolynomial quantum runtime advantage compared to classical schemes directly working in the computational basis. The question of obtaining the Pauli decomposition also has implications for near term schemes for linear systems that start with similar assumptions \cite{bravo2019variational, xu2019variational, huang2019near}. More broadly, this also leads into the question of what data structures or data sources are amenable to possible quantum speedups.

    \item { We have considered worst-case performance in two senses. First, bounds on sample complexities are constructed essentially by bounding the variances of estimators. Unlike some other approaches, we can halt our algorithm at any time if the solution is good enough with the obtained samples. It remains open whether there exists an efficient verifier for which solution quality can be checked, which could enable faster time to solution heuristically. Second, we have quoted maximum non-Clifford gate complexities, whereas in reality in our algorithm with each sample we implement circuits with varying gate complexities. Thus, a more refined measure of unparallelized total runtime should instead use the \textit{expected} gate complexity, which can lead to significant savings in asymptotic runtime bounds \cite{tosta2023randomized}. We leave it as an open question as to how far this quantity can be optimized for the problems of interest discussed in this work. }

    \item  A more precise numerical analysis of finite resource costs for concrete problems (such as the calculation of Green's functions for a particular problem of interest) would be illuminating for the feasibility of our schemes in the early fault-tolerant regime. {One can also explore numerically-obtained Fourier approximations to functions. }


    \item Our linear systems algorithm could possibly undergo further refinements, inspired by classical algorithms. For instance, it is general practice to use preconditioners for linear systems solvers to effectively reduce the effect of the condition number on the runtime (e.g. see Refs.~\cite{demmel1997applied, golub2013matrix, saad2003iterative}), and analogous quantum preconditioners have been studied \cite{clader2013preconditioned, shao2018quantum, tong2021fast}. We leave it as future work to investigate whether such techniques and beyond can be efficiently transported to the early fault-tolerant setting. 

    \item We discussed in Section \ref{sec:fourier-sampling} various properties of the Fourier series that influences the complexities of our algorithms. In particular, in order to constrain the sample complexity the variance of the randomized schemes should be constrained. It remains to be seen whether there are other functions of interest beyond the inverse function and Gaussian /exponential functions which have Fourier series approximations which can lead to algorithms with favorable complexities.
    
\end{itemize}


\textit{Note added.}  Concurrently with our work, 
Ref.~\cite{chakraborty2022implementing} of Chakraborty proposed a randomized scheme where starting from a Fourier decomposition of a function, one constructs a sampling algorithm which queries a controlled time evolution oracle.

\section*{Acknowledgments}

This project was started when MB and SW were additionally affiliated with the AWS Center for Quantum Computing, Pasadena, USA. We thank Earl Campbell for collaboration in the early stages of this project. We thank Leandro Aolita,
Fernando Brandao, Gian Camilo, Chi-Fang (Anthony) Chen, Alexander M.~Dalzell, Allan David Cony Tosta, Steven T.~Flammia, Zo\"e Holmes, Thais de Lima Silva, Thiago O.~Maciel, and Kianna Wan for discussions (in alphabetical order). SW and MB are supported by the EPSRC Grant number EP/W032643/1. SW is additionally supported by the Samsung GRP grant.


\printbibliography


\section*{Appendices}
\addcontentsline{toc}{section}{\protect\numberline{}Appendices}
\appendix

\setcounter{section}{0}
\setcounter{proposition}{0}
\makeatletter
\@addtoreset{theorem}{section}
\@addtoreset{corollary}{section}
\makeatother


In these appendices we present detailed statements for our theoretical results, as well as the proofs thereof.

In Appendix \ref{appdx:fourier-sampling} we introduce the main technical results of the paper. Specifically, in Appendix \ref{appdx:fourier-sampling-prop} we first show how to prove the result in the warm-up problem (Proposition \ref{prop:fourier-sampling}) for sampling properties of a given Fourier series, before proving our main result (Theorem \ref{thm:general-sampling}) for sampling properties of matrix functions in Appendix \ref{appdx:thm1}.  In Appendix \ref{appdx:norm-sampling} we discuss the complexities of evaluating normalized quantities that correspond to normalized quantum states. Next, in Appendix \ref{appdx:randomized-encoding} we discuss how to statistically encode classical vectors as part of our randomized scheme (Proposition \ref{prop:statistical-encoding}). We then discuss the classical power of our randomized scheme for sampling properties of low-degree polynomials in Appendix \ref{appdx:classical} (Proposition \ref{prop:classical-simp}). 

In Appendix \ref{appdx:applications} we demonstrate how our algorithms for linear systems (Corollary \ref{cor:linear-systems}), ground state property estimation (Corollary \ref{cor:ground-state}), and Gibbs state property estimation (Corrolary \ref{cor:gibbs-state}) follow from our main results. We also show how combining the linear systems algorithm with the ground state property estimation algorithm allow for a scheme to evaluate Green's functions in many body physics (Proposition \ref{cor:greens-functions}).  


\section{Fourier sampling}\label{appdx:fourier-sampling}

\subsection{Sampling from a given Fourier series - proof of Proposition \ref{prop:fourier-sampling}}
\label{appdx:fourier-sampling-prop}

We start by introducing the random compiler lemma of Ref.~\cite[{Lemma 2}]{wan2021randomized} which demonstrates how to decompose time evolution operators into Pauli gates and Pauli rotations.

\begin{lemma}[\textit{Random compiler lemma} - adapted from Lemma 2 of Ref.~\cite{wan2021randomized}]\label{lem:appdx-kianna-ham-simulation}
Let ${A}=\sum_{\ell} a_{\ell} P_{\ell}$ be a Hermitian operator that is specified as a linear combination of Pauli operators with Pauli weight $\lambda :=\sum_{\ell} |a_{\ell}|$ and real coefficients $a_{\ell}\in \mathbb{R}$. For any $t \in \mathbb{R}$ and any choice of $r \in \mathbb{N}:=\{1,2, \ldots\}$, there exists a linear decomposition
\begin{align}\label{eq:appdx-kianna-ham-simulation}
    e^{i {A} t}=\sum_{m \in T} \beta^{(r)}_m U^{(r)}_m\,,
\end{align}
for some index set $T$, unitaries $\{U^{(r)}_m\}_m$, and real numbers $\{\beta_{m}^{(r)}\}_m$ such that $\sum_{m \in T}\beta_{m}^{(r)} \leq \exp(\lambda^2 t^2/r)$. For all $m \in T$, the non-Clifford cost of controlled-$U^{(r)}_m$ is that of $r$ controlled single-qubit Pauli rotations.
\end{lemma}

\begin{proof}
We note that Ref.~\cite{wan2021randomized} provided an explicit proof for operators with normalized Pauli weight and positive coefficients (i.e. the coefficients are probabilities). We leave the core proof to the Appendix C of Ref.~\cite{wan2021randomized}, and for completeness here we explicitly demonstrate how the proof extends to the more general setting.

First we consider the scenario where $\lambda = 1$ but $a_{\ell}$ may not necessarily be positive. Ref.~\cite{wan2021randomized} expresses an $r^{th}$ of the time evolution operator as
\begin{align}\label{eq:appdx-kianna-ham-simulation-2}
    e^{iAt/r} &= \sum_{n\ \text{even}} \gamma_n(t/r) \sum_{\ell_1, \ldots, \ell_n, \ell^{\prime}} a_{\ell_1} \ldots a_{\ell_n} a_{\ell^{\prime}} \, P_{\ell_1} \ldots P_{\ell_n} \exp\left(i \theta_n P_{\ell^{\prime}}\right) \\
    &= \sum_{n\ \text{even}} \gamma_n(t/r) \sum_{\ell_1, \ldots, \ell_n, \ell^{\prime}} |a|_{\ell_1} \ldots |a|_{\ell_n} |a|_{\ell^{\prime}} \cdot \textrm{sgn}\left(a_{\ell_1} \ldots a_{\ell_n} a_{\ell^{\prime}}\right) \, P_{\ell_1} \ldots P_{\ell_n} \exp\left(i \theta_n P_{\ell^{\prime}}\right) \,,
\end{align}
where $\gamma_n(t/r)$ are coefficients satisfying $\sum_{(n\, \text{even})} \gamma_n(t/r) \leq \exp (t^2/r^2)$. In the second line we see that we can proceed as if the coefficients were positive (probabilities), by absorbing their sign into the string of unitaries to implement $P_{\ell_1} \ldots P_{\ell_n} \exp\left(i \theta_n P_{\ell^{\prime}}\right)$. By considering the product of such fractional time evolution operators $(e^{iAt/r})^r$ we can then sample strings of Pauli operators and $r$ Pauli rotations (with an absorbed phase), with total weight $\leq (\exp (t^2/r^2))^r = \exp (t^2/r)$. This consists of $r$ pairs of controlled (multi-qubit) Pauli rotations and strings of controlled Pauli gates, whose number is in theory unbounded. We remark, however, that the number of Pauli gates is with high probability zero with exponentially decaying probability for increasing gate number \cite{wan2021randomized}. Further these gates can be efficiently compiled together classically into a single controlled (multi-qubit) Pauli gate using standard Pauli product rules.

Now we can consider how to deal with non-unit Pauli weight. For $H$ with non-unit Pauli weight $\lambda$ we can simply consider the equality $e^{iAt} = e^{i\hat{A}\lambda t}$, where $\hat{A} = A/\lambda$ now has unit Pauli weight. As Lemma \ref{lem:appdx-kianna-ham-simulation} holds for all $t\in\mathbb{R}$, the steps of the original proof can follow, by considering an extended time parameter $t\rightarrow \lambda t$.
\end{proof}

\begin{algorithm}[t]
\caption{Fourier sampling of $\Tr\left[s(A)\rho s(A)^{\dag} O\right]$}\label{alg:observable}
\raggedright
\vspace{2pt}
    \textbf{Problem input:}
    \vspace{-2pt}
    \begin{itemize}
        \item $N \times N$ Hermitian matrix ${A}=\sum_{\ell} a_{\ell} P_{\ell}$, with $a_{\ell}\in \mathbb{R}\ \forall \ell$, and known Pauli weight $\lambda = \sum_{\ell} |a_{\ell}|$
        \item Fourier parameters $\{\alpha_k\}_{k \in F}$, $\{t_k\}_{k \in F}$ for series $s(A)$ according to Eq.~\eqref{eq:fourier-A}
        \item Input state $\rho$ and measurement operator $O$.
    \end{itemize}
    \vspace{2pt}
    \textbf{Parameters:} (1) Approximation parameter $\varepsilon$;\; (2) Success probability $1-\delta$;\; (3) Runtime vector $\vec{r}\in\mathbb{N}^{|F|}$. \\ \vspace{5pt}
    \textbf{Output:} Approximation to $\Tr\left[s(A)\rho s(A)^{\dag} O\right]$ with additive error at most ${\varepsilon}$ and probability at least $1-\delta$.  
    \begin{enumerate}
    \itemsep -0.8\parsep
        \color{black!100}
        \item Compute coefficients $\{\beta_{km}^{(\vec{r})} \}_{(k,m) \in F\times T}$ as in Eq.~\eqref{eq:appdx-kianna-ham-simulation} for each time parameter $\{t_k \}_{k \in F}$.
        \item $R(\vec{r}) \leftarrow \sum_{(k,m) \in F\times T} \big|\alpha_k \beta_{km}^{(r_k)} \big|$.
        \item $M(\vec{r}) \leftarrow \left\lceil 4\ln (2 / \delta)\left({R(\vec{r})}/{\varepsilon}\right)^{2}\right\rceil$. 
        \item \textbf{for} $j \in \big[M(\vec{r})\big]:$
        \vspace{-2pt}
        \begin{enumerate}[label=(\roman*)]
            \item Sample two strings of gates independently as: \vspace{-3pt}
            \begin{enumerate}[(a)]
            \setlength{\itemsep}{0pt}
                \color{black!100}
                \item  Sample indices $(k',m')$ according to probability distribution $\{|\alpha_k \beta_{km}^{(r_k)} |/R(\vec{r})\}_k$ and select corresponding string of gates $U^{(r_{k'})}_{k'm'}$ 
                \color{black}
                \item Repeat with new independent sample $(k'',m'')$, select $U^{(r_{k''})}_{k'm'}$\,.
            \end{enumerate}
            \item Run $n+1$ qubit circuit in Figure \ref{fig:circuits}(b) with input state $\dya{0}\otimes \rho$,  controlled $U^{(r_{k'})}_{k'm'}$ and anti-controlled $U^{(r_{k''})}_{k''m''}$, and measurement $X\otimes O$. Record measurement outcome $o_j$.
            \item $z_j \leftarrow R(\vec{r})^2o_j $.
        \end{enumerate}
        \item  \textbf{end for} 
        \item $\overline{z}^{(M)} \leftarrow \sum_j z_j/M(\vec{r})$. \textbf{return} $\overline{z}^{(M)}$. 
        \color{black}
    \end{enumerate}
\end{algorithm}

With the above Lemma, we can now demonstrate how to solve our warm-up problem of sampling properties of a given Fourier series. In the following we present a more precise version of Proposition \ref{prop:fourier-sampling} from the main text.

\begin{proposition}[Sampling from Fourier series -- detailed version]\label{prop:appdx-fourier-sampling}
Suppose that we have a Fourier series 
\begin{equation}\label{eq:appdx-fourier-A}
    s(A) = \sum_{k\in F} \alpha_k \exp{(it_kA)}\,,
\end{equation}
in some $N\times N$ Hermitian matrix $A$, and denote the $\ell_1$-norm of the coefficients as $\alpha:=\sum_{k\in F} |\alpha_k|$. Suppose further that $A$ has known Pauli decomposition $A=\sum_{\ell} a_{\ell} P_{\ell}$ with Pauli weight $\lambda = \sum_{\ell} |a_{\ell}|$. Then, 
\begin{enumerate}[(a).]
    \item 
    Given a procedure to prepare the pure states $\kt{\psi}$, $\kt{\phi}$ with respective unitaries $U_{\psi}$, $U_{\phi}$ with respective gate depths $d_{\psi}$, $d_{\phi}$, we have a randomized quantum algorithm that uses $\log(N)+1$ qubits to approximate $\br{\phi} s(A)\kt{\psi}$ up to additive error $\varepsilon$ with probability at least $1-\delta$, using Algorithm \ref{alg:overlap}  with
    \begin{align}       {\CC}^{\phi}_{\emph{sample}} &=  \OC\left( \log \left(\frac{2}{\delta}\right)\frac{\alpha^2}{\varepsilon^2} \right)\,,
    \end{align} 
    circuit samples, where each circuit takes the form in Figure \ref{fig:circuits}(a) and has depth 
    \begin{align}
        {\CC}^{\phi}_{\emph{gate}} &= \OC\left(\lambda^2 t_{max}^2 + d_{\psi} + d_{\phi}\right)\,,
    \end{align}
    where we denote $t_{max} = \max_{k\in F}t_k$. 
    \item Given a procedure to prepare the quantum state $\rho$in depth $d_{\rho}$ and perform measurements with measurement operator $O$, we give a randomized quantum algorithm that uses $\log(N)+1$ qubits to approximate $\Tr\left[s(A)\rho s(A)^{\dag} O\right]$ up to additive error $\varepsilon$ with probability at least $1-\delta$ using Algorithm \ref{alg:observable} with 
    \begin{align}
        {\CC}^{O}_{\emph{sample}} &=  \OC\left( \log \left(\frac{2}{\delta}\right)\frac{\|O\|^2\alpha^4}{\varepsilon^2} \right)\,, 
    \end{align}
    circuit samples, where each circuit takes the form in Figure \ref{fig:circuits}(b), and has depth 
    \begin{align}
        {\CC}^{\phi}_{\emph{gate}} &= \OC\left(\lambda^2 t_{max}^2 + d_{\psi} + d_{\phi}\right)\,.
    \end{align}
\end{enumerate}
\end{proposition}

\begin{proof}[Proof of Proposition \ref{prop:appdx-fourier-sampling}]
We note that from Lemma \ref{lem:appdx-kianna-ham-simulation}, $s(\varepsilon,A)$ can be decomposed via a linear combination as 
\begin{align}
    s(A) = \sum_{k\in F} \alpha_k \exp{(it_kA)} = \sum_{(k,m) \in F\times T} \alpha_k \beta_{km}^{(r_k)}U^{(r_k)}_{km} \,,
\end{align}
where $\sum_{m \in T}\beta_{km}^{(r_k)} \leq \exp(\lambda^2t_k^2/r_k)$ and $U^{(r_k)}_{km}$ consists of $r_k$ non-Clifford operations. The above expression can further be seen as a quantity that is proportional to a probability distribution over unitaries 
\begin{equation}\label{eq:appdx-s-decomposition}
    s(A) = R(\vec{r}) \sum_{(k,m) \in F\times T} \frac{|\alpha_k \beta_{km}^{(r_k)}|}{R(\vec{r})} \widetilde{U}^{(r_k)}_{km} = {R(\vec{r})} \sum_{(k,m) \in F\times T}  p^{(r_k)}_{km} \widetilde{U}^{(r_k)}_{km}\,,
\end{equation}
where $p^{(r_k)}_{km} = \frac{|\alpha_k \beta_{km}^{(r_k)}|}{R(\vec{r})}$ are probabilities, $\widetilde{U}^{(r_k)}_{km}= U^{(r_k)}_{km}\frac{\alpha_k \beta_{km}^{(r_k)}}{|\alpha_k \beta_{km}^{(r_k)}|}$ are unitaries that absorb the phase of the coefficients, and $R(\vec{r})$ is the weight of the linear combination, which satisfies
\begin{align}
    R(\vec{r}) &= \sum_{(k,m) \in F\times T} \left|\alpha_k \beta_{km}^{(r_k)} \right|\\
    &\leq \sum_{k \in F} \left|\alpha_k\right| \exp(\lambda^2t_{k}^2/r_k)\\ 
    &=\alpha \exp(\lambda^2t_{k}^2/r_k)\,,
\end{align}
where in the last line we have denoted $\alpha:=\sum_{k \in F} \left|\alpha_k\right|$. We note that one is free to tune $\vec{r}=(r_1,...,r_{|F|})$, and in doing so, change the gate depth of the circuit to apply each unitary, whilst also changing $R(\vec{r})$ which will feed into the sample complexity. One simple choice is to set $r_k =\lambda^2t_k^2$ for all $k$, which gives $R(\vec{r}) \leq \alpha e$, where we recall $\alpha$ is the weight of the coefficients in the Fourier series. We now consider our two settings separately.

(i) \textit{Preparation of $\Tr\left[s(\varepsilon,A)\rho s(\varepsilon,A)^{\dag} O\right]$}. Using Eq.~\eqref{eq:appdx-s-decomposition} can now express this quantity as 
\begin{align}
    \Tr\left[s(A)\rho s(A)^{\dag} O\right] = R(\vec{r})^2 \sum_{(k,m),(k',m') \in F\times T} p^{(r_k)}_{km}p^{(r_{k'})}_{k'm'}  \Tr\left[\widetilde{U}^{(r_k)}_{km}\rho \left(\widetilde{U}^{(r_{k'})}_{k'm'}\right)^{\dag} O\right]\,.
\end{align}
As described in Ref.~\cite{faehrmann2022randomizing}, given controlled access to $\widetilde{U}^{(r_k)}_{km}$ and $\widetilde{U}^{(r_{k'})}_{k'm'}$, we can collect measurement shots corresponding to the quantity  $\frac{1}{2}\left( \Tr[\widetilde{U}^{(r_k)}_{km}\rho (\widetilde{U}^{(r_{k'})}_{k'm'})^{\dag} O] + \Tr[\widetilde{U}^{(r_{k'})}_{k'm'}\rho (\widetilde{U}^{(r_{k})}_{km})^{\dag} O] \right)$ with the quantum circuit in Figure \ref{fig:circuits}(b). This leads to an unbiased estimator for $\frac{1}{R(\vec{r})^2}\Tr\left[s(A)\rho s(A)^{\dag} O\right] $. Moreover, according to Born's rule the individual measurement outcomes $o_j$ take values in the interval $[-\|O\|, \|O\|]$. We now propose the following (informal) algorithm: for some choice of $\vec{r}$, (1) sample according to probability distribution $\{p^{(r_k)}_{km}p^{(r_{k'})}_{k'm'} \}_{kmk'm'}$; (2) prepare the circuit for $\Tr[\widetilde{U}^{(r_k)}_{km}\rho (\widetilde{U}^{(r_{k'})}_{k'm'})^{\dag} O]$ and collect measurement result $z_j$; (3) multiply result by $R(\vec{r})^2$ to get $z_j$; (4) repeat procedure and average over $M$ samples. We present the full formal steps in Algorithm \ref{alg:observable}. As we are effectively sampling numbers in the interval $[-\|O\|R(\vec{r})^2, \|O\|R(\vec{r})^2]$, Hoeffding's inequality specifies that $\overline{z}^{(M)}:=\frac{1}{M}\sum^{M}_{i=1}z_j$ satisfies
\begin{equation}
    \mathrm{Prob}\bigg(\, \Big|\overline{z}^{(M)} - \Tr\left[s(A)\rho s(A)^{\dag} O\right] \Big|\geq \varepsilon \bigg) \leq 2\exp\left(-\frac{M\varepsilon^2}{2\|O\|^2R(\vec{r})^4} \right)\,.
\end{equation}
This implies that, in order to guarantee $\left|\overline{z}^{(M)} - \Tr\left[s(A)\rho s(A)^{\dag} O\right] \right| \leq \varepsilon$ with probability at least $1-\delta$, it is sufficient to perform
\begin{equation}
    M \geq \log\left(\frac{2}{\delta}\right) \frac{2\|O\|^2R(\vec{r})^4}{\varepsilon^2}
\end{equation}
circuit samples.

We now impose the choice $r_k =\lambda^2t_k^2$ for all $k$ which sets $R=\OC(\alpha)$. As specified by Lemma \ref{lem:appdx-kianna-ham-simulation}, this means each $\widetilde{U}^{(r_k)}_{km}$ consists of layers of controlled Pauli gates in between $\lambda^2t_k^2$ single-qubit controlled Pauli rotations. Thus in order to obtain an $\varepsilon$-close approximation with probability at least $1-\delta$ it is sufficient to take number of samples and non-Clifford gates satisfying
\begin{align}
    {\CC}^{O}_{\emph{sample}} =  2e^4\log \left(\frac{2}{\delta}\right)\frac{\|O\|^2\alpha^4}{\varepsilon^2} \,,\quad {\CC}^O_{\emph{gate}} = 2\lambda^2 t_{max}^2+d_{\rho}\,,
\end{align}
respectively, where we have denoted $t_{max} = \max_{k\in F}t_k$.

(ii) \textit{Preparation of $\br{\phi} s(A)\kt{\psi}$}. We can express this quantity as
\begin{equation}
    \br{\phi} s(A)\kt{\psi}  = R(\vec{r}) \sum_{(k,m) \in F\times T}  p^{(r_k)}_{km} \br{0}U^{\dag}_{\phi} \widetilde{U}^{(r_k)}_{km} U_{\psi}\kt{0}\,.
\end{equation}
The real and imaginary parts of the quantity $\br{0}U^{\dag}_{\phi} \widetilde{U}^{(r_k)}_{km} U_{\psi}\kt{0}$ can be recovered separately by the circuit in Fig.~\ref{fig:circuits}(a). In both cases, individual measurement outcomes lie in the interval $[-1,1]$. As before, Hoeffding's inequality tells us that 
\begin{equation}
    M \geq \log\left(\frac{2}{\delta}\right) \frac{4R(\vec{r})^2}{\varepsilon^2}
\end{equation}
shots are required to recover both the imaginary and real parts with probability at least $1-\delta$, to additive error $\varepsilon$. We can again choose $r_k =\lambda^2t_k^2$ for all $k$, giving respective sample and gate counts: 
\begin{align}
    {\CC}^{\phi}_{\emph{sample}} =  4e^2\log \left(\frac{2}{\delta}\right)\frac{\alpha^2}{\varepsilon^2} \,,\quad {\CC}^{\phi}_{gate} = \lambda^2 t_{max}^2+ d_{\phi} + d_{\psi}\,.
\end{align}
\end{proof}

\textit{Constant factor trade offs.} In the above we set the elements of the runtime vector to the value $r_k =\lambda^2t_k^2$ for all $k$. However, we remark that in general by tuning $\vec{r}$ slightly one can make small constant factor trade offs between the sample and gate complexity. Namely, by instead setting $r_k =a\lambda^2t_k^2$ for some constant $a$, one can reduce the sample complexity by a factor $e^{2-2/a}$ for problem (a) and a factor $e^{4-4/a}$ for problem (b), at a cost of a increase in gate depth by a factor $a$.


start

\subsection{Classical preprocessing cost}\label{appdx:preprocessing}

The focus of this work is to characterize cumulative quantum runtime. This is motivated by the fact that we expect quantum clock speeds to be significantly slower than classical ones. However, it is still useful to quantify any classical preprocessing costs in order to carry out full end-to-end comparisons with classical algorithms, which is what we do in this section.

To this end, we indicate the complexities for the classical preprocessing in Algorithm \ref{alg:overlap}, which is equivalent to those for Algorithm \ref{alg:observable}. As in the analysis above in Section \ref{appdx:fourier-sampling-prop} we assume that we set all elements of the runtime vector to be  $r_k =\lambda^2t_k^2$. The two preprocessing overheads that we need to characterize are first to evaluate the weight $R(\vec{r}) = \sum_{(k,m) \in F\times T} \big|\alpha_k \beta_{km}^{(r_k)} \big|$ of the sampling protocol (step (2) of Algorithm \ref{alg:overlap}) and second any further preprocessing in order to construct sampling access to our desired probability distribution (allowing step (4) of Algorithm \ref{alg:overlap}).

\textit{Evaluating $R(\vec{r})$.} In Ref.~\cite[Appendix C]{wan2021randomized} it is given that the weight of the coefficients in Eq.~\eqref{eq:appdx-kianna-ham-simulation} explicitly satisfy
\begin{equation}\label{eq:appdx-beta-sum}
    \sum_{m \in T}\beta_{m}^{(r_k)} = \sum_{n=0}^{\infty} \frac{1}{(2 n) !} \left(\frac{1}{\lambda t_k}\right)^{2 n} \sqrt{1+\left(\frac{(1/\lambda t_k)}{2 n+1}\right)^2}\,.
\end{equation}
If there are $K=|F|$ Fourier terms in the Fourier series in Eq.~\eqref{eq:appdx-fourier-A}, then $K$ such quantities need to be evaluated. 
If there is no efficient way to evaluate the sum in Eq.~\eqref{eq:appdx-beta-sum}, we remark the sum has exponentially vanishing terms, and can also be truncated, which we discuss more below. We denote the truncation degree to the sum in Eq.~\eqref{eq:appdx-beta-sum} which leads to additive error $\varepsilon_k$ as $J(\varepsilon_k)$. In order to keep the total additive error $\varepsilon$, for each Fourier term $k$ it is sufficient to pick $\varepsilon_k=\OC(\varepsilon/\alpha_k)\geq \OC(\varepsilon/\alpha)$. Thus, it is sufficient to approximate $R(\vec{r})$ in time $\OC(K\cdot J(\varepsilon/\alpha))$.

\textit{Sample access preprocessing.} In order to sample from a discrete probability distribution with $S$ unique probabilities in $\OC(1)$ time, it is sufficient to spend $\OC(S)$ memory and preprocessing time. This can be achieved with (for instance) the alias method \cite{walker1974new}. In step (4) of Algorithm \ref{alg:overlap} we must sequentially sample the index $k\in F$ followed by $m \in T$. 
In order to sample the index $m$, we sample from the probability distribution proportional to the linear combination in Eq.~\eqref{eq:appdx-kianna-ham-simulation}, which consists of sampling a Taylor series order $n$ followed by sampling a Pauli term $n+1$ times.
Setting up sampling access to the Pauli coefficients of $A$ costs $\OC(L)$ classical overhead. Setting up sampling overhead to the Taylor series order has linear overhead in the truncation order. Again denoting the truncation order of Eq.~\eqref{eq:appdx-beta-sum} which leads to additive error $\varepsilon/\alpha$ as $J(\varepsilon/\alpha)$, we can write the classical preprocessing overhead to construct sample access as $\OC(K + L+ J(\varepsilon/\alpha))$. Additionally, each time the loop in step (4) repeats (corresponding to one quantum sample), there is at most $\OC(J(\varepsilon/\alpha))$ overhead corresponding to sampling from the probability distribution according to the Pauli coefficients of $A$ $\OC(n)$ times.

It is shown in Ref.~\cite[Theorem 4]{wan2021randomized} that it is sufficient to pick a truncation order $J(\varepsilon_k)=\OC(\log(\lambda^2t^2_k/\varepsilon_k))$ to approximate Eq.~\eqref{eq:appdx-beta-sum} to a given error $\varepsilon_k$. Thus, in order to preserve the complexity guarantees specified by Proposition \ref{prop:appdx-fourier-sampling} with error in solution $\varepsilon$, it is sufficient to choose truncation order $J(\varepsilon/\alpha) = \OC(\log(\alpha\lambda^2t^2_k/\varepsilon))$ for each Fourier term labelled by $k$. Putting everything together, we have preprocessing overhead $\OC(K + L+ J(\varepsilon/\alpha) + K\log(\alpha\lambda^2t^2_{max}/\varepsilon)) = \OC(L+ K\log(\alpha\lambda^2t^2_{max}/\varepsilon))$. Additionally, each quantum sample comes with a classical overhead.

We comment that in Fourier approximation approaches to quantum algorithms, a classical complexity depending on $K$ is generic, as at minimum one must store $K$ values and process them. For the algorithms we consider in the rest of this manuscript, $K$ is polynomial in all problem parameters. One thing that is beneficial with the randomized framework is that $K$ only appears in the classical preprocessing overhead. In contrast, for instance, in the LCU approach $K$ would appear logarithmically in the qubit overhead and linearly in the gate overhead. Similarly, a classical overhead of $\OC(L)$ is generic for our Pauli access model. In order to obtain a better complexity, one may hope for additional structure in the Pauli basis, for instance if the Pauli coefficients only take a small number of unique values.  This can be thought of as analogous to the sparse access model, where in order to avoid $\OC(N)$ data access overhead it is not sufficient with known techniques simply to have a sparse matrix; the values must also be efficiently computable. 

end


\subsection{Generalized Fourier sampling - proof of Theorem \ref{thm:general-sampling}}\label{appdx:thm1}

We first establish a lemma relating the closeness of operators to the closeness of expectation values constructed from them.

\begin{lemma}[Tightness of expectation values pt.~1]\label{lem:appdx-expectation-overlap}
Consider two operators $Y$ and $D$ that satisfy the closeness relation $\|Y-D\| \leq \varepsilon \leq 1$, $\|Y\|\geq 1$. Then we have 
\begin{align}
    \big|\Tr[ OY\rho Y^\dag] - \Tr[ OD \rho D^\dag]\big| \leq 3\|O\| \|Y\|\varepsilon\,.
\end{align}
\end{lemma}

\begin{proof}
We have \begin{align}
    \big|\Tr[ OY\rho Y^\dag] - \Tr[ OD \rho D^\dag]\big| &= \big|\Tr[ O(Y-D)\rho Y^\dag] + \Tr[ O(D-Y) \rho (Y^\dag-D^\dag)]+ \Tr[ OY \rho (Y^\dag-D^\dag)]\big|\\
    &\leq \|O\|\left( \|Y-D\| \|Y^\dag\| + \|Y-D\|^2 + \|Y\| \|Y-D\| \right)\\
    &\leq \|O\|\left( 2\|Y\|\varepsilon + \varepsilon^2 \right)\\
    &\leq 3\|O\| \|Y\|\varepsilon\,,
\end{align}
where in the first line we have added and subtracted $\Tr[ OD\rho Y^\dag]+ \Tr[ OY\rho Y^\dag]+ \Tr[ OY\rho D^\dag]$; the first inequality is due to the triangle inequality, H\"older's tracial matrix inequality and the sub-multiplicativity of the operator norm; and the second and third inequalities are due to our starting assumptions.
\end{proof}

We can now proceed with the proof of the main theorem (Theorem \ref{thm:general-sampling}). We suppose we have a matrix function $f(A)$ that is approximated by a Fourier series as
\begin{equation}
    \left\| f(A)-s(\varepsilon,A) \right\| \leq \varepsilon\,; \quad s(\varepsilon,A) = \sum_{k\in F_{\varepsilon,A}} \alpha_k(\varepsilon,A) \exp{\left(it_k(\varepsilon,A)A\right)}\,,
\end{equation}
for some set of Fourier parameters $\alpha_k(\varepsilon,A)$ and $t_k(\varepsilon,A)$, for any $\varepsilon \leq 1$. 
Denote the $\ell_1$-norm of the coefficients as $\alpha(\varepsilon):=\sum_k |\alpha_k(\varepsilon)|$. Suppose further that $A$ has known Pauli decomposition $A=\sum_{\ell} a_{\ell} P_{\ell}$ with Pauli weight $\lambda = \sum_{\ell} |a_{\ell}|$, and that we are given some normalization constant $q$.

\begin{theorem}[Generalized sampling from Fourier approximations  -- detailed version]\label{thm:appdx-general-sampling}
Suppose we have a matrix function $f(A)$ that is approximated by a Fourier series as
\begin{equation}
    \left\| f(A)-s(\tilde{\varepsilon},A) \right\| \leq \tilde{\varepsilon}\leq1\,; \quad s(\tilde{\varepsilon},A) = \sum_{k\in F_{\tilde{\varepsilon},A}} \alpha_k(\tilde{\varepsilon},A) \exp{\big(it_k(\tilde{\varepsilon},A)A\big)}\,,
\end{equation}
where $A$ is some $N\times N$ Hermitian matrix, and for some set of Fourier parameters $\alpha_k(\tilde{\varepsilon},A)$ and $t_k(\tilde{\varepsilon},A)$ with tunable error parameter $\tilde{\varepsilon}$. Denote the $\ell_1$-norm of the coefficients as $\alpha(\tilde{\varepsilon},A):=\sum_k |\alpha_k(\tilde{\varepsilon},A)|$. Suppose further that $A$ has known Pauli decomposition $A=\sum_{\ell} a_{\ell} P_{\ell}$ with Pauli weight $\lambda = \sum_{\ell} |a_{\ell}|$. Then, 
\begin{enumerate}[(a).]
    \item Given unitaries $U_{\phi}$, $U_{\psi}$ to prepare $\kt{\phi}$, $\kt{\psi}$ respectively in depth $d_{\phi}, d_{\psi}$,  we give a randomized quantum algorithm that uses $\log(N)+1$ qubits to approximate the quantity $\frac{1}{q}\br{\phi} f(A)\kt{\psi}$ up to additive error $\varepsilon\leq \frac{1}{q}$ with probability at least $1-\delta$, using
    \begin{align}\label{eq:thm1-overlap}
        {\CC}^{\phi}_{\emph{sample}} &=  \OC\left( \log \left(\frac{2}{\delta}\right)\frac{[\alpha(\varepsilon q/2,A)]^2}{\varepsilon^2q^2} \right)\,,\quad {\CC}^{\phi}_{\emph{gate}} = \OC\left(\lambda^2 [t_{max}(\varepsilon q/2,A)]^2 + d_{\phi}+ d_{\psi} \right)\,,
    \end{align} 
    circuit samples (each calling $U_{\phi}$ and $U_{\psi}$ one time) and gate depth respectively, where each circuit takes the form in Figure \ref{fig:circuits}(a).
    \item Given a procedure to prepare the quantum state $\rho$ in gate depth $d_{\rho}$ and perform measurements with measurement operator $O$, and given normalization constant $q$, there exists a randomized quantum algorithm that uses $\log(N)+1$ qubits to approximate $\frac{1}{q^2}\Tr\left[f(A)\rho f(A)^{\dag} O\right]$ up to additive error $\varepsilon\leq \frac{\|O\|\|f(A)\|}{q^2}$ with probability at least $1-\delta$, using 
    \begin{align}\label{eq:thm1-expectation}
        {\CC}^O_{\emph{sample}} &=  \OC\left( \log \left(\frac{2}{\delta}\right)\frac{\|O\|^2\left[\alpha\big(\frac{\varepsilon q^2}{6\|O\|\|f(A)\|},A\big)\right]^4}{\varepsilon^2 q^4} \right)\,,\quad {\CC}^O_{\emph{gate}} = \OC\left(\lambda^2 \Big[t_{max}\Big(\frac{\varepsilon q^2}{6\|O\|\|f(A)\|},A\Big)\Big]^2 + d_{\rho}\right)\,,
    \end{align}
    circuit samples and non-Clifford gates respectively, where we denote $t_{max}(\varepsilon,A) = \max_{k\in F_{\varepsilon,A}}t_k(\varepsilon,A)$ and each circuit takes the form in Figure \ref{fig:circuits}(b).
\end{enumerate}
\end{theorem}

\begin{proof}[Proof of Theorem \ref{thm:general-sampling} (Generalized sampling from Fourier approximations)] We consider the two cases separately.

(i) \textit{Preparation of $\frac{1}{q^2}\Tr\left[f(A)\rho f(A)^{\dag} O\right]$.} We will use Proposition \ref{prop:appdx-fourier-sampling} to statistically approximate the quantity $\frac{1}{q^2}\Tr\left[s(\tilde{\varepsilon},A)\rho s(\tilde{\varepsilon},A)^{\dag} O\right]$. Using Lemma \ref{lem:appdx-expectation-overlap} we will then show that, for appropriately chosen $\tilde{\varepsilon}$, this leads to an $\varepsilon$-close approximation of $\frac{1}{q^2}\Tr\left[f(A)\rho f(A)^{\dag} O\right]$.

Recall we denote $\overline{z}^{(M)}$ as the statistical approximation to $\Tr\left[s(\tilde{\varepsilon},A)\rho s(\tilde{\varepsilon},A)^{\dag} O\right]$ using Algorithm \ref{alg:observable} with $M$ shots. We would like to find parameters such that 
\begin{equation}\label{eq:thm1-error}
    \left| \frac{1}{q^2}\overline{z}^{(M)} - \frac{1}{q^2}\Tr\left[f(A)\rho f(A)^{\dag} O\right] \right| \leq \varepsilon\,.
\end{equation}
We first note that due to the triangle inequality, we can write
\begin{align}\label{eq:approx-error-decomposition}
    \left| \frac{1}{q^2}\overline{z}^{(M)} - \frac{1}{q^2}\Tr\left[f(A)\rho f(A)^{\dag} O\right] \right| &\leq  \left| \frac{1}{q^2}\overline{z}^{(M)} - \frac{1}{q^2}\Tr\left[s(\tilde{\varepsilon},A)\rho s(\tilde{\varepsilon},A)^{\dag} O\right] \right| + \nonumber \\ 
    &\quad + \left| \frac{1}{q^2}\Tr\left[s(\tilde{\varepsilon},A)\rho s(\tilde{\varepsilon},A)^{\dag} O\right] - \frac{1}{q^2}\Tr\left[f(A)\rho f(A)^{\dag} O\right] \right|\,.
\end{align}
This separates the approximation error into two components; the statistical contribution of sampling from the Fourier series, and the exact contribution from the quality of the Fourier series approximation. It is sufficient to satisfy Eq.~\eqref{eq:thm1-error} by requiring both terms on the right hand side of Eq.~\eqref{eq:approx-error-decomposition} to each by bounded by $\frac{1}{2}\varepsilon$.

From Lemma \ref{lem:appdx-expectation-overlap}, we can see that the second term in Eq.~\eqref{eq:approx-error-decomposition} is bounded by $\frac{1}{2}\varepsilon$ by finding Fourier parameters such that $\left\| f(A)-s(\tilde{\varepsilon},A) \right\| \leq \frac{1}{6} \varepsilon \frac{q^2}{\|O\|\|f(A)\|}\leq 1$. Thus, it is sufficient to pick $\tilde{\varepsilon} = \frac{1}{6} \varepsilon \frac{q^2}{\|O\|\|f(A)\|}$. From Proposition \ref{prop:fourier-sampling}, the first term in Eq.~\eqref{eq:approx-error-decomposition} is bounded by $\frac{1}{2}\varepsilon$ by choosing $M\geq 8e^4\log \left(\frac{2}{\delta}\right)\frac{\|O\|^2\alpha(\varepsilon q^2/6\|O\|\|f(A)\|)^4}{\varepsilon^2 q^4}$, $\quad {\CC}_{\emph{gate}} = 2\lambda^2 t_{max}^2(\varepsilon q^2/6\|O\|\|f(A)\|) + d_{\rho}$.

(ii) \textit{Preparation of $\frac{1}{q}\br{\phi}f(A)\kt{\psi}$}. Denote as $T'^{(M')}$ the statistical approximation to $\br{\phi}f(A)\kt{\psi}$ using Algorithm \ref{alg:observable} with $M'$ shots. Again, we can split up the approximation error into respective statistical and exact contributions as
\begin{equation}\label{eq:approx-error-decomposition2}
    \left| \frac{1}{q}T'^{(M')} -  \frac{1}{q}\br{\phi}f(A)\kt{\psi} \right| \leq \left| \frac{1}{q}T'^{(M')} -  \frac{1}{q}\br{\phi}s(\tilde{\varepsilon},A)\kt{\psi} \right| + \left| \frac{1}{q}\br{\phi}s(\tilde{\varepsilon},A)\kt{\psi} -  \frac{1}{q}\br{\phi}f(A)\kt{\psi} \right|
\end{equation}

The second term in Eq.~\eqref{eq:approx-error-decomposition2} (exact contribution) is simply bounded as 
\begin{equation}
    \left| \frac{1}{q}\br{\phi}s(\tilde{\varepsilon},A)\kt{\psi} -  \frac{1}{q}\br{\phi}f(A)\kt{\psi} \right| \leq \frac{1}{q} \| s(\tilde{\varepsilon},A) - f(A) \|\,.
\end{equation}
Thus, in order to bound the exact contribution by $\frac{1}{2}\varepsilon$, using Lemma \ref{lem:appdx-expectation-overlap}, is sufficient to find Fourier parameters such that $\| s(\tilde{\varepsilon},A) - f(A) \|\leq \varepsilon q/2 \leq 1$, and we can assign $\tilde{\varepsilon} = \varepsilon q/2$. The first term in Eq.~\eqref{eq:approx-error-decomposition2} (statistical contribution) is then bounded by $\frac{1}{2}\varepsilon$ by choosing $M'\geq 8e^4\log \left(\frac{2}{\delta}\right)\frac{\|O\|^2\alpha(\varepsilon q/2)^4}{\varepsilon^2 q^2}$, $\quad {\CC}_{\emph{gate}} = 2\lambda^2 t_{max}^2(\varepsilon q/2) + d_{\phi} + d_{\psi}$. 
\end{proof}

We remark in the above proof we have split the error into statistical and exact contributions evenly. In many cases, $\alpha(\varepsilon,A)$ can have a subpolynomial dependence on $1/\varepsilon$, such as for all three of our example algorithms. In this case by distributing the error unevenly one can make a constant factor saving to the sample complexity.


\subsection{Sampling normalization constant - proof of Proposition \ref{prop:norm-sampling}}\label{appdx:norm-sampling}

In this section we investigate the complexity of sampling $q=\|f(A)\kt{\phi}\|$ for use as a normalization constant, where $\kt{\phi}$ is some given preparable input state. In the previous sections we found an estimator $\overline{z}^{(M)}:=\frac{1}{M}\sum^M_{j=1}z_j$ constructed from samples $\{z_j\}_{j=1}^M$ which is $\varepsilon$-close to a desired property of a matrix function for sufficiently large $M$. We similarly denote our statistical approximation to $q^2$ as $Q^{(M_q)}$. In the following lemma we establish how well $Q^{(M_q)}$ must approximate $q^2$ such that normalized quantities are additive $\varepsilon$-close.

\begin{lemma}[Normalization precision]\label{lem:appdx-q-approximation}
Suppose we have some $\overline{z}^{(M)} \in \mathbb{R}$, pure state $\kt{\phi}$, and operators $O, f(A)$ that satisfy
\begin{equation}
    \left|\frac{1}{q^2}\overline{z}^{(M)} - \frac{1}{q^2}\br{\phi}f(A)Of(A)^{\dag}\kt{\phi}\right| \leq \varepsilon \leq 1\,,
\end{equation}
with $\|O\|\geq 1$. Further, denote $q = \|f(A)\kt{\phi}\|_2$. Then, if $Q^{(M_q)}$ is the statistical approximation of $q^2$, one can achieve 
\begin{equation}
    \left|\frac{1}{Q^{(M_q)}}\overline{z}^{(M)} - \frac{1}{q^2}\br{\phi}f(A)Of(A)^{\dag}\kt{\phi}\right| \leq 3\varepsilon\,,
\end{equation}
if $Q^{(M_q)}$ satisfies $|Q^{(M_q)}-q^2|\leq \frac{1}{2}\frac{\varepsilon q^2 }{\|O\|}$.
\end{lemma}

\begin{proof}
First, we note that due to the triangle inequality and sub-multiplicativity we have
\begin{align}
    \left|\frac{1}{Q^{(M_q)}}\overline{z}^{(M)} - \frac{1}{q^2}\br{\phi}f(A)Of(A)^{\dag}\kt{\phi}\right| &\leq \left|\frac{1}{Q^{(M_q)}}\overline{z}^{(M)} - \frac{1}{q^2}\overline{z}^{(M)}\right| + \left|\frac{1}{q^2}\overline{z}^{(M)} - \frac{1}{q^2}\br{\phi}f(A)Of(A)^{\dag}\kt{\phi}\right|\\
    &\leq \left|\frac{1}{Q^{(M_q)}} - \frac{1}{q^2}\right| \left| \overline{z}^{(M)} \right| + \varepsilon \\
    &\leq \left|\frac{1}{Q^{(M_q)}} - \frac{1}{q^2}\right| (q^2\|O\| + q^2\varepsilon) + \varepsilon\,, \label{eq:appdx-norm-proof1}
\end{align}
assuming $\varepsilon\leq 1$, where in the penultimate inequality we have used the fact that $\left| \overline{z}^{(M)} \right| \leq \br{\phi}f(A)Of(A)^{\dag}\kt{\phi} + q^2\varepsilon$ where $\br{\phi}f(A)Of(A)^{\dag}\kt{\phi} \leq \|O\|_{\infty}\|f(A)\dya{b}f(A)^{\dag}\|_1 \leq \|O\|_{\infty}\|f(A)\kt{\phi}\|^2_2$. Denoting $|Q^{(M_q)}-q^2|\leq \nu$, we can observe 
\begin{align}
    \left|\frac{1}{Q^{(M_q)}} - \frac{1}{q^2}\right| &\leq \frac{\nu}{Q^{(M_q)}q^2} \\
    &\leq \frac{\nu}{q^2(q^2-\nu)} \\
    &= \frac{\nu}{q^4(1-\frac{\nu}{q^2})}\\
    &\leq \frac{\nu}{q^4}\left(8\left( \frac{\nu}{q^2}\right)^2 - 4 \frac{\nu}{q^2} + 2\right)\,.
\end{align}
Then one can check that $\nu= \frac{1}{2}\frac{\varepsilon q^2}{\|O\|}$ gives $\left|\frac{1}{Q^{(M_q)}} - \frac{1}{q^2}\right| \leq \frac{\varepsilon^2}{q^2\|O\|^2}(\frac{\varepsilon}{\|O\|}-1) + \frac{\varepsilon}{q^2\|O\|} \leq \frac{\varepsilon}{q^2\|O\|}$ and so, returning to Eq.~\eqref{eq:appdx-norm-proof1}, under this condition we have $\left|\frac{1}{Q^{(M_q)}}\overline{z}^{(M)} - \frac{1}{q^2}\br{\phi}f(A)Of(A)^{\dag}\kt{\phi}\right| \leq 3\varepsilon$ for $\varepsilon\leq 1$, $\|O\|\geq 1$. 
\end{proof}

By similar reasoning, we also have the following lemma for quantities of the form of a state overlap, rather than expectation values of general measurement observables.

\begin{lemma}[Normalization precision pt.~2]
\label{lem:appdx-q-approximation2}
Suppose we have some $\overline{z}^{(M)} \in \mathbb{R}$, pure states $\kt{\phi}$ $\kt{\psi}$, and an operator $f(A)$ that satisfy
\begin{equation}
    \left|\frac{1}{q}\overline{z}^{(M)} - \frac{1}{q}\br{\phi}f(A)\kt{\psi}\right| \leq \varepsilon \leq 1\,,
\end{equation}
where we denote $q = \|f(A)\kt{\phi}\|_2$. Then, if $Q^{(M_q)}$ is the statistical approximation of $q^2$, one can achieve 
\begin{equation}
    \left|\frac{1}{\sqrt{Q^{(M_q)}}}\overline{z}^{(M)} - \frac{1}{q}\br{\phi}f(A)\kt{\psi}\right| \leq 3\varepsilon\,,
\end{equation}
if $Q^{(M_q)}$ satisfies $|Q^{(M_q)}-q^2|\leq \frac{1}{2}\varepsilon q^2$.    
\end{lemma}

\begin{proof}
Similar to the proof of Lemma \ref{lem:appdx-q-approximation} we have 
\begin{align}
    \left|\frac{1}{\sqrt{Q^{(M_q)}}}\overline{z}^{(M)} - \frac{1}{q}\br{\phi}f(A)\kt{\psi}\right|  \leq \left|\frac{1}{\sqrt{Q^{(M_q)}}} - \frac{1}{q} \right| q (1+\varepsilon) + \varepsilon\,.
\end{align}
Moreover, if $|Q^{(M_q)}-q^2|\leq \nu$ then
\begin{align}
    \left|\sqrt{Q^{(M_q)}}-q\right|\leq \frac{\nu}{\sqrt{Q^{(M_q)}}+q} \leq \frac{\nu}{q}\,.
\end{align}
Thus,
\begin{align}
    \left|\frac{1}{\sqrt{Q^{(M_q)}}} - \frac{1}{q}\right| 
    &\leq \frac{\left|\sqrt{Q^{(M_q)}}-q\right|}{\sqrt{Q^{(M_q)}}q} \\
    &\leq \frac{\nu}{\sqrt{Q^{(M_q)}}q^2} \\
    &\leq \frac{\nu}{q^2(q-\nu/q)} \\
    &\leq \frac{\nu}{q^3}\left(8\left( \frac{\nu}{q^2}\right)^2 - 4 \frac{\nu}{q^2} + 2\right)\,,
\end{align}
and the choice $\nu= \frac{1}{2}{\varepsilon q^2}$ gives
\begin{align}
     \left|\frac{1}{\sqrt{Q^{(M_q)}}} - \frac{1}{q}\right| \leq \frac{\varepsilon}{q}\,,
\end{align}
which leads to the desired result, under the assumption $\varepsilon\leq 1$.
\end{proof}

We will also need the following lemma, which gives a tighter result than Lemma \ref{lem:appdx-expectation-overlap} for pure states.

\begin{lemma}[Tightness of expectation values pt.~2]\label{lem:appdx-expectation-overlap2}
Given two operators $C$ and $D$, where $C$ is Hermitian, and given a pure state $\kt{\phi}$, we have
\begin{equation}
    \left|\br{\phi}D^\dag O D\kt{\phi}- \br{\phi}COC\kt{\phi}\right| \leq 3 \left\|C\kt{\phi}\right\|_2 \|C-D\|_{\infty} \|O\|_{\infty}\,,
\end{equation}
where we have assumed that $\|C-D\|_{\infty}\leq \left\|C\kt{\phi}\right\|_2$.
\end{lemma}

\begin{proof}
We directly bound the difference in expectation values as
\begin{align}
    \left|\br{\phi}D^\dag O D\kt{\phi}- \br{\phi}COC\kt{\phi}\right| &= \Big| \br{\phi}C O (D-C)\kt{\phi}+ \br{\phi}(D^{\dag}-C)OC\kt{\phi} + \br{\phi}(D^{\dag}-C)O(D-C)\kt{\phi} \Big| \\
    &\leq \left\| \dya{\phi}C \right\|_1 \|O(D-C)\|_{\infty} + \left\| C\dya{\phi} \right\|_1 \|(D^{\dag}-C)O\|_{\infty} + \nonumber\\
    &\quad+ \left\| \dya{\phi} \right\|_1 \|(D^{\dag}-C)O(D-C)\|_{\infty} \\
    &\leq \left\| \dya{\phi} \right\|_2 \left\| \dya{\phi}C \right\|_2 \|O\|_{\infty}\|D-C\|_{\infty} + \left\| C\dya{\phi} \right\|_2 \left\| \dya{\phi} \right\|_2  \|D^{\dag}-C\|_{\infty}\|O\|_{\infty} + \nonumber \\ 
    &\quad + \|D^{\dag}-C\|_{\infty}\|O\|_{\infty}\|D-C\|_{\infty} \\
    &= 2\|C\kt{\phi}\|_2\|C-D\|_{\infty}\|O\|_{\infty} + \|C-D\|^2_{\infty}\|O\|_{\infty}\\
    &\leq 3\|C\kt{\phi}\|_2\|C-D\|_{\infty}\|O\|_{\infty}\,,
\end{align}
where in the first equality we have added and subtracted $\br{\phi}COD\kt{\phi}+\br{\phi}D^{\dag}OC\kt{\phi} + \br{\phi}COC\kt{\phi}$, in the first inequality we have used the triangle inequality and the tracial matrix H\"older's inequality, in the second inequality we have used the Cauchy-Schwarz inequality and the sub-multiplicativity of the operator norm, in the subsequent equality we have used the definition of the vector $2$-norm and the fact that $C$ is Hermitian, and in the final inequality we have used our starting assumption that $\|C-D\|_{\infty}\leq \left\|C\kt{\phi}\right\|_2$.
\end{proof}

Lemmas \ref{lem:appdx-q-approximation} and \ref{lem:appdx-q-approximation2} demonstrate that if one has a normalization factor that also needs to be statistically approximated, then one can simply modify the error parameter in Theorem \ref{thm:appdx-general-sampling} by a factor 3, and use a separate algorithm to approximate the normalization factor. In the case of the state normalization factor $q=\|f(A)\kt{\psi}\|$, we now demonstrate the complexities required for a randomized quantum algorithm that approximates it to sufficient error as specified by Lemmas \ref{lem:appdx-q-approximation}, \ref{lem:appdx-q-approximation2} and \ref{lem:appdx-expectation-overlap2} (see Proposition \ref{prop:norm-sampling} of the main text). Following previous sections we denote our matrix function of interest as $f(A)$ given matrix $A$, and denote its Fourier series approximation as $s(\varepsilon,A)$.

\begin{proposition}[Sampling normalization constant -- detailed version]\label{prop:appdx-norm-sampling}
If $q=\|f(A)\kt{\psi}\|$ where $\kt{\psi}$ is a preparable input state in depth $d_{\psi}$ with unitary $U_{\psi}$, then we give a randomized algorithm to approximate $q$ for the two algorithms in Theorem \ref{thm:appdx-general-sampling} that has success probability at least $(1-\delta)$ and complexity
\begin{equation}\label{eq:norm-complexity}
    {\CC}_{\emph{sample}} =  \OC\left( \log \left(\frac{2}{\delta}\right)\frac{c^2\,[\alpha(\frac{\varepsilon q}{12c}, A)]^4}{\varepsilon^2 q^4} \right)\,,\quad 
    {\CC}_{\emph{gate}}= \OC\left(\lambda^2 \big[t_{max}\big(\frac{\varepsilon q}{12c}, A \big)\big]^2 + d_{\psi}\right)\,,
\end{equation}
where $c = 1$ for the algorithm that prepares $\frac{1}{q}\br{\psi}U f(A)V\kt{\psi}$, and $c=\|O\|$ for the algorithm that prepares $\frac{1}{q^2}\Tr\left[f(A)\dya{\psi} f(A)^{\dag} O\right]$.
\end{proposition}

\begin{proof}[Proof of Proposition \ref{prop:appdx-norm-sampling}]
Lemmas \ref{lem:appdx-q-approximation} and \ref{lem:appdx-q-approximation2}  specify that, in order to statistically approximate the normalized expectation value to additive error $\varepsilon$, one requires a statistical approximation to $q^2$ using $M_q$ shots, which we denote as $Q^{(M_q)}$, that satisfies $|Q^{(M_q)}-q^2|\leq \frac{1}{6}\varepsilon q^2$ for problem (a) and satisfies $|Q^{(M_q)}-q^2|\leq \frac{1}{6}\frac{\varepsilon q^2}{\|O\|}$ for problem (b). We thus hereon deal with problem (b), and note that the resources required for problem (a) can be accounted for by setting $\|O\|\rightarrow 1$. 

We propose to construct $Q^{(M_q)}$ by sampling $\br{\phi}s(\varepsilon,A)^{\dag} s(\varepsilon,A)\kt{\phi}$ via Algorithm \ref{alg:overlap} and the Hadamard test circuit in Fig.~\ref{fig:circuits}(a). We require
\begin{equation}\label{eq:q-approx-breakdown}
    \left| Q^{(M_q)} - \br{\phi}f(A)^2\kt{\phi} \right| \leq \left| Q^{(M_q)} - \br{\phi}s(\varepsilon,A)^{\dag}s(\varepsilon,A)\kt{\phi} \right| + \left| \br{\phi}s(\varepsilon,A)^{\dag}s(\varepsilon,A)\kt{\phi} - \br{\phi}f(A)^2\kt{\phi} \right|\,,
\end{equation}
to be smaller than $\frac{1}{6}\frac{\varepsilon q^2}{\|O\|}$. The second term on the right hand side is bounded by Lemma \ref{lem:appdx-expectation-overlap2} as
\begin{align}
    \left| \br{\phi}s(\varepsilon,A)^{\dag}s(\varepsilon,A)\kt{\phi} - \br{\phi}f(A)^2\kt{\phi} \right| \le 3q \left\|s(\varepsilon,A) - f(A)\right\| \,.
\end{align}
Thus, evenly distributing error across statistical and exact contributions, we require a Fourier series approximation that is tight up to $\left\|s(\varepsilon,A) - f(A)\right\| \leq \frac{1}{36} \frac{\varepsilon q}{\|O\|}$, which sets the second term in Eq.~\eqref{eq:q-approx-breakdown} to be smaller than $\frac{1}{12}\frac{\varepsilon q^2}{\|O\|}$. Using Proposition \ref{prop:fourier-sampling}, we can bound the first term in Eq.~\eqref{eq:q-approx-breakdown} by the value $\frac{1}{12}\frac{\varepsilon q^2}{\|O\|}$ by using $M_q  \geq 2e^4\log\left(\frac{2}{\delta}\right)\frac{\alpha(\frac{1}{36} \frac{\varepsilon q}{\|O\|})^2}{\left(\frac{1}{12}\frac{\varepsilon q^2}{\|O\|}\right)^2}= 288e^4\log\left(\frac{2}{\delta}\right)\frac{\|O\|^2\alpha(\frac{1}{12} \frac{\varepsilon q}{\|O\|})^2}{\varepsilon^2 q^4}$ shots and gate depth $\CC_{gate} = 2\lambda^2t^2_{max}(\frac{1}{12} \frac{\varepsilon q}{\|O\|})$.
\end{proof}


\subsection{Statistical encoding of classical vectors}\label{appdx:randomized-encoding}

In this section we show how to statistically "encode" a classical vector $\vec{b}=(b_1,...,b_N)$ to recover its properties when determining state overlaps and expectation values using our randomized algorithms.
We show that with $\OC(s)$ classical pre-processing steps one can start from classical access to a vector $\vec{b}$ and invoke it as part of a larger randomized algorithm with a cost of increased circuit samples depending on $\|\vec{b}\|_1$, and minimal quantum gate overhead. This can then replace the usual (quantum) state oracle in linear algebra algorithms, such as in our statistical algorithm for linear systems as described in Proposition \ref{prop:statistical-encoding} in the main text. 
We note that, traditionally, quantum algorithms that act on classical data assume an encoding is given via a normalized input vector of the form $\kt{\vec{b}} := \frac{1}{\|\vec{b}\|_2} \sum_i {b_i} \kt{i}$. However, our compilation scheme allows for arbitrary normalization, thus we can recover the action of the true vector $\vec{b}$ and take into account its magnitude.

\begin{definition}[Classical sparse access model]\label{def:sparse-access}
We say that a classical vector $\vec{b}=(b_1,...,b_N)$ is stored with sparse access if the following set of tuples is stored and accessible in classical memory:
\begin{equation}
    \BC = \{(b_i,i)\, |\, b_i \neq 0 \}. 
\end{equation}
\end{definition}

We suppose that we wish to randomly compile properties of a matrix $G$ decomposed into a linear combination of implementable unitaries $G=\sum_jg_jV_j$ with weight of coefficients $g:=\sum_j|g_j|$. Using our randomized sampling scheme with the circuits in Figure \ref{fig:circuits} (a) and (b), we can prepare quantities of the form $\br{\psi} G\kt{\vec{b}}$ and  $\br{\vec{b}}G^{\dag}O G\kt{\vec{b}}$ with respective sample complexities $\OC(g^2/\varepsilon^2)$ and $\OC(g^4\|O\|^2/\varepsilon^2)$ and gate depth scaling with the largest gate depth in $\{ V_j\}_j$. This presumes that there is a procedure to prepare the states $\kt{\psi}$, $\kt{\vec{b}}$. If $\kt{\vec{b}}$ is not provided via an oracle, the following proposition shows how to statistically recover the encoding, and how the sample complexities change.

\begin{proposition}[Random compiling classical input vector -- detailed version]\label{prop:appdx-compiling-vector}
Suppose that a vector $\vec{b}=(b_1,...,b_N)$ with sparsity $s$ is stored classically with sparse access as defined in Definition \ref{def:sparse-access}. Suppose further that we have access to some measurement observable $O$ and a unitary $U_{\psi}$ to prepare state $\kt{\psi}$. Then, by providing classical pre-processing in $\OC(s)$ time, there exists a randomized quantum algorithm that returns an approximation of (a) $\br{\psi}G\kt{\vec{b}}$; (b)  $\br{\vec{b}}G^{\dag}OG\kt{\vec{b}}$ to additive error $\varepsilon$ and constant arbitrary probability, using 
\begin{equation}
    \CC^{\psi}_{\emph{sample}} = \OC\left(\frac{g^2}{\varepsilon^2} \frac{\|\vec{b}\|_1^2}{\|\vec{b}\|_2^2}\right)\,, \; \CC^{O}_{\emph{sample}} = \OC\left(\frac{g^4\|O\|^2}{\varepsilon^2} \frac{\|\vec{b}\|_1^4}{\|\vec{b}\|_2^4} \right)\,,
\end{equation}
for scenarios (a) and (b) respectively, where each circuit uses one application of $U_{\phi}$. The additional gate overhead for randomly compiling the input vector is one or two controlled multi-qubit NOT gates.
\end{proposition}

\begin{proof}[Proof of Proposition \ref{prop:appdx-compiling-vector}] 
We first give an overview of the main steps of our idea, before giving more precise analysis of quantum and classical resources required. We express $\kt{\vec{b}}=\sum_{i=1}^{N} b_i \kt{i}/\|\vec{b}\|_2$ as a weighted probabilistic sum over states
\begin{equation}\label{eq:b-decomposition}
    \kt{\vec{b}} =\frac{\|\vec{b}\|_1}{\|\vec{b}\|_2}\sum_{i=1}^{N} p^{(\vec{b})}_i\cdot \textrm{sgn}(b_i)\kt{i}=\frac{\|\vec{b}\|_1}{\|\vec{b}\|_2}\sum_{i=1}^{N} p_i\cdot \textrm{sgn}(b_i)X_i\kt{0}\,,
\end{equation}
where we denote the probabilities $p^{(\vec{b})}_i = {|b_i|}/{\|\vec{b}\|_1}$, we denote $\textrm{sgn}(b_i) = b_i/|b_i|$, and $X_i$ is the $\ceil{\log N}$-qubit Pauli operator in $\{ \id, X\}^{\otimes \ceil{\log N}}$ that corresponds to the binary representation of $i$. If the description in Eq.~\eqref{eq:b-decomposition} is classically accessible, then $\kt{\vec{b}}$ can be statistically encoded by sampling from this distribution as follows.

(a). \textit{Recovering $\br{\psi}G\kt{\vec{b}}$}. We can express $\br{\psi}G\kt{\vec{b}}$ as the probabilistic sum
\begin{equation}
    \br{\psi}G\kt{\vec{b}} =\frac{\|\vec{b}\|_1}{\|\vec{b}\|_2}\sum_{i=1}^{N} p^{(\vec{b})}_i\cdot \textrm{sgn}(b_i)\br{0}U_{\psi}^{\dag}GX_i\kt{0} = \frac{\|\vec{b}\|_1}{\|\vec{b}\|_2} g\sum_{i=1}^{N} p^{(\vec{b})}_i p^{(G)}_j\cdot \textrm{sgn}(b_i g_j)\br{0}U_{\psi}^{\dag} V_j X_i\kt{0}\,,
\end{equation}
where in the first equality we have used the fact that $\kt{\psi}=U_{\psi}\kt{0}$ and Eq.~\eqref{eq:b-decomposition}, and in the second equality we have decomposed $G$ into its constituent unitaries as $G=\sum_jg_jV_j$ with weight of coefficients $g:=\sum_j|g_j|$.  The quantity $\br{\psi}G\kt{\vec{b}}$ can then be statistically recovered as follows: (1) Sample indices $i',j'$ from the probability distribution $\{p^{(\vec{b})}_i p^{(G)}_i\}_{i,j}$. (2) Run two Hadamard test circuits (see Fig.~\ref{fig:circuits}(a)) to obtain one measurement sample each of $\mathrm{Re}(\br{0}U_{\psi}^{\dag} V_{j'} X_{i'}\kt{0})$ and $\mathrm{Im}(\br{0}U_{\psi}^{\dag} V_{j'} X_{i'}\kt{0})$ respectively. (3) Classically multiply the result by $\|\vec{b}\|_1g\,\textrm{sgn}(b_{i'} g_{j'})/\|\vec{b}\|_2$ and store the result. (4) Repeat process $\CC^{\phi}_{\emph{sample}}$ times and average over results. Due to Hoeffding's inequality, it is sufficient to take
\begin{equation}
    \CC^{\phi}_{\emph{sample}} = 8\log\left(\frac{2}{\delta}\right)\frac{g^2}{\varepsilon^2} \frac{\|\vec{b}\|_1}{\|\vec{b}\|_2}
\end{equation}
samples to attain an answer within additive error $\varepsilon$ and probability at least $(1-\delta)$. (The additional factor of $4$  comes from approximating $\mathrm{Re}(\br{0}U_{\psi}^{\dag} V_j X_i\kt{0})$ and $\mathrm{Im}(\br{0}U_{\psi}^{\dag} V_j X_i\kt{0})$ each to additive error $\varepsilon/\sqrt{2}$ separately.)

(b). \textit{Recovering $\br{\vec{b}}G^{\dag}OG\kt{\vec{b}}$}. Similar to the above, we can write
\begin{align}
    \br{\vec{b}}G^{\dag}OG\kt{\vec{b}} 
    &= \frac{\|\vec{b}\|^2_1}{\|\vec{b}\|^2_2}\sum_{i,\ell}p_i^{(\vec{b})}p_{\ell}^{(\vec{b})}\cdot \textrm{sgn}(b_ib_{\ell})\br{0}X_i G^{\dag}OG X_{\ell}\kt{0}\,. \\
    &= \frac{\|\vec{b}\|^2_1}{\|\vec{b}\|^2_2} g^2\sum_{i,j,k,\ell}p_i^{(\vec{b})}p_j^{(G)}p_k^{(G)}p_{\ell}^{(\vec{b})}\cdot \textrm{sgn}(b_ig_jg_kb_{\ell})\br{0}X_i V_j^{\dag}OV_k X_{\ell}\kt{0}\,.
\end{align}
We consider a very similar protocol to before: (1) Sample indices $i',j',k',\ell'$ from the probability distribution $\{p_i^{(\vec{b})}p_j^{(G)}p_k^{(G)}p_{\ell}^{(\vec{b})}\}_{ijk\ell}$. (2) Run the  circuit in Fig.~\ref{fig:circuits} to obtain one measurement sample of $\br{0}X_i V_j^{\dag}OV_k X_{\ell}\kt{0}$. (3) Classically multiply the result by $\|\vec{b}\|_1^2g^2\,\textrm{sgn}(b_ig_jg_kb_{\ell}))/\|\vec{b}\|^2_2$ and store the result. (4) Repeat process $\CC^{O}_{\emph{sample}}$ times and average over results. Due to Hoeffding's inequality, it is sufficient to take
\begin{equation}
    \CC^{O}_{\emph{sample}} = 2\log\left(\frac{2}{\delta}\right)\frac{g^4\|O\|^2}{\varepsilon^2} \frac{\|\vec{b}\|^4_1}{\|\vec{b}\|^4_2}
\end{equation}
shots to attain an answer within additive error $\varepsilon$ and probability at least $(1-\delta)$.

Now we consider the classical resources required to determine $\|\vec{b}\|_1$, $\|\vec{b}\|_2$, $\{p_i\}$, $\{\textrm{sgn}(b_i)\}$ given $\BC$. As one can obtain the set of tuples $\overline{\BC}=\big\{\big(|b_i|,i\big)\, |\, b_i \neq 0 \big\}$ in $\OC(s)$ steps from $\BC$, it then follows that all four sets of quantities can be obtained in $\OC(s)$ additional steps by first obtaining $\overline{\BC}$, then combining quantities in $\BC$ and $\overline{\BC}$.
\end{proof}

\begin{remark}[Arbitrary normalization]
We can additionally consider arbitrarily normalized states with normalization constant $m$ as $\frac{1}{m}\kt{\vec{b}}$ by changing the sample complexity by factor $\|\vec{b}\|_1 \rightarrow \frac{1}{m}\|\vec{b}\|_1$. 
\end{remark}


\subsection{Sampling from output vector}\label{appdx:vector-sampling}

In this section we demonstrate a modification of Theorem \ref{thm:appdx-general-sampling} that allows sampling from the vector corresponding to sampling $G\ket{\psi}$ in the computational basis for some state $\ket{\psi}$ and operator $G$ of interest, which is a generalization of our algorithmic framework. As in the previous section, we suppose that there is a known decomposition $G= \sum_i g_i V_i$ into implementable unitaries with $g_i \in \mathbb{C}$. For further simplicity we denote $\sum_i g_i V_i = g \sum_i p_i U_i$ where $g:=\sum_j|g_j|$, $p_i := |g_i|/g$ and $U_i:= (g_i/|g_i|) V_i$ is a unitary that absorbs the phase of $g_i$. Our discussion will remain general to any randomized algorithm that samples from such a collection of unitaries. In what follows we show that an unbiased estimator can be constructed for the vector whose entries take the values $|\bra{\vec{z}_n}G\ket{\psi}|^2$ for each $\vec{z}_n \in \{0,1\}^n$. This mimics the more common approach in quantum algorithms where one prepares the quantum state $G\ket{\psi}/\|G\ket{\psi}\|$ and collects measurement samples in the computational basis. As with Theorem \ref{thm:appdx-general-sampling}, our scheme moves part of the (coherent) quantum  complexity into sample complexity as now the operator $G$ does not need to be materialized quantumly.

In order to produce our desired estimator, we will simply implement the gates in the circuit corresponding to the generalized Hadamard test of Figure~\ref{fig:circuits}(b) and measure all qubits in the computational basis, with minor classical postprocessing. Consider the circuit in Figure~\ref{fig:vectorsampling}, which explicitly shows the gates in the circuit in Figure~\ref{fig:circuits}(b) before computational basis measurement.

\begin{figure}[h]
  \centering
    \includegraphics[width=0.35\columnwidth]{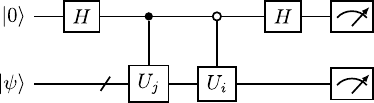}
  \caption{\textbf{Generalized Hadamard test.} We use this circuit along with a modified measurement in Algorithm \ref{alg:observable} to sample from the solution vector. Rather than measuring the observable $Z \otimes O$ we measure all qubits in the computational basis. \label{fig:vectorsampling}}
\end{figure}

The state before measurement is $\frac{1}{2}\big(\ket{0}(U_i+U_j)\ket{\psi} + \ket{1}(U_i-U_j)\ket{\psi}\big)$. Measurement in the computational basis yields 
\begin{align}
    &(0,\vec{z}_{n})\quad \text{with probability}\quad\frac{1}{4} \left| \bra{\vec{z}_{n}} (U_i+U_j)\ket{\psi} \right|^2\,, \\
    &(1,\vec{z}_{n})\quad \text{with probability}\quad\frac{1}{4} \left| \bra{\vec{z}_{n}} (U_i-U_j)\ket{\psi} \right|^2\,,
\end{align}
for all $\vec{z}_{n} \in \{0,1\}^n$.  

Our procedure is as follows. As in Algorithm \ref{alg:observable}, we sample two unitaries $(U_i, U_j)$ independently at a time with probability $p_i p_j$. We run the circuit in Figure \ref{fig:vectorsampling} concluding with computational basis measurement on all qubits, and assign vector element $g^2 (-1)^{z} \ket{\vec{z}_n}$ upon receiving string $(z,\vec{z}_{n})$. One can check that this gives an unbiased estimator for the vector 
\begin{align}
    \sum_{ij} p_i p_j  \sum_{\vec{z}_{n} \in (0,1)^n} & \sum_{z \in \{ 0,1\}}  \mathrm{Prob}\big( (z,\vec{z}_{n})\big| U_i, U_j \big) \cdot g^2 (-1)^{z} \ket{\vec{z}_n} = \\
    &=\sum_{ij} p_i p_j  \sum_{\vec{z}_{n} \in (0,1)^n} \left(\frac{1}{4} \left| \bra{\vec{z}_{n}} (U_i+U_j)\ket{\psi} \right|^2 \cdot (g^2 \ket{\vec{z}_n}) + \frac{1}{4} \left| \bra{\vec{z}_{n}} (U_i-U_j)\ket{\psi} \right|^2 \cdot (-1 g^2 \ket{\vec{z}_n}) \right) \\
    &= \sum_{\vec{z}_{n} \in (0,1)^n} \sum_{ij} g^2\, p_i p_j \left( \frac{1}{2} \bra{\psi}U_i^{\dag}\dya{\vec{z}_{n}}U_j\ket{\psi} \cdot \ket{\vec{z}_{n}} + \frac{1}{2} \bra{\psi}U_j^{\dag}\dya{\vec{z}_{n}}U_i\ket{\psi} \cdot \ket{\vec{z}_{n}} \right)\\
    &= \sum_{\vec{z}_{n} \in (0,1)^n} \big| \bra{\vec{z}_{n}} G \ket{\psi} \big|^2 \cdot \ket{\vec{z}_{n}}\,,
\end{align}
which indeed is the vector with entries $\sum_{\vec{z}_{n} \in (0,1)^n} \big| \bra{\vec{z}_{n}} G \ket{\psi} \big|^2$ as desired, which we from hereon denote as $\vec{G}$. 

How many shots do we need to take in order to obtain good convergence? We appeal to vector Bernstein inequalities. Ref.~\cite{kohler2017sub} essentially gives a vector Bernstein inequality (adapted from Refs.~\cite{gross2011recovering, ledoux1991probability}) for the sample mean $\vec{Y}_M = \sum_{i=1}^M \vec{X}_i$ of $M$ i.d.d.~random variables $\vec{X}_i \in \mathbb{R}^N$ of
\begin{equation}\label{eq:vec-bernstein}
    \mathrm{Prob}\left( \left\| \vec{Y}_M - \mathbb{E}[\vec{Y}_M]\right\| \geq \varepsilon \right) \leq \exp \left( - \frac{M \varepsilon^2}{8 \sigma^2} \right)\,,
\end{equation}
where $\sigma^2$ is an upper bound on the second moment satisfying $\sigma^2 \geq \mathbb{E}\big[\|\vec{X}_i-\mathbb{E}[\vec{X}_i]\|^2_2\big]$ for all $i \in [N]$, and $\mathbb{E}[\vec{Y}_M] = \vec{G}$. Due to the linearlity of the trace we have $\mathbb{E}\big[\|\vec{X}_i-\mathbb{E}[\vec{X}_i]\|^2_2\big] \leq \mathbb{E}\big[\|\vec{X}_i\|_2^2\big]$, and thus we can set $\sigma^2 = g^4$. This implies we obtain an estimator $\vec{Y}_{\CC_{sample}}$ satisfying
\begin{equation}
    \left\| \vec{Y}_{\CC_{sample}} - \vec{G} \right\| \leq \varepsilon \,,
\end{equation}
with probability $(1-\delta)$, using a number of shots satisfying
\begin{equation}
    \CC_{sample} = \OC\left( \frac{g^4}{\varepsilon^2} \log\left(\frac{1}{\delta} \right) \right)\,,
\end{equation}
which is the same sample complexity as specified in Theorem \ref{thm:appdx-general-sampling} to recover a general observable. Thus, here we have shown that instead of measuring an observable, one can use the same number of shots (asymptotically) to recover an approximation of the classical vector corresponding to output probabilities in the computational basis.
end


\subsection{Classical algorithm for matrix polynomials}\label{appdx:classical}

In this section we show that, given a bounded degree polynomial series $h(A) = \sum_{k\in F_d} \alpha_k A^k$ in matrix $A$ where $A$ has known Pauli decomposition with bounded Pauli weight, then one can approximate quantities of the form $\br{t}h(A)\kt{s}$ efficiently classically by means of a randomized algorithm. This is  a more detailed re-statement of Proposition \ref{prop:classical-simp} in the main text.

\begin{proposition}[Sampling from polynomial series -- detailed version]\label{prop:appdx-classical}
Suppose that we have a polynomial series of degree $d$
\begin{equation}
    h(A) = \sum_{k\in F_d} \alpha_k A^k\,,
\end{equation}
in some $N\times N$ Hermitian matrix $A$, where $F_d \subseteq [d]$, and denote the $\ell_1$-norm of the coefficients as $\alpha:=\sum_{k\in F_d} |\alpha_k|$. Suppose further that $A$ has known Pauli decomposition $A=\sum_{\ell} a_{\ell} P_{\ell}$ with Pauli weight $\lambda = \sum_{\ell} |a_{\ell}|$. 
Then, given pure states $\kt{t}$ and $\kt{s}$ which are implementable by initializing in the $\kt{0}$ state and performing $\OC(m)$ Clifford gates, one can classically approximate $\br{t}h(A)\kt{s}$ up to additive error $\varepsilon$ with probability at least $1-\delta$, using 
    \begin{align}
        {\CC}_{\emph{sample}} &=  \OC\left( \log \left(\frac{2}{\delta}\right)\frac{(\alpha \lambda^d)^2}{\varepsilon^2} \right)\,, 
    \end{align}
calls to independent classical subroutines, each requiring $\OC(\log^2(N))$ bits and at most $\OC(d\log^2(N) + m\log(N))$ time.
\end{proposition}

\begin{proof}[Proof of Proposition \ref{prop:appdx-classical}]
We first note that we can write 
\begin{equation}
    A = \lambda \sum_{\ell} p_{\ell} \cdot \textrm{sgn}(a_{\ell}) P_{\ell} \,,
\end{equation}
where $\lambda = \sum_{\ell}|a_{\ell}|$ is the Pauli weight and we denote $p_{\ell} = \frac{|a_{\ell}|}{\lambda}$ and $\textrm{sgn}(a_{\ell}) = \frac{a_{\ell}}{|a_{\ell}|}$. We can then write
\begin{align}
    h(A) &= \sum_{k\in F_d} \alpha_k \lambda^k\, \sum_{{\ell}_1,..., {\ell}_k} p_{{\ell}_1} \cdots p_{{\ell}_k}\, \textrm{sgn}(a_{{\ell}_1})\cdots\,\textrm{sgn}(a_{{\ell}_k})\, P_{{\ell}_1}\cdots P_{{\ell}_k} \\
    &= R_{h(A)} \sum_{k\in F_d} \sum_{{\ell}_1,..., {\ell}_k} q_k\, p_{{\ell}_1}  \,  \cdots p_{{\ell}_k}\, \phi(\alpha_k)\,\textrm{sgn}(a_{{\ell}_1})\cdots\,\textrm{sgn}(a_{{\ell}_k})\, P_{{\ell}_1}\cdots P_{{\ell}_k}\,,
\end{align}
where we denote $R_{h(A)} = \sum_{k\in F_d} |\alpha_k\lambda^k|\leq \alpha\lambda^d$, $q_k=\frac{|\alpha_k|\lambda^k}{R_{h(A)}}$, $\phi(\alpha_k) = \frac{\alpha_k}{|\alpha_k|}$. We note that $p_{{\ell}_1},...,p_{{\ell}_k}$ and $q_k$ are probabilities. This then allows us to statistically recover $h(A)$ with weight $R_{h(A)}$. For simplicity we subsume all indices and write 
\begin{equation}\label{eq:h-decomposition}
    h(A) = G \sum_{i \in S} \tilde{p}_i \tilde{\phi}_i \widetilde{P}_i\,,
\end{equation}
where $\tilde{p}_i=q_k\, p_{{\ell}_1}  \,  \cdots p_{{\ell}_k}$ subsumes all probabilities, $\tilde{\phi}_i=\phi(\alpha_k)\,\textrm{sgn}(a_{{\ell}_1})\cdots\,\textrm{sgn}(a_{{\ell}_k})$ subsumes all phases, $\widetilde{P}_i = P_{{\ell}_1}\cdots P_{{\ell}_k}$ is a product of $k$ Paulis.

\textit{Preparation of $\br{t}h(A)\kt{s}$}. Denote the collection of Cliffords to prepare $\kt{t}$ and $\kt{s}$ as $U_t$ and $U_s$ respectively. Using Eq.~\eqref{eq:h-decomposition} we can simply write 
\begin{align}
    \br{t} h(A)\kt{s} = R_{h(A)} \sum_{i \in S} \tilde{p}_i \tilde{\phi}_i  \br{0}U_t^{\dag} \widetilde{P}_i U_s\kt{0}\,.
\end{align}
Inspecting this, we can consider the following protocol: (1) sample from the above probability distribution, obtaining index $i$ with probability $\tilde{p}_i$; (2) sample one measurement result corresponding to $\br{0}U_t^{\dag} \widetilde{P}_i U_s\kt{0}$ by running the circuit in Figure \ref{fig:circuits}(a); (3) multiply the result by $R_{h(A)} \tilde{\phi}_i$; (4) repeat $M$ times and take mean over results.

We note that in step (2) one requires sampling one measurement result from a circuit that contains at most $2d$ controlled n-qubit Pauli operations plus $\OC(m)$ other Clifford operations. The former part can simply be decomposed into at most $2dn$ controlled single qubit Pauli operations. In order to simulate the circuit overall one requires $\OC(n^2)$ bits and $\OC(dn^2+mn)$ time \cite{aaronson2004improved, gidney2021stim}. Due to Hoeffding's inequality, in order to recover $\br{t}h(A)\kt{s}$ to additive error $\varepsilon$ with probability at least $(1-\delta)$, it is sufficient to use $M$ shots for any
\begin{equation}
    M \geq 2\log\left(\frac{2}{\delta}\right) \frac{R_{h(A)}^2}{\varepsilon^2}\,.
\end{equation}
Using the upper bound $R_{h(A)} \leq \alpha\lambda^d$ we observe that it is sufficient to perform 
\begin{equation}
    \CC_{sample} = 2\log\left(\frac{2}{\delta}\right) \frac{(\alpha\lambda^d)^2}{\varepsilon^2} 
\end{equation}
shots. 

Note that quantities of the form $\Tr\left[h(A)\rho\, h(A)^{\dag} O\right]$ (where $\rho$, $O$ are implementable by Cliffords) can be obtained by taking the modulus squared of quantities of the form $\br{t}h(A)\kt{s}$. Alternatively, they can also be prepared by simulating the circuit in Figure \ref{fig:circuits}(b).

\end{proof}

The above proposition can be trivially extended to consider polynomials of multiple matrices. This implies that subroutines such as matrix multiplication can be efficiently sampled from, given the matrices have {bounded} Pauli weights.


\section{Application-specific results}\label{appdx:applications}

\subsection{Linear systems}\label{appdx:linear-systems}

In this section we detail our results for linear systems. We first summarize a result from Ref.~\cite{childs2017quantum}. In this work, the authors find an approximation of $A^{-1}$ as a linear combination of unitaries $\sum_i \alpha_i \exp(-iAz_j)$. More precisely, they find a Fourier representation of the inverse function as

\begin{equation}
    \frac{1}{x}=\frac{i}{\sqrt{2 \pi}} \int_{0}^{\infty} \mathrm{d} y \int_{-\infty}^{\infty} \mathrm{d} z\, z e^{-z^{2} / 2} e^{-i x y z}\,.
\end{equation}
Further, They show that when truncating the integration range as
\begin{equation}\label{eq:g-def}
    g(x):=\frac{i}{\sqrt{2 \pi}} \int_{0}^{y_{J}} \mathrm{~d} y \int_{-z_{K}}^{z_{K}} \mathrm{d} z\, z e^{-z^{2} / 2} e^{-i x y z}\,,
\end{equation}
we have
\begin{equation}\label{eq:g-approximation}
    \left| g(x) - \frac{1}{x} \right| \leq \frac{1}{|x|} e^{-\left(x y_J\right)^2 / 2}+\frac{2}{|x|} e^{-z_K^2 / 2} \leq \varepsilon\,,
\end{equation}
on the domain $D_{b}:=[-1,-1 / b] \cup[1 / b, 1]$ for some $y_{J}(\varepsilon,b)=\Theta(b \sqrt{\log (b / \varepsilon)})$ and $z_{K}(\varepsilon,b)=\Theta(\sqrt{\log (b / \varepsilon)})$ (Ref.~\cite{childs2017quantum}, Lemma 12). An important point is that this choice is totally independent of the upper limit of the domain, and is wholly dependent on $b$, that is, the same statement holds over the domain $[-a,-1 / b] \cup[1 / b, a]$ with arbitrary $a\geq 1/b$. This is due to the fact that the upper bound on the approximation error in Eq.~\eqref{eq:g-approximation} is a decreasing function in $|x|$. The integral in Eq.~\eqref{eq:g-def} can be discretized as
\begin{equation}\label{eq:approx-1/x}
    h(x):=\frac{i}{\sqrt{2 \pi}} \sum_{j=0}^{J-1} \Delta_{y} \sum_{k=-K}^{K} \Delta_{z} z_{k} e^{-z_{k}^{2} / 2} e^{-i x y_{j} z_{k}}\,,
\end{equation}
where the integration range has been discretized into $J$ and $2K+1$ steps of size $\Delta_{y}$ and $\Delta_{z}$ respectively. Moreover, taking step sizes $\Delta_y=\Theta(\varepsilon / \sqrt{\log (b / \varepsilon)})$ and $\Delta_z = \Theta((b \sqrt{\log (b / \varepsilon)})^{-1})$ guarantees that $h(x)$ is $\varepsilon$-close to $1/x$ on the domain $D_b$. Note, however, as we will only sample from this distribution, for our purposes the resolution of this discretization can be taken to be arbitrarily small. By inspecting Eq.~\eqref{eq:approx-1/x}, one can observe that this corresponds to a Fourier series with maximum time parameter
\begin{equation}
    t_{\max} = \Theta\left(y_J(\varepsilon,b)z_K(\varepsilon,b)\right) = \Theta\left(b \log (b / \varepsilon)\right)\,,
\end{equation}
and coefficients with weight
\begin{equation}
    \frac{1}{\sqrt{2 \pi}} \sum_{j=0}^{J-1} \Delta_{y} \sum_{k=-K}^{K} \Delta_{z}\left|z_{k}\right| e^{-z_{k}^{2} / 2}
    = \Theta\left(y_J(\varepsilon,b)\right)
    = \Theta\left( b \sqrt{\log (b / \varepsilon)} \right)\,.
\end{equation}

By mapping the domain $D_b$ to the spectrum of some matrix $A$, this allows us to establish the following lemma.

\begin{lemma}[Fourier series approximation of inverse operator, adapted from Ref.~\cite{childs2017quantum}]\label{lem:fourier-linear-systems}
Given some matrix $A$ with finite $\|A^{-1}\|$ we have
\begin{equation}\label{eq:fourier-operator}
    \Big\|A^{-1} - \sum_{i \in S_{\varepsilon,A}} \alpha_i(\varepsilon,A) \exp(-iAz_j(\varepsilon,A)) \Big\| \leq \varepsilon\,,
\end{equation}
where $S_{\varepsilon,A}$ is some index set, with maximum time parameter $t_{max}:=\max_{i \in S_{\varepsilon,A}}(z_j)$ satisfying
\begin{align}\label{eq:t-linear-systems}
    t_{max}(\varepsilon,A) = \Theta\left(\|A^{-1}\| \log \left(\frac{\|A^{-1}\|}{\varepsilon} \right) \right)\,,
\end{align}
and Fourier coefficients $\alpha_i(\varepsilon,A)$ with $\ell_1$-norm satisfying
\begin{align}\label{eq:alpha-linear-systems}
    \alpha(\varepsilon,A) := \sum_i \left|\alpha_i(\varepsilon,A)\right| = \Theta\left( \|A^{-1}\| \sqrt{\log \left(\frac{\|A^{-1}\|}{\varepsilon} \right)} \right)\,.
\end{align}
\end{lemma}

Our result for linear systems (Corollary \ref{cor:linear-systems} in the main text) then follows directly from Theorem \ref{thm:appdx-general-sampling} and Proposition \ref{prop:appdx-norm-sampling} as follows.

\begin{corollary}[Linear systems -- detailed version]\label{cor:appdx-linear-systems}
Consider a Hermitian matrix $A$ with known Pauli decomposition $A=\sum_{\ell}a_{\ell}P_{\ell}\,;\; \lambda:=\sum_{\ell}|a_{\ell}|$. Denote $q$ as a freely chosen normalization parameter. Finally, suppose we have ability to prepare state $\kt{\vec{b}}$ in $\OC(d_{\vec{b}})$ depth.
\begin{itemize}
    \item Given ability to implement $U_{\psi}\kt{0} = \kt{\psi}$ in gate depth $d_{\psi}$, there exists a randomized quantum algorithm that returns $\frac{1}{q}\br{\psi}A^{-1}\kt{\vec{b}}$ up to additive error $\varepsilon$ with probability at least $1-\delta$, utilizing $\OC\left( \log \left(\frac{2}{\delta}\right)\frac{\|A^{-1}\|^2}{\varepsilon^2q^2}\log \left(\frac{\|A^{-1}\|^2}{\varepsilon q^2} \right)\right)$ circuit runs each with gate depth $\OC\left(\|A^{-1}\|^2\lambda^2 \log^2\left(\frac{\|A^{-1}\|^2}{\varepsilon q^2} \right) + d_{\psi} + d_{\vec{b}} \right)$.

    \item Given ability to measure observable $O\,;\; \|O\|\leq 1$, there exists a randomized quantum algorithm that returns $\frac{1}{q^2}\br{\vec{b}}A^{-1}OA^{-1}\kt{\vec{b}}$ up to additive error $\varepsilon$ with probability at least $1-\delta$, utilizing $\OC\left( \log \left(\frac{2}{\delta}\right)\frac{\|A^{-1}\|^4}{\varepsilon^2q^4}\log^2 \left(\frac{\|A^{-1}\|^2}{\varepsilon q^2} \right)\right)$ circuit runs each of gate depth $\OC\left(\|A^{-1}\|^2\lambda^2 \log^2\left(\frac{\|A^{-1}\|^2}{\varepsilon q^2}  \right) + d_{\vec{b}} \right)$.
    
    \item In the case $q=\|A^{-1}\kt{\vec{b}}\|$ and whose value is not given, there exists an auxiliary algorithm that approximates the value of $q$ for the above algorithms with the sample complexity $\OC\left( \log \left(\frac{2}{\delta}\right)\frac{\|A^{-1}\|^4}{\varepsilon^2q^4}\log^2 \left(\frac{\|A^{-1}\|}{\varepsilon q} \right)\right)$ and gate depth $\OC\left(\|A^{-1}\|^2\lambda^2 \log^2\left(\frac{\|A^{-1}\|}{\varepsilon q} \right) + d_{\vec{b}} \right)$.
\end{itemize}
\end{corollary}

\begin{proof}[Proof of Corollary \ref{cor:linear-systems} (Linear systems)]
In order to quantify the complexity of preparing $\frac{1}{q}\br{\psi}U f(A)V\kt{\psi}$, we now require $\alpha(\varepsilon q,A)$ and $t_{max}(\varepsilon q,A)$, as specified by Eq.~\eqref{eq:thm1-overlap} of Theorem \ref{thm:general-sampling}. Again inspecting Eqs.~\eqref{eq:t-linear-systems} and \eqref{eq:alpha-linear-systems} of Lemma \ref{lem:fourier-linear-systems}, these quantities scale as
\begin{equation}
    t_{max}(\varepsilon q,A)  = \Theta\left(\|A^{-1}\| \log \left(\frac{\|A^{-1}\|}{\varepsilon q} \right) \right)\,,
\end{equation}
and
\begin{equation}
    \alpha(\varepsilon q,A) =  \Theta\left( \|A^{-1}\| \sqrt{\log \left(\frac{\|A^{-1}\|}{\varepsilon q} \right)} \right)\,.
\end{equation}
Now, substituting the above into Eq.~\eqref{eq:thm1-overlap} of Theorem \ref{thm:general-sampling}, we obtain the desired result.

Similarly, in order to quantify the complexity of preparing $\frac{1}{q^2}\Tr\left[f(A)\rho f(A)^{\dag} O\right]$, we need to evaluate the quantities $\alpha\big(\frac{\varepsilon q^2}{\|f(A)\|},A\big)$ and $t_{max}\Big(\frac{\varepsilon q^2}{\|f(A)\|},A\Big)$, as specified by Eq.~\eqref{eq:thm1-expectation} of Theorem \ref{thm:general-sampling}. Inspecting Eqs.~\eqref{eq:t-linear-systems} and \eqref{eq:alpha-linear-systems} of Lemma \ref{lem:fourier-linear-systems}, we see that these quantities satisfy
\begin{equation}
    t_{max}\Big(\frac{\varepsilon q^2}{\|f(A)\|},A\Big) = \Theta\left(\|A^{-1}\| \log \left(\frac{\|A^{-1}\|^2}{\varepsilon q^2} \right) \right)\,,
\end{equation}
and
\begin{equation}
    \alpha\Big(\frac{\varepsilon q^2}{\|f(A)\|},A\Big) = \Theta\left( \|A^{-1}\| \sqrt{\log \left(\frac{\|A^{-1}\|^2}{\varepsilon q^2} \right)} \right)\,.
\end{equation}
Substituting these two expressions into Eq.~\eqref{eq:thm1-expectation} of Theorem \ref{thm:general-sampling}, we obtain the desired result.

Finally, we can use Proposition \ref{prop:norm-sampling} to characterize the complexity of approximating $\|A^{-1}\kt{\vec{b}}\|$. Eq.~\eqref{eq:norm-complexity} of Proposition \ref{prop:norm-sampling} is written in terms of $\alpha(\varepsilon q,A)$ and $t_{max}(\varepsilon q,A)$ whose scalings we have already quoted above. Substituting this into Eq.~\eqref{eq:norm-complexity}, we obtain the result in Corollary \ref{cor:linear-systems}.
\end{proof}

\begin{remark}[Non-Hermitian matrices]
Any non-Hermitian matrix $B$ can be embedded in a larger Hermitian matrix $A$ with the aid of a single qubit as 
\begin{equation}
    A = \begin{bmatrix}
0 & B \\
B^{\dag} & 0 
\end{bmatrix} \,.
\end{equation}
Then, one can verify that $A^{-1}\begin{pmatrix}\vec{b}\\ 0 \end{pmatrix} = B^{-1}\vec{b}$, $\|A\| = \|B\|$ and $ \|A^{-1}\| = \|B^{-1}\|$. In Pauli representation, given respective Hermitian and anti-Hermitian components $H(B)$ and $iH_2(B)$ of $B$ such that $B = H_1(B) + i H_2(B)$, this embedding can be explicitly written as
\begin{equation}
    A = X \otimes H_1(B) - Y \otimes H_2(B) = \frac{1}{2}X \otimes (B + B^{\dag}) - \frac{1}{2}Y \otimes (B - B^{\dag}) \,.
\end{equation}
From this it is clear that the Pauli weight of $A$ is bounded as $\lambda  \leq 2\lambda_B$.
\end{remark}

Using the above remark we see that the asymptotic complexities for the linear systems problem as stated in $A$ can be simply be translated to $B$ via the substitution $\lambda  \rightarrow \lambda_B$ and $\|A^{-1}\| \rightarrow \|B^{-1}\|$.


\subsection{Ground state sampling}\label{sec:appdx-gs}

In this section we study the task of sampling properties of the ground state of a given Hamiltonian. In order to establish our result the relevant function we consider is the Gaussian function $e^{-\frac{1}{2}\tau^2x^2}$. We refer to the following lemmas which show how this can obtain approximations of ground state observables.

\begin{lemma}[Ground state projection - adapted from Keen et al.~\cite{keen2021quantum}]\label{lem:appdx-GS-projector}
Suppose we have Hamiltonian $H=\sum_l E_l \dya{E_l}$ with all eigenvalues $E_l \geq 0$. Assume the spectral gap is lower bounded by $\Delta \leq E_1 - E_0$. Additionally, we suppose that we have an initial trial state $\kt{\psi_0}$ with overlap with the ground state $\gamma:= |\braket{\psi_0}{E_0}|$. Then, the state $\kt{\psi}=\kt{\widetilde{\psi}}/\|\kt{\widetilde{\psi}}\|$, where $\kt{\widetilde{\psi}} = e^{-\frac{1}{2}\tau^2H^2}\kt{\psi_0}$, satisfies
\begin{equation}
    1-|\braket{\psi}{E_0}| \leq \frac{1}{2}\varepsilon^2\,,
\end{equation}
for any $\tau\geq \tau_{\varepsilon}$ where $\tau_{\varepsilon}$ satisfies
\begin{equation}
    \tau_{\varepsilon} = \frac{1}{\Delta}\sqrt{2 \log\frac{1}{\varepsilon\gamma}}\,.
\end{equation}
\end{lemma}

With the following lemma we can see how this affects the closeness of expectation values.

\begin{lemma}[Tightness of expectation values, pt.~3]\label{lem:pure-expecation-overlap}
For some operator $O$ and pure states $\kt{\psi}$, $E_0$, we have
\begin{align}
    \Tr\big[\dya{\psi}O\big] - \Tr\big[\dya{E_0}O\big] \leq 2\sqrt{2}\|O\| \sqrt{1-|\braket{\psi}{E_0}|} \,.
\end{align}
\end{lemma}

\begin{proof}
We have
\begin{align}
    \Tr\big[\dya{\psi}O\big] - \Tr\big[\dya{E_0}O\big] &\leq  \|O\|_{\infty} \left\| \dya{\psi}-\dya{E_0} \right\|_1 \\
    &\leq 2\|O\|_{\infty} \sqrt{1-|\braket{\psi}{E_0}|^2} \\
    &= 2\|O\|_{\infty} \sqrt{\left(1-|\braket{\psi}{E_0}|\right)\left(1+|\braket{\psi}{E_0}|\right)} \\
    &\leq 2\sqrt{2}\|O\|_{\infty} \sqrt{1-|\braket{\psi}{E_0}|}\,,
\end{align}
where in the first line we use the tracial matrix H\"older's inequality, in the second line we use the relation between the trace distance and the fidelity for pure states, in the third line we complete the square, and in the final line we use the fact that $|\braket{\psi}{E_0}|\leq 1$.
\end{proof}

\begin{lemma}[Ground state observable projection]\label{cor:GS-projector}
Under the conditions specified in Lemma \ref{lem:appdx-GS-projector} with $\tau \geq \tau_{\varepsilon}$, the normalized state $\kt{\psi} = e^{-\frac{1}{2}\tau^2H^2}\kt{\psi_0}/\|e^{-\frac{1}{2}\tau^2H^2}\kt{\psi_0}\|$ satisfies
\begin{equation}
    \Tr\big[\dya{\psi}O\big] - \Tr\big[\dya{E_0}O\big] \leq 2\|O\| \varepsilon\,,
\end{equation}
for any measurement operator $O$.
\end{lemma}  

\begin{proof}
    This follows as a direct implication of Lemmas \ref{lem:appdx-GS-projector} and \ref{lem:pure-expecation-overlap}.
\end{proof}

We now introduce the Hubbard-Stratonovich transformation \cite{stratonovich1957method, hubbard1959calculation}. This gives us a way to decompose the operator $e^{-\frac{1}{2} \tau^2 {H}^2}$ into a linear combination of implementable unitaries. It states that, for Hermitian $H$, we have
\begin{equation}
    e^{-\frac{1}{2} \tau^2 {H}^2}=\frac{1}{\sqrt{2 \pi}} \int_{-\infty}^{\infty} d z e^{-\frac{1}{2} z^2} e^{-i z \tau {H}}\,.
\end{equation}

The following lemma, adapted from Ref.~\cite{keen2021quantum}, shows that the integral in the  Hubbard-Stratonovich transformation can be discretized and truncated to give an approximate (discrete) Fourier series for the operator $e^{-\frac{1}{2} \tau^2 {H}^2}$.

\begin{lemma}[Approximate Hubbard-Stratonovich transformation -- adapted from Appendix A2 of Ref.~\cite{keen2021quantum}]\label{lem:fourier-gs}
Truncating and discretizing the integral, and assuming $\|H\|\leq 1$, we have a Fourier series approximation
\begin{equation}
    \left\|\frac{1}{\sqrt{2 \pi}} \sum_{k=-N_z}^{N_z} \Delta_z e^{-\frac{1}{2} z_k^2} e^{-i z_k \tau {H}} - e^{-\frac{1}{2} \tau^2 {H}^2} \right\| \leq \varepsilon\,,
\end{equation}
for choice of $\Delta_z = \OC\left(\tau^{-1}\right)$ and $\Delta_z N_z = \sqrt{2\log(\frac{2}{\varepsilon})}$, where we have denoted $z_k=k\Delta_z$. This has maximum time evolution parameter
\begin{align}\label{eq:t-gs}
    t_{\max}(\varepsilon) := \max_k{(z_k\tau)} = \Delta_z N_z \tau = \tau\sqrt{2\log\left(\frac{2}{\varepsilon}\right)}\,.
\end{align}
Further, the coefficients $\{\alpha_k(\varepsilon)\}_k=\{ \frac{1}{\sqrt{2\pi}}\Delta_z e^{-\frac{1}{2} z_k^2} \}_k$ have weight 
\begin{equation}\label{eq:alpha-gs}
    \alpha(\varepsilon):=\sum_k|\alpha_k(\varepsilon)| \leq 1 + \frac{1}{\sqrt{2\pi}}\Delta_z \,.
\end{equation}
\end{lemma}

\begin{proof}
The bound on $\alpha(\varepsilon)$ can be seen by noting that 
\begin{equation}
    \sum_k|\alpha_k(\varepsilon)| = \frac{1}{\sqrt{2 \pi}} \sum_{k=-N_z}^{N_z} \Delta_z e^{-\frac{1}{2} z_k^2} \leq \frac{1}{\sqrt{2 \pi}} \sum_{k=-\infty}^{\infty} \Delta_z e^{-\frac{1}{2} z_k^2}\,,
\end{equation}
is simply a discretized Gaussian integral, where the discretization error can be bounded by the step size $\Delta_z$ multiplied by the maximum value of the function $\frac{1}{\sqrt{2 \pi}}$. We refer the reader to Appendix A2 of Ref.~\cite{keen2021quantum} for the rest of the proof of the claim on $t_{max}(\varepsilon)$.
\end{proof}

From hereon we will make the soft imposition that we choose $\Delta_z \leq 1$. This gives an error-independent bound for Lemma \ref{lem:fourier-gs} of $\alpha(\varepsilon) \leq 1 + \frac{1}{\sqrt{2\pi}}< 1.4$.

\begin{remark}[Hamiltonians with non-positive spectra]
For Hamiltonians with non-positive spectra and where we have upper bound on magnitude of ground state energy $\lambda_0 \geq |E_0|$, we can shift the spectrum $H' = H + E_0 \id$ without changing the weight of the coefficients or the Hamiltonian simulation problem. This is because in the Hubbard-Stratonovich transformation the Hamiltonian appears in the term $e^{-i z_k \tau {H'}} = e^{-i z_k \tau {H}}e^{-i z_k \tau E_0{\id}} = e^{-i z_k \tau E_0}e^{-i z_k \tau {H}}$ and the phase factor $e^{-i z_k \tau E_0}$ can be absorbed into the coefficients $\alpha_k$ without changing the weight $\alpha$. Further, the eigenstates of $H'$ are clearly eigenstates of $H$, with shifted eigenenergies.
\end{remark}

We can now present our result for ground state property estimation (Corollary \ref{cor:ground-state} in the main text).

\begin{corollary}[Ground state property estimation - detailed version]\label{cor:appdx-ground-state}
Consider a Hamiltonian $H=\sum_l E_l \dya{E_l}$ with all eigenvalues $E_l \geq 0$, and known Pauli decomposition $H=\sum_{\ell}a_{\ell}P_{\ell}$ with $\lambda:=\sum_{\ell}|a_{\ell}|$. Assume the spectral gap is lower bounded by $\Delta \leq E_1 - E_0$. Additionally, we suppose that we have an initial trial state $\kt{\psi_0}$ with overlap with the ground state $\gamma:= |\braket{\psi_0}{E_0}|$. Finally, we assume that $E_0\leq {\Delta}/\sqrt{2 \log\frac{\|O\|}{\varepsilon\gamma}}$. Then, given ability to measure observable $O$, there exists a random algorithm that returns $\br{E_0}O\kt{E_0}$ up to additive error $\varepsilon$ with probability at least $(1-\delta)^2$, which consists of 
\begin{itemize}
    \item A core routine that requires $\OC\left(\log\left(\frac{2}{\delta}\right) \frac{\|O\|^2}{\varepsilon^2\gamma^4}\right)$ circuit runs of the form in Figure \ref{fig:circuits}(a) each of non-Clifford depth at most $\OC\left(\frac{\lambda^2}{\Delta^2}\log^2\left(\frac{\|O\|}{\varepsilon\gamma^2}\right) \right)$.
    \item A subroutine that approximates the normalization constant, using asmyptotically equivalent resources to the core routine.
\end{itemize}

\end{corollary}

\begin{proof}[Proof of Corollary \ref{cor:ground-state}]
We first presume the state normalization constant $q=\|e^{-\frac{1}{2} \tau^2 {H}^2}\kt{\psi_0}\|$ is given to us exactly, and use Lemma \ref{cor:GS-projector} to show how to obtain an $\varepsilon/6$-additive statistical approximation to $\frac{1}{q^2}\br{\psi_0}e^{-\frac{1}{2} \tau^2 {H}^2}O e^{-\frac{1}{2}\tau^2 {H}^2} \kt{\psi_0}$. We then use Proposition \ref{prop:appdx-norm-sampling} to consider the overhead of approximating $q$, which will relax the approximation error to $\varepsilon/2$. This gives the desired $\varepsilon$-approximation to $\br{E_0}O\kt{E_0}$ for appropriate choice of $\tau$ as specified by Lemma \ref{lem:appdx-GS-projector}.

In order to quantify the complexity of approximating $\frac{1}{q^2}\br{\psi_0}e^{-\frac{1}{2} \tau^2 {H}^2}O e^{-\frac{1}{2}\tau^2 {H}^2} \kt{\psi_0}$, we need to evaluate the quantities $\alpha\big(\frac{\varepsilon q^2}{\|O\|\|e^{-\frac{1}{2}\tau^2 {H}^2}\|},H\big)$ and $t_{max}\Big(\frac{\varepsilon q^2}{\|O\|\|e^{-\frac{1}{2}\tau^2 {H}^2}\|},H\Big)$, as specified by Eq.~\eqref{eq:thm1-expectation} of Theorem \ref{thm:general-sampling}, {and upper bound $1/q$. We recall from the above discussion, by choosing that the step size $\Delta_z$ in the Fourier series in Lemma \ref{lem:fourier-gs} be less than 1 we fix $\alpha\big(\frac{\varepsilon q^2}{\|O\|\|e^{-\frac{1}{2}\tau^2 {H}^2}\|},H\big)=\OC(1)$. Moreover, we have $\frac{1}{q}\leq \frac{1}{\gamma}e^{\frac{1}{2} \tau^2 {E_0}^2}\leq \frac{2}{\gamma}$ for $E_0 \tau \leq 1$ as $e^{-\frac{1}{2} \tau^2 {E_0}^2} \geq 1-\frac{1}{2}\tau^2 {E_0}^2 \geq \frac{1}{2}$. Therefore,} the sample complexity can be evaluated simply as
\begin{align}
    {\CC}^O_{\emph{sample}} = \OC\left( \log \left(\frac{2}{\delta}\right)\frac{\|O\|^2}{\varepsilon^2 q^4} \right) = \OC\left( \log \left(\frac{2}{\delta}\right)\frac{\|O\|^2}{\varepsilon^2 \gamma^4} \right)\,.
\end{align}

{For the gate complexity,} using Eq.~\eqref{eq:t-gs} of Lemma \eqref{lem:fourier-gs} we have
\begin{align}
    t_{max}\Big(\frac{\varepsilon q^2}{\|O\|\|e^{-\frac{1}{2} \tau^2 {H}^2}\|},H\Big) &= \tau\sqrt{\log\left(\frac{\|O\|\|e^{-\frac{1}{2} \tau^2 {H}^2}\|}{\varepsilon q^2}\right)} \\
    & \leq \tau\sqrt{\log\left(\frac{\|O\|}{\varepsilon \gamma^2}\right)} \,,
\end{align}
where in the second line we have again used the fact that $\frac{1}{q}\leq \frac{1}{\gamma}e^{\frac{1}{2} \tau^2 {E_0}^2}$. 
Suppose we would like the exact approximation error to be $\varepsilon/2$, that is, $\big|\frac{1}{q^2}\br{\psi_0}e^{-\frac{1}{2} \tau^2 {H}^2}O e^{-\frac{1}{2}\tau^2 {H}^2}\kt{\psi_0} - \br{E_0}O\kt{E_0}\big|\leq \varepsilon/2$. Lemma \ref{cor:GS-projector} specifies that, in order to satisfy this, it is sufficient to have $\tau= \tau_{\varepsilon/2\|O\|} = \frac{1}{\Delta}\sqrt{2 \log\frac{2\|O\|}{\varepsilon\gamma}}$. Thus, overall we have
\begin{align}
    t_{max}\Big(\frac{\varepsilon q^2}{\|O\|\|e^{-\frac{1}{2} \tau^2 {H}^2}\|},H\Big) \leq \frac{1}{\Delta}\sqrt{2 \log\left(\frac{2\|O\|}{\varepsilon\gamma} \right)}\sqrt{\log\left(\frac{\|O\|}{\varepsilon \gamma^2}\right)} = \OC\left(\frac{1}{\Delta}\log\left(\frac{\|O\|}{\varepsilon \gamma^2}\right)\right)\,,
\end{align}
where we have used the fact that $\gamma \leq 1$. Substituting this into Eq.~\eqref{eq:thm1-expectation} of Theorem \ref{thm:appdx-general-sampling} we obtain the stated result for non-Clifford gate complexity.

Finally, we consider the complexity of approximating the state norm $q=\|e^{-\frac{1}{2} \tau_{\varepsilon/\|O\|}^2 {H}^2}\kt{\psi_0}\|$ by following Proposition \ref{prop:appdx-norm-sampling}. Eq.~\eqref{eq:norm-complexity} of Proposition \ref{prop:appdx-norm-sampling} expresses complexities in terms of $\alpha(\varepsilon q, H)$ and $t_{max}(\varepsilon q, H)$. From the above discussion we have that $\alpha(\varepsilon q, H)=\OC(1)$. Using Eq.~\eqref{eq:t-gs} and Corollary \ref{cor:ground-state} we can assign
\begin{align}
    t_{max}(\varepsilon q, H) &= \tau \sqrt{2\log\left(\frac{2}{\varepsilon q}\right)} \\
    &= \frac{1}{\Delta}\sqrt{2 \log\frac{\|O\|}{\varepsilon\gamma}}\sqrt{2\log\left(\frac{2}{\varepsilon q}\right)} \\
    & = \OC\left(\frac{1}{\Delta} \sqrt{ \log\frac{\|O\|}{\varepsilon\gamma} \log\frac{4}{\varepsilon\gamma}} \right) \,,
\end{align}
where in the final line we have used the fact that $\frac{1}{q}\leq \frac{2}{\gamma}$. This leads to the stated complexities for the normalization subroutine of 
\begin{equation}
    \CC^{norm}_{gate} = \OC\left(\log\left(\frac{2}{\delta}\right) \frac{1}{\gamma^4\varepsilon^2}\right)\,,\quad {\CC}^{norm}_{sample} =\OC\left(\frac{\lambda^2}{\Delta^2}\log\frac{\|O\|}{\varepsilon\gamma} \log\frac{4}{\varepsilon\gamma}\right)\,.
\end{equation}

\end{proof}


\subsubsection{Power method}\label{appdx:power-method}

In this section we explore the feasibility of using a randomized scheme based on the power method to find dominant eigenvalues. Quantum algorithms for estimating eigenvalues via the power method have previously been studied \cite{nghiem2023quantum, seki2021quantum}, though not in a randomized setting starting from Pauli access.  Given an observable of interest, our aim is to approximate $\br{E_0}O\ket{E_0}$ with $\br{\psi_0}H^kOH^k\ket{\psi_0}/\|H^k\ket{\psi_0}\|^2$, where $\ket{E_0}$ is the ground state of Hamiltonian $H$, $\ket{\psi_0}$ is a given trial state with overlap $\gamma:=|\braket{E_0}{\psi_0}|$, and $\|H^k\ket{\psi_0}\|^2$. With this, we have the following proposition: 

\begin{supprop}[Power method]\label{lem:appdx-power-method}
For Hamiltonians $H$ with negative spectra and first excited energy $E_1<0$ we have
\begin{align}
    \left|\frac{\br{\psi_0}H^kOH^k\ket{\psi_0}}{\|H^k\ket{\psi_0}\|^2} - \br{E_0}O\ket{E_0} \right| \leq \varepsilon\,,
\end{align}
for a given error parameter $\varepsilon\leq \frac{4\|O\|\sqrt{1-\gamma^2}}{\gamma}$ if $k$ satisfies
\begin{align}
    k = \Omega\left(  \frac{\log\Big( \frac{\gamma\varepsilon}{\|O\|\sqrt{1-\gamma^2}}\Big)}{\log\left( 1- \frac{\Delta}{|E_0|} \right)} \right) = \Omega\bigg( \frac{|E_0|}{\Delta} \log\bigg( \frac{\|O\|\sqrt{1-\gamma^2}}{\gamma\varepsilon}\bigg) \bigg) \,. \label{eq:expectation-overlap-3}
\end{align}
\end{supprop}

\begin{proof}
Consider the decomposition of the trial state $\kt{\psi_0}$ in the energy eigenbasis of $H$ as
\begin{align}
    \kt{\psi_0} = c_0 \kt{E_0} + \sum_{j>0}c_j \kt{E_j}\,,
\end{align}
where $\{c_k\}_k$ are coefficients. Using this, we can express $H^k\kt{\psi_0}$ as 
\begin{align}
    H^k\kt{\psi_0} = c_0 E_0^k\kt{E_0} + \sum_{j>0}c_jE_j^k \kt{E_j}\,,
\end{align}
and thus the overlap of the normalized power method approximation with the true ground state
\begin{align}
    1- \left|\frac{\br{E_0}H^k\kt{\psi_0}}{\|H^k\ket{\psi_0}\|} \right| &= 1 - \frac{|c_0E_0^k|}{\sqrt{c_0^2E^{2k}+ \sum_{j>0}c_j^2 E_j^{2k}}} \\
    &= 1 - \bigg[ 1 + \sum_{j>0} \frac{c_j^2}{c_0^2} \left( \frac{E_j}{E_0} \right)^{2k} \bigg]^{-1/2} \\
    &\leq 1 - \bigg[ 1 + \frac{1-\gamma^2}{\gamma^2} \left( \frac{E_1}{E_0} \right)^{2k} \bigg]^{-1/2} \\
    &\leq 1 - \bigg[ 1 - \frac{1}{2}\frac{1-\gamma^2}{\gamma^2} \left( \frac{E_1}{E_0} \right)^{2k} \bigg] \\
    &= \frac{1-\gamma^2}{2\gamma^2} \left( \frac{E_1}{E_0} \right)^{2k}\,,
\end{align}
where in the first inequality we have used the fact that $\gamma=E_0^2$ and $\sum_{j>0}c_j^2 = 1-\gamma^2$, and the second is due to the fact that $(1+x)^{1/2}\geq 1+x/2$ for $x\geq -1$. Thus, in order to constrain the state overlap $1- \left|\frac{\br{E_0}H^k\kt{\psi_0}}{\|H^k\ket{\psi_0}\|} \right| = \varepsilon'$ each of the following conditions are sufficient 
\begin{align}
    &1 - \frac{|c_0E_0^k|}{\sqrt{c_0^2E^{2k}+ \sum_{j>0}c_j^2 E_j^{2k}}} \leq \varepsilon'\,, \\
    \Rightarrow&\quad \left(1-\frac{\Delta}{|E_0|} \right)^{2k} \leq \frac{2\gamma^2\varepsilon'}{1-\gamma^2}\,,\\
    \Rightarrow&\quad\quad k\geq \frac{1}{2} \frac{\log\left( \frac{2\gamma^2\varepsilon'}{1-\gamma^2}\right)}{\log\left( 1- \frac{\Delta}{|E_0|} \right)}\,.
\end{align}
We recall that Lemma \ref{lem:pure-expecation-overlap} relates the overlap of pure state expectation values to the overlap of states, which specifies that the state overlap $\varepsilon'$ must be at most $\frac{\varepsilon^2}{8\|O\|^2}$ in order to constrain the expectation value within additive error $\varepsilon$. Using this, we obtain the first equality of Eq.~\eqref{eq:expectation-overlap-3}, where starting our assumption on $\varepsilon$ ensures the bound on $k$ is positive and thus meaningful. The second equality can be established by noting that
\begin{align}
    1-\frac{|E_0|}{\Delta} \leq \frac{1}{\log \big( 1-\frac{\Delta}{|E_0|} \big)} \leq - \frac{|E_0|}{\Delta}\,,
\end{align}
for $E_1<0$.
\end{proof}

Now we suppose that we are given a recipe to prepare the trial state by a series of Clifford gates or we have access to the amplitudes (for which Proposition \ref{prop:appdx-compiling-vector} can be used). The above lemma along with Proposition \ref{prop:appdx-classical} implies that, in this setting, there is a classical algorithm that solves the ground state property estimation problem with exponential number of samples in $\Delta^{-1}$ and polynomial in all other parameters.


\subsection{Gibbs state property estimation}\label{appdx:gibbs}

In this section we demonstrate how the Fourier decomposition of the exponential function found in Ref.~\cite{holmes2022quantum} leads to a randomized quantum algorithm to sample properties of the Gibbs state in our scheme. In Ref.~\cite{holmes2022quantum}, the authors assume one has access to the purification of the Gibbs state of some intermediate Hamiltonian $H_0$, and aims to construct an approximation of the Gibbs state of some other Hamiltonian $H$. More precisely, we start with the state
\begin{equation} \label{eq:thermofield-double}
\ket{\Psi_0} = \frac{1}{    \sqrt{\ZC_0}} \sum_i e^{-\beta E_{0,i}/2} \ket{E_{0,i}} \ket{E_{0,i}^*} \in \HC_A \otimes \HC_B\,,
\end{equation}
where $\ZC_0:=\Tr[e^{-\beta H_0}]$ is the partition function for $H_0$, and $\{\ket{E_{0,i}}\}_i$ are the eigenstates of $H_0$ with corresponding eigenvalues $\{E_{0,i}\}_i$ ($\{\ket{E_{0,i}^*}\}_i$ are the eigenstates of $H_0^*$, the complex conjugate of $H_0$ in the computational basis). This state satisfies
\begin{equation}
    \Tr_A[\ketbra{\Psi_0}{\Psi_0}] = \Tr_B[\ketbra{\Psi_0}{\Psi_0}] = \frac{e^{-\beta H_0}}{\ZC_0}\,,
\end{equation}
that is, it is the purification of the Gibbs state of $H_0$. The goal of Ref.~\cite{holmes2022quantum} is then to prepare the Gibbs state of $H$
\begin{equation}
    \Tr_A[\ketbra{\Psi}{\Psi}] = \Tr_B[\ketbra{\Psi}{\Psi}] = \frac{e^{-\beta H}}{\ZC}=:\gamma_{\beta}\,,
\end{equation}
where $\ZC:=\Tr[e^{-\beta H}]$. As our algorithmic framework only allows the estimation of observable and state overlaps, our randomized algorithm will approximate the expectation value $\Tr[\gamma_{\beta} O]$ for a given measurement operator $O$.

It is observed that the so-called "work operator" $W:=H\otimes \id + \id \otimes H^*_0$ enables the transformation
\begin{align}     
    \sqrt{\frac{\ZC_0}{\ZC}} e^{-\beta W/2} \ket{\Psi_0} = \ket{\Psi} \,,
\end{align}
where $\sqrt{{\ZC_0}/{\ZC}} = 1/\|e^{-\beta W/2} \ket{\Psi_0} \|$ is the normalization factor. Thus, the relevant function of interest for us here is $e^{-\beta W/2}$. Ref.~\cite{holmes2022quantum} presents an Fourier decomposition of $e^{-\beta W/2}$ as follows.

\begin{lemma}[Exponential operator - adpated from Lemmas 3.1, 3.2, 3.4 of Ref.~\cite{holmes2022quantum}]\label{lem:zoe}
There exists an LCU Fourier decomposition $X = \sum_{j=0}^{2J} \alpha_j e^{i\tau_jW}$ that satisfies
\begin{align}\label{eq:gibbs-td}
    \left\| \frac{\Tr_A[X\ketbra{\Psi_0}{\Psi_0}X^{\dag}]}{\|X\ket{\Psi_0}\|^2} - \gamma_{\beta} \right\|_1 \leq 2\varepsilon\,,
\end{align}
where if $[H_0,H]=0$ the parameters in the decomposition satisify
\begin{align}
    \alpha &:= \sum_{j=0}^{2J} |\alpha_j| \leq  2 e^{\max \{ 4, \sqrt{\ln 6/\varepsilon} \}} e^{\beta\|V\|/2} \,, \\
    \|X \ket{\Psi_0}\| &\geq \frac{1}{2} \|e^{-\beta W/2} \ket{\Psi_0}\|  =  \frac{1}{2}\sqrt{\frac{\ZC}{\ZC_0}}  \,, \\
    \tau_{\max} &:=\max_j |\tau_j| =  \frac{\pi \beta}{z} \left( \ceil{\frac{1}{3}z^{3/2}} -1 \right)\,,
\end{align}
where $z \leq \beta(\|W\| + \|V\|) + 2(\max \{ 4, \sqrt{\ln 6/\varepsilon} \})^2$ and we denote $V := H-H_0$.
\end{lemma}

We remark that due to the tracial matrix H\"older's inequality, Eq.~\eqref{eq:gibbs-td} implies that the expectation value with respect to any observable $O$ is close up to additive error $2\|O\|\varepsilon$. Lemma \ref{lem:zoe} already specifies the conditions required to approximate observables with respect to the Gibbs state. Thus, we will not need to make use of Theorem \ref{thm:general-sampling} here, and can directly use Proposition \ref{prop:fourier-sampling}. Eq.~\eqref{eq:gibbs-td} can also be satisfied for non-commuting Hamiltonians, and such a setting can also be transported to our framework. For simplicity, we only detail the commuting case here. 

\begin{corollary}[Gibbs state property estimation - detailed version]\label{cor:appdx-gibbs-state}
    Suppose access to the quantum state $\ket{\Psi_0}$ defined in Eq.~\eqref{eq:thermofield-double} and the ability to measure the observable $O$. Further, suppose the Pauli decompositions of $H_0$ and $H$ are known, and $[H_0,H]=0$. Then, we give a randomized quantum algorithm to approximate $\Tr[\gamma_{\beta} O]$ to additive error $\varepsilon$ and success probability at least $(1-\delta)^2$, utilizing:
    \begin{itemize}
        \item a core routine using $\OC(\log(\frac{1}{\delta})  \frac{e^{\sqrt{\ln \|O\|/\varepsilon} }}{\varepsilon^2} \frac{\ZC_0^2}{\ZC^2} e^{2\beta\|V\|} )$ circuit runs each of non-Clifford depth at most $\OC\big(\lambda_W^2\beta^3\big(\|W\| + \|V\| + \log\big(\frac{\|O\|}{\varepsilon}\big)\big)\big)$\,.
        \item a subroutine to give the appropriate normalization, using asymptotically equivalent resources to the core routine.  
    \end{itemize}
\end{corollary}

\begin{proof}
Being explicit, we first note that Lemma \ref{lem:zoe} straightforwardly implies that there exists an LCU  Fourier decomposition $X = \sum_{j=0}^{2J} \alpha_j e^{i\tau_jW}$ for which the state
\begin{equation}
    \rho = \frac{\Tr_A[X\ketbra{\Psi_0}{\Psi_0}X^{\dag}]}{\|X\ket{\Psi_0}\|^2}
\end{equation}
satisfies
\begin{align}
    \big| \Tr[\rho O] - \Tr[\gamma_{\beta} O] \big| \leq \varepsilon/2\,,
\end{align}
where the parameters in the decomposition satisify
\begin{align}
    \alpha &:= \sum_{j=0}^{2J} |\alpha_j| \leq  2 e^{\max \{ 4, \sqrt{\ln 24\|O\|/\varepsilon} \}} e^{\beta\|V\|/2} \,, \\
    \|X \ket{\Psi_0}\| &\geq \frac{1}{2} \|e^{-\beta W/2} \ket{\Psi_0}\|  =  \frac{1}{2}\sqrt{\frac{\ZC}{\ZC_0}}  \,, \\
    \tau_{\max} &:=\max_j |\tau_j| =  \frac{\pi \beta}{z} \left( \ceil{\frac{1}{3}z^{3/2}} -1 \right)\,,
\end{align}   
where $z \leq \beta(\|W\| + \|V\|) + 2(\max \{ 4, \sqrt{\ln 24\|O\|/\varepsilon} \})^2$. Proposition \ref{prop:appdx-fourier-sampling} and Lemma \ref{lem:appdx-q-approximation} then give the resources required to statistically approximate $\Tr[\rho O]$ to additive error $\varepsilon/2$. By the triangle inequality, this implies a statistical approximation of the exact answer $\Tr[\gamma_{\beta} O]$ to additive error $\varepsilon$. We detail this below.

We start by presuming that $\|X\ket{\Psi_0}\|$ is known exactly. Then, with this exact quantity Proposition \ref{prop:appdx-fourier-sampling} gives a randomized algorithm to statistically approximate $\Tr[\rho_1 O]$ to additive error $\varepsilon/6$ using 
\begin{align}
    \CC_{sample} = 2e^4\log\left( \frac{2}{\delta}\right) \frac{36 \|O\|^2 \alpha^4}{\|X\ket{\Psi_0}\|^4 \varepsilon^2} \leq 18432e^4 \log\left( \frac{2}{\delta}\right) \frac{\|O\|^2}{\varepsilon^2} \frac{\ZC_0^2}{\ZC^2} e^{4\max \{ 4, \sqrt{\ln 24\|O\|/\varepsilon} \}} e^{2\beta\|V\|}
\end{align}
circuits, each of non-Clifford depth at most
\begin{align}
    \CC_{gate} = 2\lambda_W^2 \tau_{\max}^2 \leq 2\lambda_W^2 \frac{\pi^2 \beta^2}{z^2} \left( \ceil{\frac{1}{3}z^{3/2}} -1 \right)^2\,,
\end{align}
and thus our stated result follows. 

Now we check the complexity of approximating $\|X\ket{\Psi_0}\|^2$. We recall that Lemma \ref{lem:appdx-q-approximation} specifies that in order to statistically approximate $\Tr[\rho O]$ to additive error $\varepsilon/2$, one requires the above stated conditions (an approximation of $\Tr[\rho O]$ to additive error $\varepsilon/6$), and an approximation to $\|X\ket{\Psi_0}\|^2$ with additive error $\frac{1}{12}\frac{\varepsilon \|X\ket{\Psi_0}\|^2 }{\|O\|}$. Thus, (once again using Proposition \ref{prop:appdx-fourier-sampling}) we see it is sufficient to use
\begin{align}
    \CC^{norm}_{sample} = 2e^4\log\left( \frac{2}{\delta}\right) \frac{144 \|O\|^2 \alpha^4}{\|X\ket{\Psi_0}\|^4 \varepsilon^2} \leq 73728e^4 \log\left( \frac{2}{\delta}\right) \frac{ \|O\|^2}{\varepsilon^2} \frac{\ZC_0^2}{\ZC^2} e^{4\max \{ 4, \sqrt{\ln 24\|O\|/\varepsilon} \}} e^{2\beta\|V\|}
\end{align}
circuits, each of non-Clifford depth at most
\begin{align}
    \CC^{norm}_{gate} = 2\lambda_W^2 \tau_{\max}^2 \leq 2\lambda_W^2 \frac{\pi^2 \beta^2}{z^2} \left( \ceil{\frac{1}{3}z^{3/2}} -1 \right)^2\,.
\end{align}
We see that this equivalent to the complexity of the core routine, up to a constant factor. 
\end{proof}

We remark that, as with all our algorithms, the constant factor appearing in the sample complexity can be refined slightly by dividing the statistical and exact contributions to error unevenly. For instance, by choosing the statistical contribution to be $90\%$ of the total error (rather than $50\%$) the sample complexity for the normalization constant becomes $\CC^{norm}_{sample}  \leq 22756e^4 \log\left( \frac{2}{\delta}\right) \frac{ \|O\|^2}{\varepsilon^2} \frac{\ZC_0^2}{\ZC^2} e^{4\max \{ 4, \sqrt{\ln 120\|O\|/\varepsilon} \}} e^{2\beta\|V\|}$.
    

\subsection{Evaluating Green's functions}\label{appdx:greens}

In this section we present the detailed statement of our result for evaluating Green's functions as defined in Eqs.~\eqref{eq:greens1} and \eqref{eq:greens2} (Proposition \ref{cor:greens-functions} in the main text).

\begin{proposition}[Green's function estimation -- detailed version]\label{cor:appdx-greens-functions}
Consider a Hamiltonian $H=\sum_l E_l \dya{E_l}$ with all eigenvalues $E_l \geq 0$, and known Pauli decomposition $H=\sum_{\ell}a_{\ell}P_{\ell}\,;\; \lambda:=\sum_{\ell}|a_{\ell}|$. Assume the spectral gap is lower bounded by $\Delta \leq E_1 - E_0$. Additionally, we suppose that we can freely prepare an initial trial state $\kt{\psi_0}$ with overlap with the ground state $\gamma:= |\braket{\psi_0}{E_0}|$. Given parameters $\omega, \eta$ and the ground state energy $E_0$, there exists a random compiler that returns Eqs.~\eqref{eq:greens1} and \eqref{eq:greens2} up to additive error $\varepsilon$ with probability at least $(1-\delta)^2$, utilizing $\OC\left(\log\left(\frac{2}{\delta}\right) \frac{\|(\Gamma^{(\pm)})^{-1}\|^2}{\gamma^4\varepsilon^2} \log\left( \frac{\|(\Gamma^{(\pm)})^{-1}\|}{\varepsilon}\right)\right)$ circuit runs respectively, each of non-Clifford depth at most $\OC\left(\frac{\lambda_H^2}{\Delta^2}\log^2\left(\frac{2\|(\Gamma^{(\pm)})^{-1}\|}{\varepsilon\gamma^2}\right) +  (|\hbar\omega \pm E_0|+\eta+ \lambda_{H})^2\|\Gamma^{(\pm)-1}\|^2 \log^2 \left(\frac{\|\Gamma^{(\pm)-1}\|}{\varepsilon}\right) \right)$ with a normalization subroutine that has  smaller complexities than the main algorithm.
\end{proposition}


\begin{proof}[Proof of Proposition \ref{cor:greens-functions}.]
We recall that, from Sections \ref{sec:linear-systems} and \ref{sec:gs}, we have decompositions of the matrix inverse function and Gaussian function in terms of Pauli gates and Pauli rotations, provided decomposition of the matrix in terms of Pauli operators.

Namely, by combining Lemma \ref{lem:appdx-expectation-overlap} and Lemma \ref{lem:fourier-gs} we have for some Hamiltonian
\begin{equation}\label{eq:gf-gs-overlap}
    \left\|e^{-\frac{1}{2} \tau^2 {H}^2} - f_{GS}(\varepsilon,H) \right\| \leq \varepsilon\,,
\end{equation}
where $f_{GS}(\varepsilon,H)=\sum_{i \in S_{\varepsilon,H}}h_i(\varepsilon,H)\,V_i(\varepsilon,H)$ is a linear combination of gates with weight $R_{GS}:=\sum_{i \in S_{\varepsilon,H}}|h_i(\varepsilon,H)|$. The weight of this linear combination satisfies $R_{GS} = \Theta\left( 1 \right)$
and each $V_i$ has non-Clifford gate depth at most $\Theta\left(\lambda_H^2 \tau^2 \log\left(\frac{2}{\varepsilon} \right)\right)$ where $\lambda_{H}:=\sum_{\ell} |a_{\ell}|$ is the Pauli weight of $H$. Moreover, we recount that Corollary \ref{cor:GS-projector} states that applying the operator $e^{-\frac{1}{2} \tau^2 {H}^2}$ (after normalization) with $\tau = \frac{1}{\Delta}\sqrt{2 \log\frac{1}{\varepsilon\gamma}}$ approximately projects to the ground state with $2\|O\|\varepsilon$ additive error when considering an expectation value of $O$.

Likewise, by combining Lemma \ref{lem:appdx-kianna-ham-simulation} and Eq.~\eqref{eq:fourier-operator} in Lemma \ref{lem:fourier-linear-systems}, given some operator $\Gamma = \sum_l a_l P_l$ we have
\begin{equation}\label{eq:gf-inv-overlap}
    \Big\|\Gamma^{-1} - f_{inv}(\varepsilon,\Gamma) \Big\| \leq \varepsilon\,,
\end{equation}
where $f_{inv}(\varepsilon,\Gamma)=\sum_{i \in S_{\varepsilon,\Gamma}}\gamma_i(\varepsilon,\Gamma)\,U_i(\varepsilon,\Gamma)$ is a linear combination of gates with weight $R_{inv}:=\sum_{i \in S_{\varepsilon,\Gamma}}|\gamma_i(\varepsilon,\Gamma)|$ satisfying 
\begin{equation}
    R_{inv} = \Theta\left( \|\Gamma^{-1}\| \sqrt{\log \left(\frac{\|\Gamma^{-1}\|}{\varepsilon} \right)} \right)\,,
\end{equation}
and each $U_i$ has non-Clifford gate depth at most $\Theta\left(\lambda_{\Gamma}^2\|\Gamma^{-1}\|^2 \log^2 \left(\frac{\|\Gamma^{-1}\|}{\varepsilon} \right) \right)$, where $\lambda_{\Gamma}$ is the Pauli weight of $\Gamma$.

We now specifically consider the Green's function problem. We wish to find the inverse of the (non-Hermitian) operator $\Gamma^{(\pm)} = (\hbar\omega\pm E_0)\id \mp H + i\eta\id$. Following the remark in Section \ref{sec:linear-systems}, this can be embedded in a Hermitian operator by dilating the space and considering $Y^{(\pm)}= X \otimes \left((\hbar\omega\pm E_0)\id \mp H\right) + Y \otimes \eta\id$. These operators have Pauli weight $\lambda_{Y^{(\pm)}} \leq |\hbar\omega \pm E_0| + \lambda_H + \eta $, and satisfy $\|Y^{(\pm)}\| = \|\Gamma^{(\pm){-1}}\|$. We can then write the quantity we wish to prepare as 
\begin{equation}
    G^{(\pm)} = \br{1,E_0}\hat{a}_i\left(Y^{(\pm)}\right)^{-1}\hat{a}^{\dag}_j\kt{0,E_0}\,.
\end{equation}

In order to sample from the ground state, we will need to approximate the normalization constant $q^2:=\|e^{-\frac{1}{2}\tau^2H^2}\kt{\psi_0}\|^2$, where $\kt{\psi_0}$ is the trial ground state. We denote the statistical approximation of this quantity with $M_q$ shots as $Q^{(M_q)}$. We denote the statistical approximation of $\br{1,\psi_0}e^{-\frac{1}{2}\tau^2H^2}\hat{a}_i(Y^{(\pm)})^{-1}\hat{a}^{\dag}_je^{-\frac{1}{2}\tau^2H^2}\kt{0,\psi_0}$ as $E^{M}$. We can then express the full approximation error as
\begin{align}
    \left| \frac{1}{Q^{(M_q)}}E^{M} - \br{1,E_0}\hat{a}_i\left(Y^{(\pm)}\right)^{-1}\hat{a}^{\dag}_j\kt{0,E_0} \right| \leq& \\ 
    &\hspace{-18em}\leq \left| \br{1,E_0}\hat{a}_i\left(Y^{(\pm)}\right)^{-1}\hat{a}^{\dag}_j\kt{0,E_0} - \frac{1}{q^2} \br{1,\psi_0}e^{-\frac{1}{2}\tau^2H^2}\hat{a}_i(Y^{(\pm)})^{-1}\hat{a}^{\dag}_je^{-\frac{1}{2}\tau^2H^2}\kt{0,\psi_0} \right| \label{eq:gf-1}\\
    &\hspace{-17em}+ \frac{1}{q^2} \left| \br{1,\psi_0}e^{-\frac{1}{2}\tau^2H^2}\hat{a}_i(Y^{(\pm)})^{-1}\hat{a}^{\dag}_je^{-\frac{1}{2}\tau^2H^2}\kt{0,\psi_0} - \br{1,\psi_0}f_{GS}(\tilde{\varepsilon}_1,H)^{\dag}\hat{a}_i(Y^{(\pm)})^{-1}\hat{a}^{\dag}_jf_{GS}(\tilde{\varepsilon}_1,H)\kt{0,\psi_0} \right| \label{eq:gf-2} \\
    &\hspace{-17em}+\frac{1}{q^2} \left| \br{1,\psi_0}f_{GS}(\tilde{\varepsilon}_1,H)^{\dag}\hat{a}_i(Y^{(\pm)})^{-1}\hat{a}^{\dag}_jf_{GS}(\tilde{\varepsilon}_1,H)\kt{0,\psi_0} - \br{1,\psi_0}f_{GS}(\tilde{\varepsilon}_1,H)^{\dag}\hat{a}_i f_{inv}(\tilde{\varepsilon}_2,Y^{(\pm)}) \hat{a}^{\dag}_jf_{GS}(\tilde{\varepsilon}_1,H)\kt{0,\psi_0} \right| \label{eq:gf-3} \\
    &\hspace{-17em}+ \frac{1}{q^2} \left| \br{1,\psi_0}f_{GS}(\tilde{\varepsilon}_1,H)^{\dag}\hat{a}_i f_{inv}(\tilde{\varepsilon}_2,Y^{(\pm)}) \hat{a}^{\dag}_jf_{GS}(\tilde{\varepsilon}_1,H)\kt{0,\psi_0} - \overline{z}^{(M)} \right| \label{eq:gf-4} \\
    &\hspace{-17em}+ \left|\frac{1}{q^2} - \frac{1}{Q^{(M_q)}} \right| \cdot \left|\overline{z}^{(M)} \right| \label{eq:gf-5}\,,
\end{align}
where in the above we have used a chain of triangle inequalities. We now specify conditions so that the right hand side of the above is $\OC(\varepsilon)$. Corollary \ref{cor:GS-projector} bounds \eqref{eq:gf-1} by $2\varepsilon$ with the choice $\tau = \frac{1}{\Delta}\sqrt{2 \log\frac{\|\hat{a}_iY^{(\pm){-1}}\hat{a}^{\dag}_j\|}{\varepsilon\gamma}} \leq \frac{1}{\Delta}\sqrt{2 \log\frac{\|Y^{(\pm){-1}}\|}{\varepsilon\gamma}} = \frac{1}{\Delta}\sqrt{2 \log\frac{\|\Gamma^{(\pm){-1}}\|}{\varepsilon\gamma}}$. The term in \eqref{eq:gf-2} is bounded by $24\varepsilon$ by Eq.~\eqref{eq:gf-gs-overlap} and Lemma \ref{lem:appdx-expectation-overlap} under the condition that each $V_i$ appearing in $f_{GS}(\tilde{\varepsilon}_1,H)$ has non-Clifford gate depth at most $2\lambda_H^2 \tau^2 \log\left(\frac{2\|\Gamma^{(\pm){-1}}\|}{\varepsilon q^2} \right)$, which ensures that $\left\|e^{-\frac{1}{2} \tau^2 {H}^2} - f_{GS}(\tilde{\varepsilon}_1,H) \right\| \leq \frac{\varepsilon q^2}{\|\Gamma^{(\pm){-1}}\|}$.  The term in \eqref{eq:gf-3} can be bounded with H\"older's inequality as
\begin{align}
    \eqref{eq:gf-3} &\leq \left\|Y^{(\pm)}- f_{inv}(\tilde{\varepsilon}_2,Y^{(\pm)})\right\|_{\infty} \cdot \frac{1}{q^2} \left\| \hat{a}^{\dag}_jf_{GS}(\tilde{\varepsilon}_1,H) \dyad{1,\psi_0}{0,\psi_0} f_{GS}(\tilde{\varepsilon}_1,H)^{\dag}\hat{a}_i\right\|_{1} \\
    &= \left\|Y^{(\pm)}- f_{inv}(\tilde{\varepsilon}_2,Y^{(\pm)})\right\|_{\infty} \cdot \frac{1}{q^2} \left\| f_{GS}(\tilde{\varepsilon}_1,H) \dyad{1,\psi_0}{0,\psi_0} f_{GS}(\tilde{\varepsilon}_1,H)^{\dag}\right\|_{1}\\
    &\leq \left\|Y^{(\pm)}- f_{inv}(\tilde{\varepsilon}_2,Y^{(\pm)})\right\|_{\infty} \cdot \frac{1}{q^2} \left\| f_{GS}(\tilde{\varepsilon},H)\dya{\psi_0}\right\|_{2}^2\\
    &\leq \left\|Y^{(\pm)}- f_{inv}(\tilde{\varepsilon}_2,Y^{(\pm)})\right\|_{\infty} \cdot \frac{1}{q^2} \left(\left\| e^{-\frac{1}{2} \tau^2 {H}^2}\dya{\psi_0}\right\|_{2} + \left\|\left(e^{-\frac{1}{2} \tau^2 {H}^2} - f_{GS}(\tilde{\varepsilon}_1,H) \right)\dya{\psi_0} \right\|_2\right)^2 \\
    &\leq \left\|Y^{(\pm)}- f_{inv}(\tilde{\varepsilon}_2,Y^{(\pm)})\right\|_{\infty} \cdot \frac{1}{q^2} \left(\left\| e^{-\frac{1}{2} \tau^2 {H}^2}\kt{\psi_0}\right\|_{2} + \left\|\left(e^{-\frac{1}{2} \tau^2 {H}^2} - f_{GS}(\tilde{\varepsilon}_1,H) \right) \right\|_{\infty}\right)^2 \\
    &\leq \left\|Y^{(\pm)}- f_{inv}(\tilde{\varepsilon}_2,Y^{(\pm)})\right\|_{\infty} \left(1 +  2\frac{\varepsilon q}{\|\Gamma^{(\pm){-1}}\|} + \frac{\varepsilon^2q^2}{\|\Gamma^{(\pm){-1}}\|^2} \right)\,,
\end{align}
where in the equality we have used the unitary invariance of the Schatten norm, noting that $\hat{a}^{\dag}_j$ and $\hat{a}_i$ can be expressed as probabilistic combinations of two unitaries. In the second inequality we have again used H\"older's inequality, and the third inequality is due to the triangle inequality. Thus, assuming $\varepsilon\leq \frac{\|\Gamma^{(\pm){-1}}\|}{q}$ (of which we expect the right-hand-side to be much larger than 1), the term in \eqref{eq:gf-3} can be bounded by $4\varepsilon$ by setting $\tilde{\varepsilon}_2 = \varepsilon$ which sets $R_{inv} = \Theta\left( \|\Gamma^{(\pm)-1}\| \sqrt{\log \left(\frac{\|\Gamma^{(\pm)-1}\|}{\varepsilon} \right)} \right)$ and the non-Clifford gate depth of each $U_i$ as being at most $\Theta\left(\lambda_{Y^{(\pm)}}^2\|\Gamma^{(\pm)-1}\|^2 \log^2 \left(\frac{\|\Gamma^{(\pm)-1}\|}{\varepsilon} \right) \right)$. The term in \eqref{eq:gf-4} can be bounded by $\varepsilon$ with probability $(1-\delta)$ by asking for number of shots
\begin{equation}
    M\geq 2\log\left(\frac{2}{\delta}\right)\frac{R^2_{inv}R^4_{GS}}{\varepsilon q^4}= \OC\left(\log\left(\frac{2}{\delta}\right)\frac{\|\Gamma^{(\pm)-1}\|^2}{\varepsilon^2\gamma^4} \log \left(\frac{\|\Gamma^{(\pm)-1}\|}{\varepsilon} \right) \right)\,,
\end{equation}
where we have used the fact that $\frac{1}{q}\leq \frac{2}{\gamma}$. Finally, the term in \eqref{eq:gf-5} can be bounded by $\varepsilon$ with probability at least $(1-\delta)$ by using number of shots
\begin{equation}
    M_q\geq 2\log\left(\frac{2}{\delta}\right)\frac{\|Y^{(\pm)-1}\|^2 R^4_{GS}}{\varepsilon q^4}= \OC\left(\log\left(\frac{2}{\delta}\right)\frac{\|\Gamma^{(\pm)-1}\|^2}{\varepsilon^2\gamma^4} \right)\,,
\end{equation}
for the normalization subroutine. 

Finally, we note the total gate complexity consists of the sum of the gate complexities for $U_i$ and $V_i$ respectively. From the above discussion, we see that the non-Clifford gate complexity of $V_i$ is $\OC\left(\frac{\lambda_H^2}{\Delta^2} \log\left(\frac{2\|\Gamma^{(\pm){-1}}\|}{\varepsilon \gamma^2} \right)\right)$ and the non-Clifford gate gate complexity of $U_i$ is $\OC\left( (|\hbar\omega \pm E_0|+\eta+ \lambda_{H})^2\|\Gamma^{(\pm)-1}\|^2 \log^2 \left(\frac{\|\Gamma^{(\pm)-1}\|}{\varepsilon} \right) \right)$. The non-Clifford gate complexity for the normalization subroutine as given by Proposition \ref{prop:norm-sampling} is $\OC\left(\frac{\lambda_H^2}{\Delta^2} \log\left(\frac{\|\Gamma^{(\pm){-1}}\|}{\varepsilon \gamma}\right) \log\left(\frac{1}{\varepsilon \gamma}\right) \right)$.
\end{proof}

We now briefly remark on the scenario if we would like to evaluate the Green's function but the exact ground state energy is not given. In this case, an $\varepsilon$-additive approximation to $E_0$ would yield an amplification of the additive error of the evaluation of the Green's functions by an additive term $\OC(\varepsilon\|\Gamma^{(\pm)-1}\| \|(\Gamma^{(\pm)}\pm\varepsilon\id)^{-1}\|)$. This can be seen by noting that 
\begin{align}
    \left\| \left(\Gamma^{(\pm)} \pm\varepsilon\id\right)^{-1} - \Gamma^{(\pm)-1} \right\| &= \left\| \left(\id - \Gamma^{(\pm)-1}\left(\Gamma^{(\pm)} \pm\varepsilon\id\right)\right) \left(\Gamma^{(\pm)} \pm\varepsilon\id\right)^{-1} \right\|\\
    &=\varepsilon\left\| \Gamma^{(\pm)-1} \left(\Gamma^{(\pm)} \pm\varepsilon\id\right)^{-1}  \right\|\,.
\end{align}

\subsection{Additional information for complexity comparison}

In this section we provide some additional information on state of the art quantum algorithms in the literature, to aid exposition of Tables \ref{table:linear-systems}, \ref{table:linear-systems2}, \ref{table:gs} and \ref{table:gibbs}.   

\subsubsection{Data access}\label{sec:appdx-data-access}

In our tables we quote two generic block encoding strategies. First, we use the result of Ref.~\cite{clader2022quantum} which gives an explicit block encoding of $A/\|A\|_F$ starting from a description of $A \in \mathbb{C}^{N \times N}$ in the computational basis. Throughout the manuscript we quote the minimal gate depth implementation of this construction, which is achieved using $\OC(N^2)$ qubits and $\OC(\log N)$ gate depth each call to the block encoding. Ref.~\cite{clader2022quantum} gives an alternative construction using $\OC(N)$ qubits but $\OC(N)$ gate depth each call. A second block encoding strategy we quote is to use $\texttt{SELECT}$ and $\texttt{PREPARE}$ oracles when $A$ is given in the Pauli basis $A=\sum_{\ell=1}^L a_{\ell} P_{\ell}$ with Pauli weight $\lambda:=\sum_{\ell=1}^L |a_{\ell}|$. In general this strategy works for any unitary decomposition where one presumes controlled versions of the unitary can be implemented in $\OC(1)$ gate depth. This strategy generically uses $\OC(\log L)$ qubits, and $\OC(L)$ gate depth each call to the block encoding. The block encoding has subnormalization $\lambda$.

We also conceptualize the ability to perform Hamiltonian simulation as a type of data access. In Tables \ref{table:gs} and \ref{table:gibbs} we quote two Hamiltonian simuation strategies; one which uses no additional ancillary qubits, and a second which is optimal in gate complexity. First, we quote $1^{st}$ order Trotter time evolution. For our purposes we again work in the Pauli access model, and the goal is to return an approximation of $\exp(iAt)$ for some real time $t$. $1^{st}$ order Trotter uses (in worst case) $O\left(L^3(\Lambda t)^2 / \epsilon\right)$ gates where we have defined $\Lambda:=\max_{\ell}a_{\ell}$. As we state in the main text, one can replace $\Lambda$ with a smaller quantity by exploiting commutator structure between terms \cite{childs2021theory}. The second Hamiltonian simulation strategy we quote is the so-called qubitization approach of Ref.~\cite{low2019hamiltonian}. Here one constructs a block encoding using $\texttt{SELECT}$ and $\texttt{PREPARE}$ again and makes calls to this block encoding. The algorithm uses $\OC(\log L)$ ancillary qubits, has query complexity $\OC(\lambda t + \frac{\log(1/\varepsilon)}{\log \log(1/\varepsilon)})$, where each query uses $\OC(L)$ gates as stated above. Thus, each call to this time evolution as an oracle requires $\OC(\lambda L t + \frac{\log(1/\varepsilon)}{\log \log(1/\varepsilon)})$ gates.

Finally, in Table 1 for the HHL algorithm \cite{harrow2009quantum} we also require a Hamiltonian simulation subroutine when starting from a Frobenius norm block encoding. Again using Ref.~\cite{low2019hamiltonian}, an approximation of $\exp (iAt)$ can be obtained with $\OC(\mu t + \frac{\log(1/\varepsilon)}{\log \log(1/\varepsilon)})$ calls to the block encoding, where $\mu$ is the block encoding subnormalization. 

\subsubsection{Linear systems}

The HHL algorithm \cite{harrow2009quantum} consists of three key steps. First, conditional Hamiltonian simulation of maximum time $\OC(\|A^{-1}\|/\varepsilon)$ is required to perform phase estimation to estimate each inverse eigenvalue $\lambda_j^{-1}$ to additive error $\OC(\frac{1}{\lambda_j}\frac{\varepsilon}{\|A^{-1}\|}) = \OC(\varepsilon)$. Second, after rotating by angle $\textrm{arcsin}(1/\lambda_j\|A^{-1}\|)$ conditioned on $\lambda_j$ and undoing phase estimating we have the state $\ket{\lambda_j}\ket{0} \rightarrow \ket{\lambda_j} \left(\frac{1}{\lambda_j \|A^{-1}\|}\ket{0} + \sqrt{1-\frac{1}{\lambda_j^2 \|A^{-1}\|^2}} \right)$. Finally, the desired state can thus be obtained with $\OC(\kappa)$ amplitude amplification steps. Putting this all together, starting with a Frobenius norm block encoding (which requires $\OC(\log N)$ gates per call), approximately preparing the state $\ket{A^{-1}b}$ requires $\OC(\|A\|_F \|A^{-1}\|\kappa \log N /\varepsilon)$ = $\OC(\kappa_F\kappa \log N /\varepsilon)$ gate depth to implement.

In the quantum linear systems algorithm of Ref.~\cite{costa2021optimal} the dominant part is a filtering step with complexity $\OC(\frac{1}{\Delta}\log(1/\varepsilon))$, where $\Delta$ is a lower bound on the gap of the adiabatic Hamiltonian which interpolates between Hamiltonians 
\begin{align}
    H_0:=\left(\begin{array}{cc}
0 & Q_b \\
Q_b & 0
\end{array}\right)\,, \quad
H_1:=\left(\begin{array}{cc}
0 & \hat{A} Q_b \\
Q_b \hat{A} & 0
\end{array}\right)\,,
\end{align}
for some encodable matrix $\hat{A}$, and where $Q_b=\id-\dya{b}$. It is shown in Appendix A of Ref.~\cite{an2022quantum} that this has gap lower bounded by $\|\hat{A}^{-1}\|$. In Ref.~\cite{costa2021optimal} the authors call for a block encoding construction, thus we block encode an adiabatic Hamiltonian with $\hat{A}=A/\mu$, leading to algorithm complexity $\mu \|A^{-1}\|\log(1/\varepsilon)$ where $\mu$ is the block encoding subnormalization.

\subsubsection{Ground states and Gibbs states}

In Tables \ref{table:gs} and \ref{table:gibbs} all the algorithms we compare with either use Hamiltonian simulation or block encoding oracles. Thus, the quoted gate complexities in our tables can be extracted by inspecting the query complexity listed in each relevant work and multiplying by the corresponding gate complexity for oracle access detailed in Section \ref{sec:appdx-data-access}.


\end{document}